\documentclass[aps,prx,twocolumn,superscriptaddress,showpacs,longbibliography]{revtex4-1}
\usepackage{float}
\usepackage{graphicx}  
\usepackage{svg}
\usepackage{color}
\usepackage{bm}        
\usepackage{braket}
\usepackage{amssymb}   
\usepackage{epstopdf}
\usepackage{xfrac}
\usepackage{cancel}
\usepackage{soul}
\usepackage{amsmath}
\usepackage{amsthm}
\usepackage{amsfonts}
\usepackage{dsfont}
\usepackage{bbm}
\definecolor{darkblue}{rgb}{0.1,0.2,0.6}
\definecolor{darkred}{rgb}{0.8,0.1,0.2}
\definecolor{darkgreen}{rgb}{0.31,0.62,0.24}
\definecolor{OliveGreen}{cmyk}{0.64, 0, 0.95, 0.40}
\usepackage[colorlinks,citecolor=darkblue,linkcolor=darkblue,urlcolor=darkblue]{hyperref} 
\usepackage{enumerate}
\usepackage{setspace}
\usepackage{url}  
\usepackage{mathrsfs}
\usepackage{mathtools} 

\usepackage{manfnt} 
\newcommand*\cube{\mbox{\mancube}}

\usepackage[percent]{overpic}
\usepackage{verbatim}

\usepackage[most]{tcolorbox}
\usepackage[symbol]{footmisc}

\newcommand{\Z}{\mathbb{Z}}

\newcommand{\Tr}{\text{Tr}}
\newcommand{\tr}{\text{Tr}}

\newcommand{\tcO}{\tilde{\cal O}}
\newcommand{\cO}{{\cal O}}

\newcommand{\bp}{\boldsymbol{p}}
\newcommand{\btp}{\boldsymbol{\tilde{p}}}
\def\i{\textbf{i}}

\usepackage{algorithmicx}
\usepackage{algorithm}
\usepackage[noend]{algpseudocode}

\newcounter{protocol}
\makeatletter

\newtheorem{theorem}{Theorem}
\newtheorem{corollary}{Corollary}
\newtheorem{lemma}{Lemma}
\newcommand{\SM}[1]{{\rm SM}_{#1}}

\newcommand{\tK}{\tilde{K}}
\newcommand{\ttau}{\tilde{\tau}}
\newcommand{\cV}{{\cal V}}
\newcommand{\cVV}{{\cal V}\oplus{\cal V}}
\newcommand{\vx}{{\cal V}_x}
\newcommand{\vz}{{\cal V}_z}
\newcommand{\vs}{{\cal V}_s}
\newcommand{\vsdp}{\left({\cal V}_s^*\right)^\perp}
\newcommand{\rb}{\rho_{\rm b}}
\newcommand{\rp}{\rho_{\rm p}}

\newcommand{\kket}[1]{|#1\rangle\rangle}
\newcommand{\bbra}[1]{\langle\langle #1|}
\newcommand{\sfa}{\mathsf{a}}
\newcommand{\sfb}{\mathsf{b}}
\newcommand{\sfe}{\mathsf{e}}
\newcommand{\sff}{\mathsf{f}}
\newcommand{\lx}{\lambda_x}
\newcommand{\ly}{\lambda_y}
\newcommand{\lz}{\lambda_z}
\newcommand{\tlx}{\tilde{\lambda}_x}
\newcommand{\tly}{\tilde{\lambda}_y}
\newcommand{\tlz}{\tilde{\lambda}_z}
\newcommand{\ninfeq}{\stackrel{N\to \infty}{=\joinrel=}}

\tcbset{
    frame code={}
    center title,
    left=0pt,
    right=0pt,
    top=0pt,
    bottom=0pt,
    colback=gray!20,
    colframe=white,
    width=0.5\dimexpr\textwidth\relax,
    enlarge left by=0mm,
    boxsep=5pt,
    arc=0pt,outer arc=0pt,
    }

\begin{document}

\title{Tapestry of dualities in decohered quantum error correction codes}

\author{Kaixiang Su}
\thanks{These two authors contributed equally.}
\affiliation{Department of Physics, University of California,
Santa Barbara, CA 93106}

\author{Zhou Yang}
\thanks{These two authors contributed equally.}
\affiliation{Department of Physics, Cornell University, Ithaca, New York 14853, USA}

\author{Chao-Ming Jian}
\affiliation{Department of Physics, Cornell University, Ithaca, New York 14853, USA}

\date{\today}

\begin{abstract}
Quantum error correction (QEC) codes protect quantum information from errors due to decoherence. Many of them also serve as prototypical models for exotic topological quantum matters. Investigating the behavior of the QEC codes under decoherence sheds light on not only the codes' robustness against errors but also new out-of-equilibrium quantum phases driven by decoherence. The phase transitions, including the error threshold, of the decohered QEC codes can be probed by the systems' R\'enyi entropies $S_R$ with different R\'enyi indices $R$.
In this paper, we study the general construction of the statistical models that characterize the R\'enyi entropies of QEC codes decohered by Pauli noise. We show that these statistical models can be organized into a ``tapestry" woven by rich duality relations among them. For Calderbank-Shor-Steane (CSS) codes with bit-flip and phase-flip errors, we show that each R\'enyi entropy is captured by a pair of dual statistical models with randomness. For $R=2,3,\infty$, there are additional dualities that map between the two error types, relating the critical bit-flip and phase-flip error rates of the decoherence-induced phase transitions in the CSS codes. For CSS codes with an ``$em$ symmetry" between the $X$-type and the $Z$-type stabilizers, the dualities with $R=2,3,\infty$ become self-dualities with {\it super-universal} self-dual error rates. These self-dualities strongly constrain the phase transitions of the code signaled by $S_{R=2,3,\infty}$. For general stabilizer codes decohered by generic Pauli noise, we also construct the statistical models that characterize the systems' entropies and obtain general duality relations between Pauli noise with different error rates. 
\end{abstract}

\maketitle


\section{Introduction}

Quantum error correction (QEC) is a fundamental scheme in quantum computation to protect quantum information from decoherence caused by errors and noise \cite{QEC_RMP2015,GottesmanReview2009}. A general QEC code achieves this protection by encoding the quantum information into the ``logical qubits" formed by many-body states of physical qubits with entanglement patterns resilient against errors. From the quantum matter perspective, a large class of QEC codes, epitomized by the toric code \cite{Kitaev2003}, can be viewed as systems with topological orders, where the long-range-entangled and topologically-protected ground states serve as the logical qubits. Recently, there has been tremendous progress in the experimental realizations of the toric code and related QEC codes on noisy intermediate-scale quantum platforms \cite{Satzinger_2021,GoogleAISurfaceCode,GoogleAInonAblebian,RydbergSL}. The preparation of a QEC code on a noisy quantum platform is, in general, an out-of-equilibrium process inevitably subject to decoherence. Instead of a Gibbs ensemble, the prepared state is an ``error-corrupted" mixed state, i.e., a classical mixture of different error patterns on top of the code's logical states. Understanding such error-corrupted mixed states brings insights into the robustness of QEC codes against errors and the effect of decoherence on topological orders.

An important metric for the performance of a QEC code is the error threshold, which pertains to the ``decodability" of the logical qubits from the error-corrupted mixed state. The seminal work Ref. \onlinecite{dennis2002topological} pointed out that the decodability of the 2D toric code with errors is captured by the 2D classical random-bond Ising model, and the error threshold is identified as a continuous phase transition. This method that maps the QEC code's decodability to classical statistical models has been widely generalized to 
other stabilizer codes (see Ref. \onlinecite{Matin-DelgadoColorCodes,Martin-Delgado_FractonThreshold,Preskill_3dToricCodeThreshold,BombinTopoResilience2012,Flammia2021} for examples).

Recent works Ref. \onlinecite{FanBaoTopoMemory,BaoFanError_Field_Double,LeeJianXu2023} provide a different perspective on the error threshold of the toric code by viewing it as a singularity intrinsic to the error-corrupted mixed state $\rho$. This singularity manifests a decoherence-induced phase transition (DIPT) of the system's total von Neumann entropy $S_1 = -\Tr \rho \log \rho$ as a function of the decoherence strength, i.e. the error rate. Moreover, the error threshold naturally belongs to the rich family of DIPTs of the R\'enyi entropies $S_R = \frac{1}{1-R}\log (\Tr \rho^R)$ (for $R=2,3,...$ and $R\rightarrow 1$) in the decohered toric code \cite{BaoFanError_Field_Double,FanBaoTopoMemory,LeeJianXu2023}. For a general QEC code under decoherence, the DIPTs of the R\'enyi entropies probe the singular changes in the entanglement pattern 
in the error-corrupted mixed state. In light of the close relation between QEC codes and topological orders in systems free of decoherence, studying the behavior of the error-corrupted mixed state in decohered QECs and the DIPTs therein provides an interesting path toward understanding topological orders and other decoherence-induced exotic phases in beyond-Gibbs-ensemble mixed states.

In this paper, we focus on developing a general framework to study the error-corrupted mixed states in stabilizer codes decohered by Pauli noise, including bit-flip and phase-flip errors in particular. We systematically construct the statistical models arising from the R\'enyi entropies $S_R$ for a general $R$ and identify a rich set of duality relations that weave these statistical models into a ``tapestry of dualities". Our discussion contains three major parts, focusing on three classes of translation-invariant stabilizer codes: (1) general Calderbank-Shor-Steane (CSS) codes, (2) CSS codes with an ``$em$ symmetry" between the $X$-type and the $Z$-type stabilizers, and (3) general stabilizer codes.

We first investigate general CSS codes decohered by bit-flip and phase-flip errors. The tapestry of dualities for this class of decohered codes is shown in Fig. \ref{fig:duality1}. The two error types cause 
independent decoherence in a CSS code and are, hence, studied separately. The behavior of the decohered CSS code is intimately related to a pair of dual statistical models, named SM$_1$ and SM$_2$, which are obtained from ungauging the CSS code. Under bit-flip errors, the R\'enyi entropy $S_R$ of the error-corrupted mixed state is described by a pair of dual statistical models with randomness: (1) 
$R$-replica SM$_1$ with real random couplings (rRC) and (2) $R$-replica SM$_2$ with imaginary random couplings (iRC). These two statistical models are dual under a high-low-temperature (HLT) duality. Similarly, under phase-flip errors, $S_R$ is described by another pair of random statistical models: (1) $R$-replica SM$_1$ with iRC and (2) $R$-replica SM$_2$ with rRC, which are again HLT dual to each other. The patterns of randomness in these statistical models are interchanged as we switch from bit-flip errors to phase-flip errors. Additionally, for $R=2,3,\infty$, we discover extra duality relations that map between (the random statistical models associated with) bit-flip and phase-flip errors. We call these dualities the bit-phase-decoherence (BPD) dualities. The random statistical models and the dualities provide powerful tools to study the phases and the DIPTs in the decohered CSS codes. We propose a conjecture on the monotonicity of the DIPTs' critical error rates as functions of the R\'enyi index $R$, and discuss the alternative interpretation of the DIPT with $R=2$ as a quantum phase transition in the doubled Hilbert space. We also discuss two concrete examples, the decohered 3D toric code \cite{kitaev1997quantum} and the decohered X-cube model \cite{Vijay_2016}, demonstrating the abovementioned general results.

When the CSS codes are endowed with additional symmetries, the duality structure under decoherence can be enriched. We analyze the case when the CSS codes have a symmetry that exchanges the $X$-type and $Z$-type stabilizers, which we dub ``electric-magnetic symmetry" ($em$ symmetry). Under the bit-flip and phase-flip errors, the tapestry of dualities of the $em$-symmetric CSS code is shown in Fig. \ref{fig:self-duality}, which is effectively the tapestry shown in Fig. \ref{fig:duality1} ``folded in half" by the symmetry. The $em$ symmetry interchanges the bit-flip and phase-flip errors. The original pair of dual statistical models SM$_1$ and SM$_2$ become the same model, call it SM, with a self-duality. Under either bit-flip or phase-flip errors, the R\'enyi entropies $S_R$ are captured by a pair of statistical models with randomness: (1) $R$-replica SM with rRC and (2) $R$-replica SM with iRC. They are related by the HLT duality for any $R$. The BPD dualities for $R=2,3,\infty$ (combined with the $em$ symmetry) become self-dualities with {\it super-universal} self-dual error rates shared by all $em$-symmetric CSS codes in different dimensions.  For a given $em$-symmetric CSS code (and a fixed $R=2,3$ or $\infty$), if there is a unique DIPT of $S_R$ as the error rate varies, the critical error rate must coincide with the super-universal self-dual values. We demonstrate our general results in the concrete settings of the decohered 2D toric code \cite{kitaev1997quantum}, and the decohered Haah's code in 3D \cite{haah2011local}. 

For general stabilizer codes, bit-flip and phase-flip errors are no longer special. We therefore consider the decoherence by general Pauli noise. The general decohered stabilizer codes are tied to a single statistical model SM, which is self-dual under an HLT duality. The R\'enyi entropy $S_R$ of the decohered code is captured by both (1) the $R$-replica SM with rRC and (2) the $R$-replica SM with iRC, two random statistical models related by an HLT duality. Moreover, we find extra dualities between different error rates for $R=2, \infty$ and call them ``general-Pauli-noise dualities" (GPN dualities). The GPN dualities are the generalizations of the BPD dualities of decohered CSS codes. We obtain the surface of super-universal self-dual error rates under the GPN dualities.

The rest of the paper is organized as follows. Sec. \ref{sec:CSS} focuses on the general CSS codes decohered by bit-flip and phase-flip errors. Sec. \ref{sec:em-symmetric CSS} discusses the decohered CSS codes with $em$ symmetry. Sec. \ref{sec:GPN} analyzes general stabilizer codes with general Pauli noise. We present our conclusion and outlook in Sec. \ref{sec:conclusions}.

\section{CSS code under decoherence}
\label{sec:CSS}
In this section, we will first review the basics of general CSS codes and introduce the decoherence model that describes bit-flip and phase-flip errors. Next, we provide the general construction of the statistical models, the $R$-replica $\SM1$ and $\SM2$ with randomness, that describe the R\'enyi entropies $S_R$ of the error-corrupted mixed states of the decohered CSS code. We will address how the $R\rightarrow 1$ limit recovers the disordered statistical model that captures the code's error threshold for decodability. Then, we present the HLT and BPD dualities between these statistical models and the tapestry of dualities they form. The DIPTs are signaled by the singular behavior of the R\'enyi entropies $S_R$ and, hence, are identified with the phase transitions in the random statistical models. For $R=2$, by generalizing Ref. \onlinecite{BaoFanError_Field_Double,LeeYouXu2022}, we give an alternative interpretation of the DIPT as a quantum phase transition in the doubled Hilbert space. We also propose a conjecture on the monotonicity of the DIPTs' critical error rates as a function of the R\'enyi index $R$. Finally, we study the decohered 3D toric code and the decohered X-cube model as concrete examples.

\subsection{Introduction to CSS codes and error models for decoherence}
\label{sec:CSScode_ErrorModel}
First, we briefly review the general description of a CSS code. Consider a set of $N$ qubits (labeled by $\mu$) arranged into a lattice. A general CSS code is a stabilizer code with two types of stabilizers: $X$-type and $Z$-type. Each $X$-type generator, denoted as $A_i[X]$, is a product of only Pauli-$X$ operators $X_\mu$, while each $Z$-type generator, denoted as $B_j[Z]$, is a product of only Pauli-$Z$ operators $Z_\mu$. All the stabilizers commute with each other:
\begin{align}
    \big[A_i[X],B_j[Z]\big] = 0,~~~~\forall~i,j.
    \label{eq:ABcommutation}
\end{align}
The logical subspace of the CSS code is the ground-state Hilbert space of the CSS code Hamiltonian
\begin{align}
    H_{\rm css} = - \sum_i A_i[X] - \sum_j B_j[Z].
    \label{eq:CSS_Ham}
\end{align}
Each logical state $|\Omega\rangle$ satisfies 
\begin{align}
    A_i[X]\ket{\Omega}= B_j[Z]\ket{\Omega} = \ket{\Omega}, \hspace{.5cm} \forall ~i,j.
\end{align}

For this work, we work with CSS codes with the following general properties. We assume that $H_{\rm css}$ is translation invariant (modulo boundary effects of the lattice if they exist) and all the terms in $H_{\rm css}$ are local. The number of $A_i[X]$ terms does not have to equal the number of $B_j[Z]$ terms. But we assume that there is no local logical operator and the rate of the code $N_{\rm c}/N$ vanishes in the large system limit. Here, $N_{\rm c}$ denotes the total number of logical qubits. These properties are rather general. They are broadly satisfied by commonly discussed topological CSS codes, including the toric code \cite{kitaev1997quantum}, the surface code \cite{bravyi1998quantum, dennis2002topological}, color codes \cite{bombin2006topological}, the $X$-cube model \cite{Vijay_2016}, Haah's code \cite{haah2011local}, etc.

Next, we introduce the error models for the decoherence that we will study in this section. We consider an error model with bit-flip errors and phase-flip errors. For a single site $\mu$, when a bit-flip (phase-flip) error occurs, the system's state is acted on by $X_\mu$ ($Z_\mu$). The classical probabilities $p_x$ and $p_z$ for the appearance of these errors on a single qubit are called the error rates. The decoherence caused by these two types of errors is described by the following quantum channels respectively:
\begin{align}
    & {\cal N}_{x,\mu} (\rho)= (1-p_x) \rho + p_x X_\mu \rho X_\mu, \nonumber \\ 
    & {\cal N}_{z,\mu} (\rho)= (1-p_z) \rho + p_z Z_\mu \rho Z_\mu.
\end{align}
When all the physical qubits are subject to such decoherence, a logical state $\rho_0 =|\Omega\rangle \langle \Omega |$, initially pure, becomes an error-corrupted mixed-state:
\begin{align}
    \rho_0 \rightarrow {\cal N}_x \circ {\cal N}_z (\rho_0),
\end{align} 
a mixture of the logical state dressed by all error patterns. Here, we've defined the total decoherence channels ${\cal N}_{x/z} \equiv \otimes_\mu {\cal N}_{x/z,\mu}$. It suffices to consider the range of error rates $0\leq p_x, p_z\leq \frac{1}{2}$ because the system experiences the maximal amount of bit-flip (phase-flip) decoherence when $p_{x} =1/2$ ($p_{z} =1/2$).

The bit-flip errors only affect the $Z$-type stabilizers, while the phase-flip errors only influence the $X$-type stabilizers. Hence, the two error types cause independent decoherence effects (more details in the next subsection). Without loss of generality, it suffices to independently study the two types of mixed states $\rho_{\rm b}$ and $\rho_{\rm p}$ corrupted only by the bit-flip and phase-flip errors respectively: 
\begin{align}
    \rho_{\rm b} = {\cal N}_x (\rho_0),~~\rho_{\rm p} = {\cal N}_z (\rho_0).
\end{align}
We can write these mixed states as a summation over error chains. For example, 
\begin{align}
     \rho_{\rm b} = \sum_{E} P_x(E) X_E\rho_0 X_E,
     \label{eq:rho_b_Expansion}
\end{align}
where $\sum_{E}$ is a summation over all possible error chain $E$. An error chain $E$ is the set of qubits where the errors occur and $X_E \equiv \prod_{\mu \in E} X_\mu$. $ P_x(E)$ is the probability for the error chain $E$ to appear: 
\begin{align}
    P_x(E) = p_x^{|E|} (1-p_x)^{N-|E|}. \label{eq:decohprobx}
\end{align}
Here, $|E|$ denotes the number of qubits in $E$. Similarly, we can write 
\begin{align}
     \rho_{\rm p} = \sum_{E} P_z(E) Z_E\rho_0 Z_E,
     \label{eq:rho_p_Expansion}
\end{align}
with $P_z(E) = p_z^{|E|} (1-p_z)^{N-|E|}$ and $Z_E \equiv \prod_{\mu \in E} Z_\mu$. In the next subsection, we develop the statistical models that capture the behavior of the error-corrupted mixed states $\rho_{\rm p/b}$. 

As a preparation for the subsequent analysis, it is helpful to introduce an $N$-dimensional ${\mathbb Z}_2$ vector space ${\cal V}$ for the lattice of $N$ qubits. Any subset $W$ of qubits can be represented as an $N$-component vector $W\in {\cal V}$ whose $\mu$th entry is 1 if the $\mu$th site belongs to the subset $W$ and 0 if not. Here and for the rest of the paper, we use the same notation for a subset of qubits and its corresponding vector. For a subset $W$, its cardinality $|W|$ counts the number of nonzero entries in the vector $W$. With two vectors $W_1, W_2\in {\cal V}$, $W_1 + W_2$ and $W_1 \cdot W_2$ represent the addition and the dot product in the ${\mathbb Z}_2$ vector space ${\cal V}$. In terms of subsets, $W_1+W_2$ represents $W_1\cup W_2 - W_1 \cap W_2$ and $W_1\cdot W_2$ is the parity of $|W_1 \cap W_2|$.

For each stabilizer $A_i[X]$, we denote the subset of qubits it includes as ${\mathsf a}_i$. Similarly, every $B_j[Z]$ has a corresponding vector ${\mathsf b}_j$. The commutation relation Eq. \eqref{eq:ABcommutation} is equivalent to
\begin{align}
    {\mathsf a}_i \cdot {\mathsf b}_j = 0~~~~\forall i,j.
    \label{eq:a_dot_b}
\end{align}
The collection of vectors $\{{\mathsf a}_i\}$ ($\{{\mathsf b}_j\}$) span a sub vector space denoted as ${\cal V}_x$ (${\cal V}_z$). The relation Eq. \eqref{eq:a_dot_b} implies that ${\cal V}_x$ belongs to the orthogonal subspace of ${\cal V}_z$, namely ${\cal V}_x \subset {\cal V}_z^\perp$. Similarly, we have ${\cal V}_z \subset {\cal V}_x^\perp$. A logical $X$-operator ($Z$-operator) of the CSS code must be a product of $X_\mu$ ($Z_\mu$) on a subset of qubits whose corresponding vector belongs to the quotient space ${\cal V}_z^\perp/{\cal V}_x$ (${\cal V}_x^\perp/{\cal V}_z$). The number of logical qubits is then given by $N_{\rm c} = {\rm dim}\left({\cal V}_z^\perp/{\cal V}_x\right) = {\rm dim}\left({\cal V}_x^\perp/{\cal V}_z\right)$. Our assumption that the code rate $N_c /N \rightarrow 0$ in the infinite system limit conceptually implies ${\cal V}_z^\perp$ (${\cal V}_x^\perp$) is almost the same as ${\cal V}_x$ (${\cal V}_z$), which is important for the dualities we will discuss later. The assumption that the CSS code is free of local logical operators implies that the quotient spaces ${\cal V}_z^\perp/{\cal V}_x$ and ${\cal V}_x^\perp/{\cal V}_z$ only contain vectors that represent non-local sets that involve infinitely many qubits in the infinite system limit. 

In addition to characterizing the stabilizers, it proves convenient to treat any error chain $E$ as a vector in ${\cal V}$ when analyzing the effects of decoherence. For example, the error chains $E, E'$ follow the multiplication and commutation relations,
\begin{align}
    X_{E} X_{E'} = X_{E+E'}&, ~~~Z_{E} Z_{E'} = Z_{E+E'}, \nonumber \\
    X_{E'} Z_{E} = &(-1)^{E \cdot E'} Z_{E} X_{E'}.
\end{align}
For any subset $E$, if $E\in \vx$, $X_E$ can be written as a product of the $A_i[X]$ stabilizers. Similarly, $Z_E$ is product of the $B_j[Z]$ stabilizers if and only if $E\in \vz$. The summation $\sum_E$ in Eq. \eqref{eq:rho_b_Expansion} and \eqref{eq:rho_p_Expansion} can be viewed as a summation $\sum_{E\in {\cal V}}$ over all $\mathbb Z_2$ vectors in $\cal V$.

\subsection{Statistical models for decohered CSS codes} \label{sec:sm_construction}
\begin{figure*}[ht!]
    \centering
    \includegraphics[width=.9\linewidth]{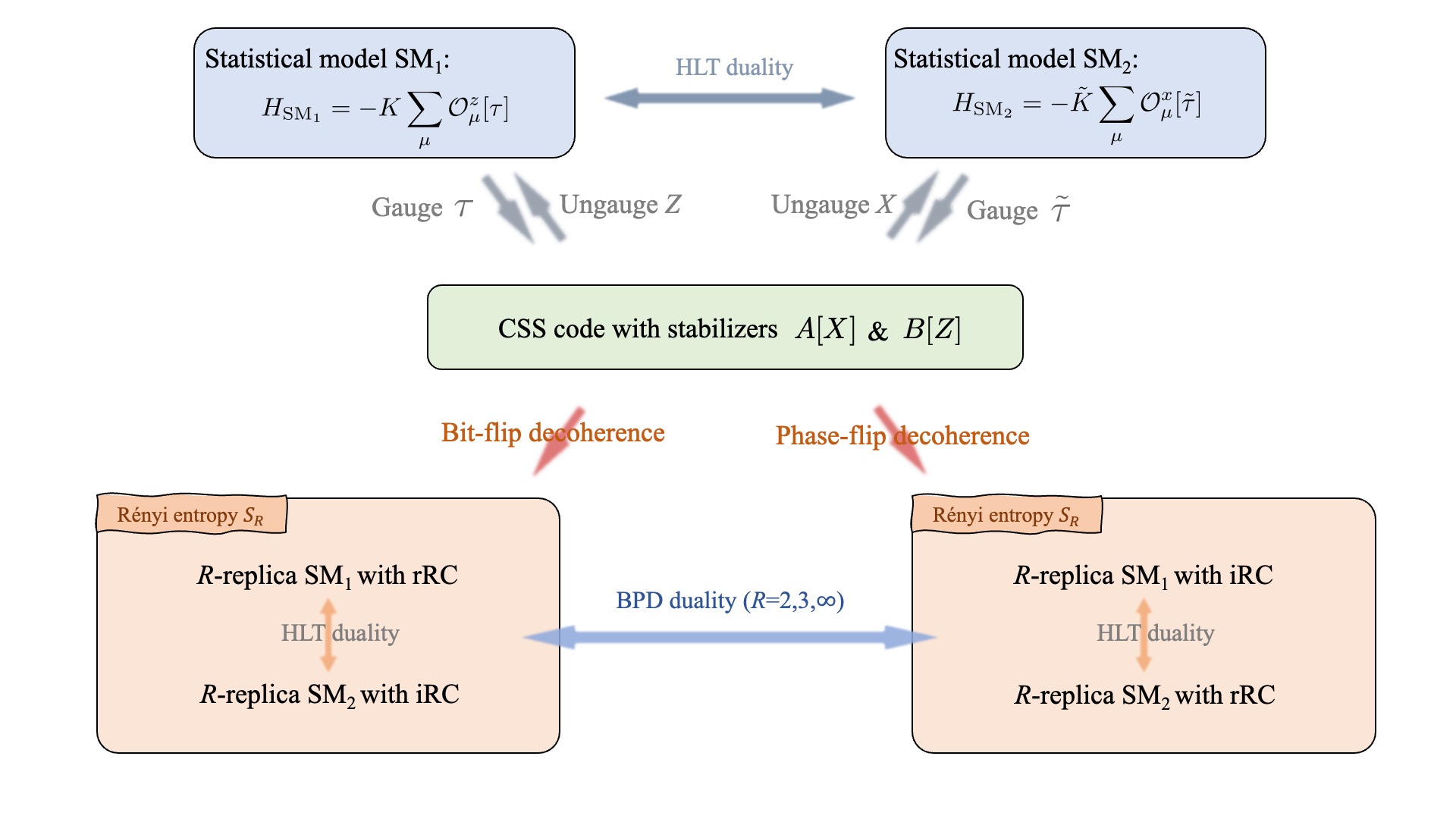}
    \caption{The tapestry of dualities of a general CSS code. In Sec. \ref{sec:sm_construction}, we introduce the statistical models $\SM{1,2}$ through ungauging the CSS code and introduce the $R$-replica $\SM{1,2}$ with randomness (rRC and iRC) as the descriptions of the R\'enyi entropies of the decohered CSS code. The HLT and BPD dualities are discussed in Sec. \ref{sec:sm_construction} and \ref{sec:tapestry}.}
    \label{fig:duality1}
\end{figure*}

We now develop the systematic framework to characterize the effects of decoherence in the error-corrupted mixed state $\rho = {\cal N}_x \circ {\cal N}_z (\rho_0)$ using statistical models. These decoherence effects are manifested in the behavior of $\Tr \rho^R$, which are directly related to the $R$-th R\'enyi entropies $S_R = \frac{1}{1-R} \log\left( \Tr \rho^R \right)$. Singularities in $S_R$ (as functions of the error rates $p_x$ and $p_z$) signal the DIPTs. In particular, the singularity in the limit $R\rightarrow 1$ is the error threshold of the decodability of the code. In the following, we show that $\Tr \rho^R$ can be formulated as the partition functions of a family of statistical models of classical $\mathbb Z_2$ spins with random couplings. The R\'enyi entropies are, therefore, the respective free energies. The DIPTs can then be identified as the phase transitions in these statistical models. This subsection is devoted to the general construction of these statistical models for the decohered CSS code in the infinite-system limit. Concrete examples will be provided in the subsequent subsections Sec. \ref{sec:3dtoric} and \ref{sec:xcube}. 

For a CSS code ${\cal C}$, decoherence effects of the bit-flip errors decouple from those of the phase-flip errors. The reason is that the former only affects the $B_j[Z]$ stabilizers while the latter only affects the $A_i[X]$ stabilizers. The decoupling is manifested by the following identity in the infinite-system limit
\begin{align}
    \Tr \rho^R = \Tr \rho_{\rm b}^R  \times \Tr \rho_{\rm p}^R.
    \label{eq:bp_fractorized}
\end{align}
In infinite-system limit, the absence of local logic operators implies $|\bra{\Omega} X_E Z_{E'}\ket{\Omega}| = |\bra{\Omega} X_E \ket{\Omega}| \times  |\bra{\Omega} Z_{E'} \ket{\Omega}|$ for the logical state $\ket{\Omega}$ and any pair of finite subsets of qubits $E,E'$. This property leads to Eq. \eqref{eq:bp_fractorized}. The decoupling of the bit-flip and phase-flip errors allows us to independently formulate the statistical models for $\Tr \rho_{\rm b}^R $ and $\Tr \rho_{\rm p}^R$ without loss of generality. Also, note that the absence of local logical operators makes all logical states locally indistinguishable. For this reason, different choices of the error-free logical state $\rho_0= |\Omega\rangle \bra{\Omega}$ will not affect the construction of the statistical models in the infinite-system limit. 

As the first step towards the statistical models, we use Eqs. \eqref{eq:rho_b_Expansion} and \eqref{eq:rho_p_Expansion} to expand $\Tr \rho_{\rm b/p}^R $ as a sum over the vector spaces ${\cal V}$, $\vx$, and $\vz$. Take $\Tr \rho_{\rm b}^2$ as an example. In the infinite-system limit,
\begin{align}
    \Tr(\rho_{\rm b}^2) & = \sum_{E, E' \in {\cal V}}P_x(E)P_x(E')|\braket{X_{E}X_{E'}}_\Omega|^2 \nonumber\\
    & = \sum_{E \in {\cal V}} \sum_{C \in {\cal V}_x }P_x(E)P_x(E+C).
    \label{eq:2ndRenyi_rb}
\end{align}
Here, we've used the fact that $|\braket{X_{E}X_{E'}}_\Omega|^2 =0$ unless $X_{E}X_{E'} = X_{E+E'}$ can be written as a product of the $A_i[X]$ stabilizers, i.e. $E + E' \in {\cal V}_x$.  When $E + E' \in {\cal V}_x$, we have $|\braket{X_{E}X_{E'}}_\Omega|^2 =1$ and $E' = E + C$ for some $C \in \vx$. For general $R$, we have
\begin{align}
    & \Tr(\rb^R) \nonumber \\
    & = \sum_{E \in {\cal V}} \, \sum_{C_{2,..,R} \in \vx} P_x(E) P_x(E+C_2)\dots P_x(E+C_R) \nonumber\\
    &= \frac{1}{2^{{\rm dim} \vx}} ~ \sum_{E\in {\cal V}} \left(\sum_{C_{1,2,..,R} \in \vx}P_x(E + C_1)\dots P_x(E+C_R)\right). \label{eq:renyisymm_b}
\end{align}
The last line is obtained from the second by shifting $E \rightarrow E+C_1$, $C_{2,...,R} \rightarrow - C_1 +  C_{2,...,R}$ with $C_1 \in \vx$. 

Later, we will see that $\Tr(\rb^R)$ can be mapped to the partition function of an $R$-replica statistical model of classical $\mathbb Z_2$ spins with random couplings. $E \in {\cal V}$ specifies a pattern of the random couplings, and the partition function of each replica is given by $\sum_{C\in \vx} P_x(E + C)$ (before averaging/summing over all randomness patterns $E$). The interaction between different replicas is only mediated by the random coupling $E$. 

In the $R\rightarrow 1$ limit, the R\'enyi entropy $S_R$ reduces to the von Neumann entropy $S_1$:  
\begin{align}
     S_1(\rb) & = \lim_{R\rightarrow 1} \frac{1}{1-R} \left(\Tr(\rb^R)-1 \right) \nonumber \\
    &= -\sum_{E\in {\cal V}} P_x(E) \log\left(\sum_{C\in \vx}P_x(E+C)\right), \label{eq:vn}
\end{align}
which turns out to be the quenched-disorder-average of the free energy $-\log\left(\sum_{C\in \vx}P_x(E+C)\right)$ in the random statistical model.

Similar to $\Tr(\rb^R)$, we can expand $\Tr(\rp^R)$ as
\begin{align}
    & \Tr(\rp^R) \nonumber \\
    &= \frac{1}{2^{{\rm dim} \vz}} ~  \sum_{E\in {\cal V}} \left(\sum_{C_{1,2,..,R} \in \vz}P_z(E + C_1)\dots P_z(E+C_R)\right), \label{eq:renyisymm_p}
\end{align}
which will also be mapped to the partition function of certain $R$-replica statistical models with random couplings. Consequently, the von Neumann entropy $S_1(\rp) = \lim_{R\rightarrow 1} S_R(\rp)$ will be identified as the quenched-disorder-average free energy of the corresponding statistical model.

In the following, we present the systematic construction of the random statistical models of classical $\mathbb Z_2$ spins whose partition functions map to $\Tr(\rho_{\rm b/p}^R)$. For the sake of clarity, we introduce these statistical models in two steps. First, we construct these statistical models without randomness. These non-random statistical models naturally arise from ungauging the CSS code. In the second step, we establish the mapping between the multi-replica version of these statistical models (with randomness) and the quantities $\Tr (\rho_{\rm b/p}^R)$ (which are related to the R\'enyi entropies $S_R(\rho_{\rm b/p})$). 

Now, we perform the first step of our construction. For a general CSS code ${\cal C}$, we introduce two (non-random) classical statistical models, denoted as $\SM1$ and $\SM2$. They can be obtained naturally by treating the CSS code as a generalized $\mathbb{Z}_2$ gauge theory and then ungauging it. As we will see, this ungauging procedure is exactly the inverse of the gauging procedure introduced in Ref. \onlinecite{Vijay_2016}, which constructs CSS codes from statistical models of classical $\mathbb{Z}_2$ spins. In this step, we explain the ungauging as a procedure introduced by hand. When we consider $\Tr (\rho_{\rm b/p}^R)$ later in the second step, the ungauging is effectively implemented by the errors.  

The ungauging procedure follows from the intuition based on the 2D toric code (see Fig. \ref{fig:toricungauge}), which is equivalent to a conventional $\mathbb{Z}_2$ gauge theory. The general idea is as follows. For a general CSS code $\cal C$, we can treat the $Z_\mu$'s as the gauge field variables. The $A_i[X]$ terms in the CSS code Hamiltonian Eq. \eqref{eq:CSS_Ham} implement the generalized Gauss law while the $B_j[Z]$ terms describe the energy costs of gauge fluxes. Ungauging the CSS code ${\cal C}$ from this perspective produces the classical statistical model $\SM1$. If we instead treat $X_\mu$ as the gauge field and interchange the roles of $A_i[X]$ and $B_j[Z]$, ungauging the CSS code $\cal C$ results in classical statistical model SM$_2$.

\begin{figure}[tb]
    \centering
    \includegraphics[width=0.95 \linewidth]{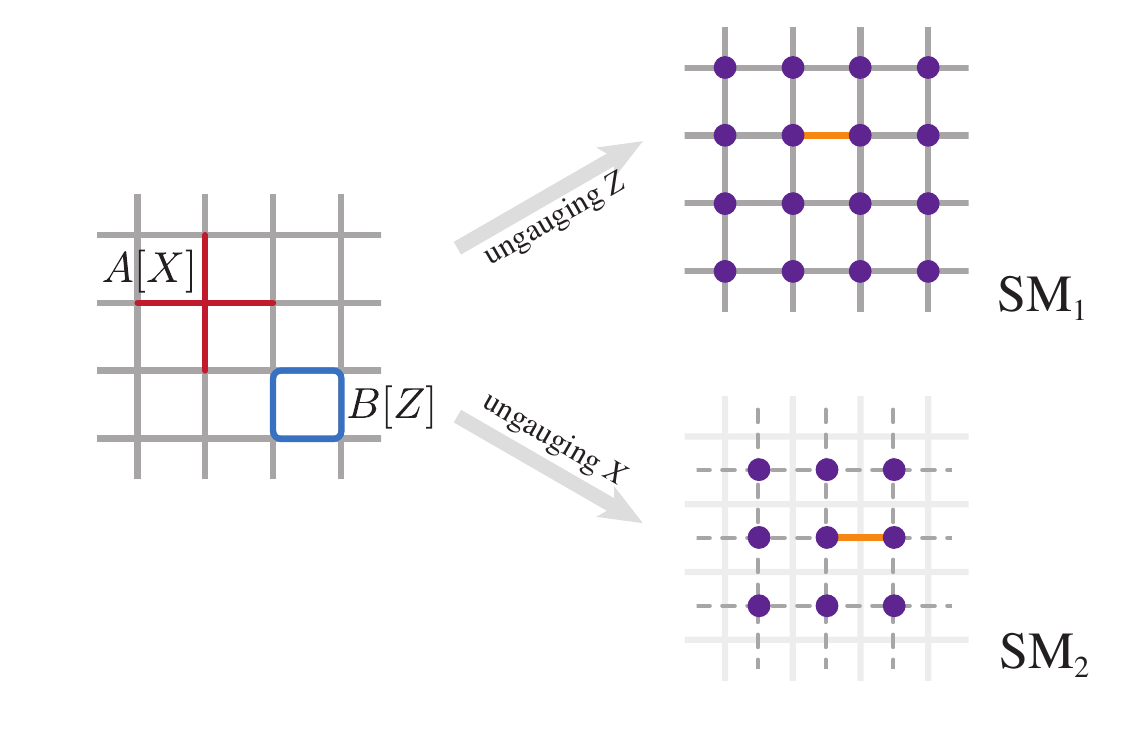}
    \caption{For the 2D toric code, the two types of stabilizers $A[X]$ and $B[Z]$ are pictorially represented above. Ungauging the gauge field $Z$ and the gauge field ungauging $X$ in the 2D toric code produce the 2D Ising model on the original lattice and that on the dual lattice, respectively. The original square lattice is depicted in gray solid lines, and the dual lattice is depicted in gray dashed lines. Each purple dot represents a classical $\mathbb Z_2$ spin $\tau$/$\tilde\tau$, and the orange edge presents the near-neighbor interactions in the Ising models. (see Sec. \ref{sec:2dtc} for the microscopic details) }
    \label{fig:toricungauge}
\end{figure}

We now explicitly carry out the ungauging procedure to construct the non-random version of SM$_1$ from the general CSS code $\cal C$.  We introduce a $\mathbb{Z}_2$ spin variable $\tau_i$ at the center of each $A_i[X]$ term and treat it as the matter field coupled to the generalized $\mathbb{Z}_2$ gauge field $Z_\mu$. The operator $A_i[X] \tau_i^x$ for each $i$ generates the gauge transformation, where $\tau_i^x$ is a Pauli-$X$ operator. The gauge field $Z_\mu$ couples to the matter fields via the gauge-invariant interaction $Z_\mu {\cal O}^z_\mu[\tau^z]$ with
\begin{align}
   {\cal O}^z_\mu[\tau^z] = \prod_{i} (\tau^z_i)^{(\mathsf{a}_i)_{\mu}}.
   \label{eq:def_Oz}
\end{align}
Recall that $\mathsf{a}_i \in {\cal V}$ is the $\mathbb {Z}_2$ vector associated with the stabilizer $A_i[X]$. $(\mathsf{a}_i)_\mu$ is the $\mu$th component of $\mathsf{a}_i$. Due to the locality of each stabilizer $A_i[X]$, ${\cal O}^z_\mu[\tau^z]$ must be a local term as well for each $\mu$. The commutation relation $[A_i[X] \tau_i^x, Z_\mu {\cal O}^z_\mu[\tau^z] ] = 0 $ for any $i,\mu$ implies the gauge-invariance of the interaction $Z_\mu {\cal O}^z_\mu[\tau^z]$. Hence, we can write the following Hamiltonian which describes the generalized $\mathbb Z_2$ gauge theory coupled to matter fields:
\begin{align}
    H_{\rm gt1} = - \sum_i A_i[X] \tau_i^x - K \sum_\mu Z_\mu {\cal O}^z_\mu[\tau^z] - \sum_j B_j[Z],
    \label{eq:gauge_matter_int}
\end{align}
where $K>0$ is a coupling constant and the first term $\sum_i A_i[X] \tau_i^x$ leads to an emergent Gauss law (with matter fields) $A_i[X] \tau_i^x = 1,~\forall i$ at low energies. 

If we add an extra term $- g \sum_i \tau^x_i$ with a large positive $g$ to the Hamiltonian $H_{\rm gt1}$, the matter field $\tau_i$ will be gapped out and the CSS code Hamiltonian Eq. \eqref{eq:CSS_Ham} will be recovered as the low-energy theory.

On the other hand, the model $H_{\rm gt1}$ allows us to ungauge by turning off all the gauge fields, namely setting $Z_\mu = 1$, and removing the Gauss-law terms $-\sum_i A_i[X] \tau_i^x$. The resulting Hamiltonian is the Hamiltonian of SM$_1$ (up to a constant):
\begin{align}
    H_{{\rm SM}_1} = - K \sum_\mu {\cal O}^z_\mu[\tau]. 
    \label{eq:SM1_Ham_clean}
\end{align}
Here, we've suppressed the superscript $z$ for the spin variables $\tau_i$ to emphasize that $H_{{\rm SM}_1}$ is essentially a {\it classical} Hamiltonian with commuting spin variables $\tau_i = \pm 1$. It is helpful to note that the lattice of the spin variable $\tau_i$ of $\SM1$ differs from the lattice of the qubits labeled by $\mu$.

The construction of SM$_2$ is parallel to SM$_1$. We simply need to interchange the roles of the $A_i[X]$ and $B_j[Z]$ stabilizers. For $\SM2$, a $\mathbb{Z}_2$ spin variable $\ttau_j$ is introduced at the center of each $B_j[Z]$ stabilizer. With $X_\mu$ treated as the generalized $\mathbb{Z}_2$ gauge field, ungauging the CSS code produces the classical Hamiltonian of SM$_2$:
\begin{align}
    H_{{\rm SM}_2} = - \tilde{K} \sum_\mu \cO^x_\mu[\ttau],
    \label{eq:SM2_Ham_clean}
\end{align}
where $\tilde{K}$ is the coupling constant of SM$_2$ and the spin interaction ${\cal O}^x_\mu[\ttau]$ is given by 
\begin{align}
    {\cal O}_\mu^x [\ttau] = \prod_j (\ttau_j)^{(\mathsf{b}_j)_\mu}.
     \label{eq:def_Ox}
\end{align}
We emphasize that each $\cO_\mu^x [\ttau]$ term is a local spin interaction due to the locality of the $B_j[Z]$ stabilizers. 

The two (non-random) statistical models SM$_1$ and SM$_2$ are dual to each other under a Kramers-Wannier-like HLT duality. We defer the details of this duality to the next subsection, where we include this duality as a component of the tapestry of dualities. 

With $\SM{1,2}$ introduced, we are ready to establish the mapping between $\Tr\rho_{\rm b/p}^R$ for error-corrupted mixed states $\rho_{\rm b/p}$ and the $R$-replica $\SM{1,2}$ with random couplings. There are two types of random couplings: real random couplings and imaginary random couplings, which we refer to as rRC and iRC for short. 

For the error-corrupted mixed state $\rho_{\rm b}$, we obtain the following theorem.
\begin{tcolorbox}
\begin{theorem}\label{thm:sm_b}   
{\rm For the CSS code in the infinite system limit, $\Tr\rho_{\rm b}^R$ of the mixed state $\rb$ with bit-flip errors is proportional to the partition function of the $R$-replica SM$_1$ with real random couplings (rRC). It is also proportional to the partition function of the $R$-replica SM$_2$ with imaginary random couplings (iRC):}
    \begin{align}
        {\rm Tr}(\rb^R) &\propto \sum_{E\in {\cal V}} \left(Z_{\rm SM_1}(K, E) \right)^R \nonumber \\
        &\propto \sum_{E\in {\cal V}}  \left(W_{\rm SM_2}(\tilde{K},E)\right)^R
        \label{eq:thm1},
    \end{align}
{\rm where the coupling constants $K$ and $\tilde K$ are given by the bit-flip error rate $p_x \in [0,1/2]$ via $e^{-2K} = \tanh{\tilde{K}} = \frac{p_x}{1-p_x}$. The proportionality constants only depend on $p_x $ smoothly and, hence, are unimportant for potential phases and DIPTs.} 
\end{theorem}
\end{tcolorbox}

On the first line of Eq. \eqref{eq:thm1}, $Z_{\rm SM_1}(K, E)$ is the single-replica partition function of SM$_1$ in the presence of the rRC configuration labeled by the error chain $E$:
\begin{align}
        Z_{\rm SM_1}(K, E) &= \sum_{\{\tau_i=\pm 1\}}\exp \left(K \sum_\mu  (-1)^{E_\mu}\mathcal{O}^z_\mu[\tau]\right).
        \label{eq:ZSM1_dis}
\end{align}
Here, $E_\mu$ is the $\mu$th component of the $\mathbb{Z}_2$ vector $E\in {\cal V}$. $E$ adds extra sign randomness to the real coupling constant $K$. Hence the name rRC. Note that $K>0$ since $p_x\in [0,1/2]$. One can show that
\begin{align}
   Z_{\rm SM_1}(K, E) = [p_x(1-p_x)]^{-N/2} N_s\sum_{C \in \vx} P_x(E+C),
   \label{eq:ZSM1_loop}
\end{align}
where $N_s$ is a factor related to internal symmetries of $\SM1$. We refer to App. \ref{app:thmsm} for the derivation of this expression. The first line of Eq. \eqref{eq:thm1} is an immediate consequence of Eqs. \eqref{eq:ZSM1_loop} and \eqref{eq:renyisymm_b}. Eq. \eqref{eq:ZSM1_loop} also shows that $Z_{\rm SM_1}(K, E)$ is proportional to the total probability of all the error chains within the same equivalence class $[E]_x \in {\cal V}/ \vx$. Physically, the error chains within the same class $[E]_x$ differ from each other only by elements in $\vx$ and, hence, cause excitations on the same set of $B_j[Z]$ stabilizers.

To provide some heuristics, we argue that the bit-flip errors suppress the coherent quantum fluctuations of the generalized $\mathbb Z_2$ gauge field $Z_\mu$. Hence, in the CSS code with bit-flip decoherence, a modified version of the Eq. \eqref{eq:gauge_matter_int} should appear. The modification includes turning off the first term (that generates coherent quantum fluctuations of the gauge field $Z_\mu$) and treating gauge field $Z_\mu$ as a classical degree of freedom with values $\pm1$. After this modification, $\SM1$ with rRC as defined in Eq. \eqref{eq:ZSM1_dis} emerges.

On the second line of Eq. \eqref{eq:thm1}, $W_{\rm SM_2}(\tilde{K},E)$ is the single-replica partition function of SM$_2$ in the presence of the iRC configuration labeled $E$:
\begin{align}
   & W_{\rm SM_2}(\tilde{K},E) = \sum_{\{\tilde{\tau}_j = \pm 1\}}\exp(-H_{\rm SM_2}) \prod_{\mu} \left( \cO^x_\mu[\tilde\tau]\right)^{E_\mu}
    \nonumber \\
   & = \sum_{\{\tilde{\tau}_j = \pm 1\}} \exp\left(-H_{\rm SM_2} -  \frac{\i \pi}{2} \sum_\mu E_\mu \Big(\cO^x_\mu[\tilde\tau] - 1\Big)\right),
    \label{eq:WSM2_dis}
\end{align}
where $H_{\rm SM_2}$ is the same as in Eq. \eqref{eq:SM2_Ham_clean}. For $p_x\in [0,1/2]$, $\tilde K$ is positive. In $W_{\rm SM_2}(\tilde{K},E)$, a non-trivial $E$ introduces extra imaginary coupling constants for the spin interaction $\cO^x_\mu[\tilde\tau]$. And hence the name iRC.

Based on Theorem \ref{thm:sm_b}, the phases and the DIPTs of the decohered CSS code probed by the R\'enyi entropy $S_R(\rb)$ can be investigated via studying the phases and transitions in the $R$-replica $\SM1$ with rRC and in the $R$-replica $\SM2$ with iRC. The fact that two different statistical models describe the behavior of the same quantity $\Tr \rb^R$ (or $S_R(\rb)$) indicates a duality between them. This duality is, in fact, an HLT duality, which we will discuss more about in the next subsection. For the proof of Theorem \ref{thm:sm_b} (including the HLT duality), we refer to App. \ref{app:thmsm}. The applications of Theorem \ref{thm:sm_b} in concrete examples are provided in Sec. \ref{sec:3dtoric} and \ref{sec:xcube}.

Using Eq. \eqref{eq:ZSM1_loop} and taking the $R\rightarrow 1$ limit of Theorem \ref{thm:sm_b}, we obtain 
\begin{tcolorbox}
\begin{corollary} \label{corollary:S1_b}
{\rm The von Neumann entropy of the error-corrupted state $\rb$ is given by the quenched-disorder-averaged free energy of SM$_1$ with rRC (up to an unimportant additive constant):}
\begin{align}
    S_1(\rb) = -\sum_{E\in {\cal V}} P_x(E) \log(Z_{\rm SM_1}(K, E)),
    \label{eq:S1_b}
\end{align}
{\rm where $K$ is given by the bit-flip error rate $p_x$ via $e^{-2K} = \frac{p_x}{1-p_x}$.}
\end{corollary}
\end{tcolorbox}

\noindent Here, we treat the rRC configuration $E$ as a disorder following the probability distribution $P_x(E)$. Therefore, the right-hand side of Eq. \eqref{eq:S1_b} is interpreted as a quenched-disorder-averaged free energy. In a general disordered spin model, the probability distribution of disorder and the coupling constant $K$ can be independent parameters. The relation $e^{-2K} = \frac{p_x}{1-p_x}$ that appears in $\SM1$ with rRC is known as the Nishimori condition \cite{nishimori1981internal,nishimori2001statistical}.

We remark that the right-hand side of Eq. \eqref{eq:S1_b} is exactly the (averaged) free energy of the disordered statistical model introduced by Ref. \onlinecite{dennis2002topological} (and later generalized by Ref. \onlinecite{Flammia2021}) to study the decodability of $\mathbb{Z}_2$ stabilizer codes with Pauli errors (bit-flip errors in this case). In particular, Ref. \onlinecite{dennis2002topological} established that the phase transition (where the quenched-disorder-averaged free energy becomes singular) in the disordered statistical model signals the decoding error threshold for the code.  From the perspective of our current work, the error threshold is viewed as the $R\rightarrow 1$ limit of the family of DIPTs labeled by $R$. Recall that the DIPT with index $R$ is associated with the singularity of the R\'enyi entropy $S_R$ of the decohered code. We generally expect the critical error rate $p_{x}^\star(R) $ of the DIPT to be a function of $R$. The bit-flip error threshold for decoding is given by the limit $p_{x}^\star(R \rightarrow 1)$. As a clarification, for a generic $R$, we use the term ``critical error rate" for the error rate where the R\'enyi entropy $S_R$ develops singularity and the DIPT with index $R$ occurs. The term ``(decoding) error threshold" is only associated with the DIPT in the limit $R\rightarrow 1$ and is hence given by the critical error rate in the same limit. The same terminology convention applies to all types of errors considered in this paper. We will discuss the physical meaning of DIPTs at different $R$'s and the dependence of the critical error rate $p_{x}^\star(R)$ on $R$ in Sec. \ref{sec:differentR}.

For the mixed state $\rp$ corrupted by phase-flip errors, we prove a similar theorem that maps $\Tr\rho_{\rm p}^R$ to the partition functions of two $R$-replica statistical models with randomness.
\begin{tcolorbox}
\begin{theorem}\label{thm:sm_p}   
{\rm For the CSS code in the infinite system limit, $\Tr\rho_{\rm p}^R$ of the mixed state $\rp$ corrupted by phase-flip errors is proportional to the partition function of the $R$-replica SM$_2$ with rRC. It is also proportional to the partition function of the $R$-replica SM$_1$ with iRC: }
    \begin{align}
        {\rm Tr}(\rp^R) &\propto \sum_{E\in {\cal V}} \left(Z_{\rm SM_2}(\tilde{K}, E) \right)^R \nonumber \\
        &\propto \sum_{E\in {\cal V}}  \left(W_{\rm SM_1}(K,E)\right)^R,
    \end{align}
{\rm where the coupling constants $K$ and $\tilde K$ are given by the phase-flip error rate $p_z\in [0,1/2]$ via $e^{-2\tilde K} = \tanh{K} = \frac{p_z}{1-p_z}$. The proportionality constants only depend on $p_z$ smoothly and, hence, are unimportant for potential phases and DIPTs.} 
\end{theorem}
\end{tcolorbox} 
\noindent Here, the partition function of the $R$-replica SM$_2$ with rRC and that of the $R$-replica SM$_1$ with iRC are defined in parallel with Eq. \eqref{eq:ZSM1_dis} and \eqref{eq:WSM2_dis}:
 \begin{align}
        Z_{\rm SM_2}(\tilde{K}, E) &= \sum_{\{\ttau_j=\pm 1\}}\exp \left(\tilde{K} \sum_\mu  (-1)^{E_\mu}\cO^x_\mu[\ttau]\right),
        \label{eq:ZSM2_dis}
\end{align}
and 
\begin{align}
    & W_{\rm SM_1}(K,E) = \sum_{\{\tau_i = \pm 1\}}\exp(-H_{\rm SM_1}) \prod_{\mu} \left( {\cal O}^z_\mu[\tau]\right)^{E_\mu}
    \label{eq:WSM1_dis}
     \nonumber \\
   & = \sum_{\{\tau_i = \pm 1\}} \exp\left(-H_{\rm SM_1} -  \frac{\i \pi}{2} \sum_\mu E_\mu \Big(\cO^z_\mu[\tau] - 1\Big)\right).
\end{align}
Note that statistical models in Theorem \ref{thm:sm_p} for the phase-flip errors are related to those in Theorem \ref{thm:sm_b} for the bit-flip errors by interchanging the roles of rRC and iRC. 

Similar to Eq. \eqref{eq:ZSM1_loop}, we can show that 
\begin{align}
    Z_{\rm SM_2}(\tilde{K}, E)  = [p_z(1-p_z)]^{-N/2} \tilde{N}_s\sum_{C \in \vz} P_z(E+C),
   \label{eq:ZSM2_loop}
\end{align}
where $\tilde{N}_s$ is a factor associated with the internal symmetries of ${\rm SM}_2$. This equation implies that $Z_{\rm SM_2}(\tilde{K}, E)$ is the total probability of phase-flip error within the same equivalence class $[E]_z \in {\cal V}/\vz$. All error chains within the same class $[E]_z$ only differ from each other by elements in $\vz$ and, hence, lead to the same pattern of excitations on the $A_i[X]$ stabilizers. 

Just like the case with bit-flip errors, the two $R$-replica random statistical models in Theorem \ref{thm:sm_p} are dual to each other under an HLT duality, which we will discuss more in the next subsection.

With Eq. \eqref{eq:ZSM2_loop}, we can take the $R\rightarrow 1$ limit of Theorem \ref{thm:sm_p} and obtain the following corollary.
\begin{tcolorbox}
\begin{corollary} \label{corollary:S1_p}
{\rm The von Neumann entropy of the error-corrupted state $\rp$ is given by the quenched-disorder-averaged free energy of SM$_2$ with rRC (up to unimportant additive constants):}
\begin{align}
    S_1(\rp) = -\sum_{E\in {\cal V}} P_z(E) \log(Z_{\rm SM_2}(\tilde{K}, E)),
    \label{eq:S1_p}
\end{align}
{\rm where $\tK$ is given by the phase-flip error rate $p_z$ via $e^{-2\tK} = \frac{p_z}{1-p_z}$.}
\end{corollary}
\end{tcolorbox}
Again, combining Ref. \onlinecite{dennis2002topological} and Eq. \eqref{eq:S1_p}, we conclude that the phase-flip error decoding threshold of the code is the critical phase-flip error rate $p_{z}^\star(R)$ of the DIPT in the limit $R\rightarrow 1$.

\subsection{Tapestry of dualities of a general CSS code}
\label{sec:tapestry}

In this subsection, we discuss the tapestry (Fig. \ref{fig:duality1}) woven by the dualities among the statistical models originating from a general CSS code $\cal C$. 

First, there is a Kramers-Wannier-like high-low-temperature (HLT) duality between the non-random SM$_1$ and SM$_2$ obtained from the same CSS code $\cal C$ through ungauging. For bit-flip errors, Theorem \ref{thm:sm_b} shows that two $R$-replica statistical models with random couplings describe the same quantity $\Tr \rb^R$ associated with the $R$th R\'enyi entropy $S_R(\rb)$ of the error-corrupted mixed state $\rb$. As shown in the bottom left corner of Fig. \ref{fig:duality1}, these two statistical models are also dual to each other under an HLT duality, a generalization of the HLT duality that relates the non-random $\SM1$ and $\SM2$.  Similarly, for phase-flip errors, the $R$th R\'enyi entropy $S_R(\rp)$ of the mixed state $\rp$ is described by another pair of $R$-replica statistical models with random couplings as stated in Theorem \ref{thm:sm_p}. These two $R$-replica statistical models are also dual under an HLT duality, as shown in the bottom right corner of Fig. \ref{fig:duality1}. In addition to these HLT dualities, we find that there are extra dualities, dubbed the BPD (bit-phase-decoherence) dualities, that map the random $R$-replica statistical models associated with the bit-flip errors to those associated with phase-flip errors. In the following, we provide a general discussion and the physical intuition for these dualities. The full technical details of the proofs of these dualities are presented in the App. \ref{app:thmsm} and \ref{app:bpd}. 

Now, we discuss the HLT duality between the non-random SM$_1$ and SM$_2$. We provide below the physical intuition of this duality by comparing the high-temperature expansion of the partition of $\SM1$ and the low-temperature expansion of the partition function of $\SM2$. The partition function of the non-random SM$_1$ can be written as a high-temperature expansion:
\begin{align}
        Z_{\rm SM_1}(K) &= \sum_{\{\tau_i=\pm 1\}}\exp \left(K \sum_\mu \mathcal{O}^z_\mu[\tau]\right) \nonumber \\
        & = \sum_{\{\tau_i=\pm 1\}} (\cosh K)^N \prod_\mu  \left(1+ \mathcal{O}^z_\mu[\tau] \tanh K \right).
        \label{eq:ZSM1_clean}
\end{align}
Recall that $N$ is the number of qubits in the original CSS code, which equals the number of $\mathcal{O}^z_\mu[\tau]$ terms in $\SM1$. After we expand the product over $\mu$ and sum over all the spin configurations $\{\tau_i = \pm 1\}$, the non-vanishing contribution must come from the products of $\mathcal{O}^z_\mu[\tau]$ terms that equal identity. Based on Eq. \eqref{eq:a_dot_b} and Eq. \eqref{eq:def_Oz}, we know that $\prod_{\mu} \left(\mathcal{O}^z_\mu[\tau]\right)^{({\rm b}_j)_\mu} = 1$ for any $j$.  In fact, we can readily see that $\prod_{\mu} \left(\mathcal{O}^z_\mu[\tau]\right)^{C_\mu} = 1$ for a local set $C \in {\cal V}$ implies $C \in \vz$ (in the absence of any local logical operators in the code $\cal C$ which we've assumed). Hence, in the infinite system limit, we can write 
\begin{align}
    Z_{\rm SM_1}(K) \propto (\cosh K)^N \sum_{C \in \vz} \left( \tanh K \right)^{|C|}
    \label{eq:SM1_HTE}
\end{align}
with a proportionality constant only depending on the system size. Strictly speaking, the summation here should run over the space $\vx^\perp$ instead. Note that the space $\vx^\perp/\vz$ is associated with the logical-$Z$ operators of the code which are all non-local. Also, the number of logical-$Z$ operators $N_c$ obeys $N_c/N\rightarrow 0$ as $N\rightarrow \infty$. Hence, replacing $\sum_{C \in\vx^\perp}$ by $\sum_{C \in \vz}$ does not change the free energy density of SM$_1$ in the infinite system limit. 

For SM$_2$, we perform a low-temperature expansion of the partition function. A classical ground state of the SM$_2$ Hamiltonian $H_{{\rm SM}_2}$ Eq. \eqref{eq:SM2_Ham_clean} is the state with $\ttau_j = +1$ for all $j$. The ground state energy is $- N \tK$. The energy cost of a single spin flip $\ttau_j\rightarrow -1$ at the site $j$ is $ + 2\tK |{\rm b}_j|$. The low-temperature expansion is a summation over all possible spin flips on top of the classical ground state:
\begin{align}
    Z_{\rm SM_2}(\tK) \propto e^{ N\tK} \sum_{C \in \vz} e^{- 2\tK |C|},
     \label{eq:SM2_LTE}
\end{align}
where the proportionality constant is given by the degeneracy of the classical Hamiltonian $H_{{\rm SM}_2}$, which is independent of $\tK$ but related to the internal symmetry of $H_{{\rm SM}_2}$. By comparing Eq. \eqref{eq:SM2_LTE} and Eq. \eqref{eq:SM1_HTE}, we establish the HLT duality between the non-random SM$_1$ and SM$_2$ under the condition 
\begin{align}
    \tanh{K} = e^{-2\tK}. \label{eq:kw}
\end{align}
This condition relates a large (small) positive value of $K$ to a small (large) positive value of $\tK$, which is natural for an HLT duality. Another way to obtain this HLT duality is to apply Wegner's general approach for dualities in Ising-type models \cite{wegner1971duality}.

Now, we discuss the dualities among the random statistical models appearing in the decohered CSS code. In the case of bit-flip errors, as shown in Theorem \ref{thm:sm_b}, the $R$-replica $\SM1$ with rRC is dual to the $R$-replica $\SM2$ with iRC under an HLT duality. The duality relation between the coupling constants $K$ and $\tK$ is the same as Eq. \eqref{eq:kw} for the non-random case. The proof of the HLT duality between the two $R$-replica random statistical models generalizes the discussion above for the non-random case. The mathematical details of the proof are provided in App. \ref{app:thmsm}. For the case of phase-flip errors, a parallel analysis can be made for the two different $R$-replica random statistical models that describe $\tr \rp^R$, as shown in Theorem \ref{thm:sm_p}. These two statistical models are also dual to each other under a similar HLT duality. 

Next, we introduce the BPD dualities between the random statistical models associated with bit-flip errors and those associated with phase-flip errors. The statement of these BPD dualities is the following.

\begin{tcolorbox}
\begin{theorem}\label{thm:BPD}   
{\rm For $R = 2,3$, and $\infty$, the $R$-replica random statistical models that describe the R\'enyi entropy $S_R(\rb)$ caused by the bit-flip errors are dual to the statistical models that describe the R\'enyi entropy $S_R(\rp)$ caused by the phase-flip errors. We can summarize these BPD dualities as
}
\begin{align}
        {\rm Tr}(\rb^R) \propto {\rm Tr}(\rp^R), ~~~ {\rm for~}R=2,3, {\rm and~ } \infty,
\end{align}
{\rm (up to unimportant proportionality constants) when the error rates $0<p_{x,z}<\frac{1}{2}$ satisfy the duality relations}
\begin{align}
    \left[(1-p_x)^R + p_x^R\right]\left[(1-p_z)^R + p_z^R\right] = \frac{1}{2^{R-1}}. \label{eq:BPD_duality_condition_1}
\end{align}
\end{theorem}
\end{tcolorbox}
In particular, in $R\to \infty$ limit, the duality condition reduces to
\begin{equation}
    (1-p_x)(1-p_z) = \frac{1}{2}
\end{equation}
for the range of error rates $p_{x,z} \in (0,1/2)$ under consideration. Note that duality relation Eq. \eqref{eq:BPD_duality_condition_1} maps the weak bit-flip decoherence ($p_x$ close to 0) to the strong phase-flip decoherence ($p_z$ close to 1/2), and vice versa. Hence, the BPD duality is a ``strong-weak" duality for the strength of decoherence.
We provide the mathematical details of the proof of this duality in App. \ref{app:bpd}. Here, we sketch the general idea that leads to these BPD dualities. Based on Theorem \ref{thm:sm_b} and \ref{thm:sm_p}, we can use the $R$-replica SM$_1$ with rRC to describe $S_R(\rb)$ and the $R$-replica SM$_1$ with iRC to describe $S_R(\rp)$. Integrating out the randomness, rRC or iRC, results in the interactions between the $R$ replicas of $\SM1$ (see App. \ref{app:intdisorder}). For a general $R$, the inter-replica interactions mediated by rRC differ from those mediated by iRC. However, the cases with $R=2,3$ are exceptions. The inter-replica interactions mediated by rRC and iRC are identical when the error rates $p_{x,z}$ satisfy the relations Eq. \eqref{eq:BPD_duality_condition_1}. Therefore, there are BPD dualities between the $R$-replica random statistical models for $S_R(\rb)$ and $S_R(\rp)$ when $R=2,3$.

We can also understand the BPD duality with $R=2$ using the HLT duality between the non-random SM$_1$ and SM$_2$. Combining Eqs. \eqref{eq:2ndRenyi_rb}, \eqref{eq:ZSM1_dis} and \eqref{eq:ZSM1_loop}, we can integrate out the randomness $E$ (see App. \ref{app:intdisorder}) and obtain
\begin{align}
    \Tr \rb^2 \propto Z_{\rm SM{_1}} (K') {~~\rm with~~ } K'= {\rm arctanh}\left((1-2p_x)^2\right),
    \label{eq:2ndRenyi_rb_clean_SM1}
\end{align}
where $Z_{\rm SM{_1}}(K') $ is the partition function of the non-random SM$_1$. Essentially, integrating out the real random-coupling ``renormalizes" the coupling constant from $K = -\frac{1}{2}\log \frac{p_x}{1-p_x}$ to $K'$. Similarly, we can show that
\begin{align}
    \Tr \rp^2 \propto Z_{\rm SM{_2}} (\tK') {~~\rm with~~ } \tK' = {\rm arctanh}\left((1-2p_z)^2\right).
    \label{eq:2ndRenyi_rp_clean_SM2}
\end{align}
The BPD duality with $R=2$ shown in Theorem \ref{thm:BPD} is equivalent to the HLT duality between the non-random statistical models in Eq. \eqref{eq:2ndRenyi_rb_clean_SM1} and \eqref{eq:2ndRenyi_rp_clean_SM2}.

The BPD duality at $R\rightarrow \infty$ can also be justified using the HLT duality between the non-random SM$_1$ and SM$_2$. From Eq. \eqref{eq:thm1}, we notice that $\Tr \rb^R \Big|_{R\rightarrow \infty}$ is dominated by the randomness pattern $E\in {\cal V}$ that maximizes the partition function $Z_{\SM{1}}(K, E)$. A high-temperature expansion similar to Eq. \eqref{eq:SM1_HTE} implies that the maximum of $Z_{\SM{1}}(K, E)$ can be reached by the trivial randomness pattern $E=0$. Therefore, the behavior of $\Tr \rb^R \Big|_{R\rightarrow \infty}$ is dominated by (the $R$th power of) the  partition function $Z_{\SM{1}}(K)^R$ of the non-random $\SM1$. Similarly, $\Tr \rp^R \Big|_{R\rightarrow \infty}$ is dominated by the partition function $Z_{\SM{2}}(\tK)^R$ without randomness. Therefore, the HLT duality between the non-random $\SM{1}$ and $\SM{2}$ leads to the BPD duality between the R\'enyi entropies $S_R(\rb)$ and $S_R(\rp)$ in the limit $R\rightarrow \infty$.

\subsection{Discussion on the interpretation of DIPTs and the monotonicity of critical error rates}
\label{sec:differentR}

In this subsection, we briefly discuss the physical interpretation of DIPTs with different values of $R$. Also, we present a conjecture on the monotonicity of the critical error rates $p^\star_{x,z}(R)$ of the DIPTs as functions of the R\'enyi index $R$. 

As discussed in the previous subsections, the critical error rates $p^\star_{x,z}$ of the DIPT in the limit $R\rightarrow 1$ matches the error threshold for the decodability of the logical information in the error-corrupted CSS code. For the DIPT with $R=2$, in addition to the singularity of the 2nd R\'enyi entropy, we can also view it from the perspective of a quantum phase transition in the doubled Hilbert space, which we explain in the following. 

The discussions of bit-flip errors and phase-flip errors are completely parallel. Let's take the former as an example. Using the Choi-Jamio{\l}kowski isomorphism \cite{CHOI1975,JAMIOLKOWSKI1972}, we can map the mixed-state density matrix $\rb$ to a pure state $\kket{\rb}$, called the Choi representation of $\rb$, in the doubled Hilbert space. $\kket{\rb}$ is related to its error-free counterpart $\kket{\rho_0}$ via
\begin{align}
    \kket{\rb} \propto \left(
    e^{\tK \sum_\mu X_\mu \otimes X_\mu} \right)\kket{\rho_0},
\end{align}
with $\tK = {\rm arctanh}\left(\frac{p_x}{1-p_x}\right)$. One can show that $\kket{\rb}$ is a ground state of the following parent Hamiltonian in the doubled Hilbert space
\begin{align}
    H_{\rm b}^D = H_{\rm css} \otimes \mathds{1} +  \mathds{1}  \otimes H_{\rm css} + 2 \sum_j e^{-2\tK \sum_\mu ({\rm b}_j)_\mu X_\mu \otimes X_\mu}, 
    \label{eq:Ham_doubled}
\end{align}
where $H_{\rm css}$ is the CSS-code Hamiltonian given in Eq. \eqref{eq:CSS_Ham}. The construction of this parent Hamiltonian generalizes the construction for the decohered cluster states studied in Ref. \onlinecite{LeeYouXu2022}. A similar construction for the decohered 2D toric code was given in Ref. \onlinecite{BaoFanError_Field_Double}. App. \ref{app:choi} contains the details of the construction of Eq. \eqref{eq:Ham_doubled} in the most general setting. An interesting property is that $H_{\rm b}^D$ is frustration-free up to additive constants. In other words, we can decompose this Hamiltonian into a sum of positive-semidefinite local terms, and each of such local terms annihilates the ground state $\kket{\rb}$. 

Let's analyze the phase diagram of this model Eq. \eqref{eq:Ham_doubled} in the doubled Hilbert space. When $p_x$ is close to 0, we can treat the last term of $H_{\rm b}^D $ as a perturbation.  Since $H_{\rm css}$ by itself is a gapped Hamiltonian, the ground state $\kket{\rb}$ for small $p_x$ is smoothly connected to $\kket{\rho_0}$ which contains two decoupled copies of the same CSS code (in the doubled Hilbert space). Another regime is represented by the limit $p_x=1/2$. The ground state at $p_x=1/2$ is a stabilizer state with stabilizers,
\begin{align}
    A_i[X]\otimes \mathds{1},~\mathds{1} \otimes  A_i[X],~ X_\mu \otimes X_\mu,~\text{and}~B_j[Z]\otimes B_j[Z]. \nonumber 
\end{align}
This stabilizer state at $p_x =1/2$ is equivalent to a single copy of the original CSS code $\cal C$ (embedded inside the doubled Hilbert space). If two copies of the original CSS codes represent a quantum phase of matter different 
from a single copy, there must be a quantum phase transition between them. This quantum phase transition must occur at the critical error rate $p_x^\star(R=2)$ of the DIPT with $R=2$. Note that this quantum phase transition lives in the same spatial dimension as the statistical model but has non-trivial temporal dynamics in the doubled Hilbert space. The spatial correlation of this quantum phase transition can be equivalently captured by the random statistical models with $R=2$ discussed in previous subsections.  For example, the equal-time correlation function $\bbra{\rb} O_1 \otimes O_2 \kket{\rb} \propto \Tr\left(\rb O_1 \rb O_2^{\mathsf T}\right)$ can be translated into a correlation function of the 2D random statistical models with $R=2$. Here, $O_2^{\mathsf T}$ is the transpose of the operator $O_2$. When the statistical models for $R=2$ exhibit critical behavior, the same critical behavior appears in the spatial correlation of the quantum system in the doubled Hilbert space, indicating a quantum phase transition.

The physical implication of the DIPTs with a general $R$ (beyond the singularities of the R\'enyi entropies) and the relations among the DIPTs with different $R$'s are both interesting future research directions. Pertaining to the latter, we would like to present the following conjecture.

\vspace{0.2cm}
\noindent {\bf Conjecture:} The critical error rates $p_{x,z}^\star(R)$ are both monotonically increasing functions of $R$.
\vspace{0.2cm}

\noindent Here, we've assumed that for every error type, bit-flip or phase-flip, and for every R\'enyi index $R$, there is a unique DIPT in the range of error rates $p_{x,z} \in (0,1/2)$. The associated critical error rates are the subject of the conjecture above. 

A piece of evidence for this conjecture is given by the relation
\begin{align}
    p_{x,z}^\star(R=2)< p_{x,z}^\star(R\rightarrow \infty),
\end{align}
which we prove in the following. Take the case of bit-flip errors as an example. As discussed in Sec. \ref{sec:tapestry}, the behavior of $S_R$ for both $R=2$ and $R\rightarrow \infty$ are related to the non-random SM$_1$. The ``renormalized" coupling constant in the $R=2$ case is $K'={\rm arctanh}\left((1-2p_x)^2)\right)$ (see Eq. \eqref{eq:2ndRenyi_rb_clean_SM1}) and the coupling constant in the $R\rightarrow\infty$ case is $K=\frac{1}{2}\log\frac{1-p_x}{p_x}$. The fact $K>K'$ for $p_x\in \left(0,\frac{1}{2}\right)$ implies $p_{x}^\star(R=2)< p_{x}^\star(R\rightarrow \infty)$.

As we will see in Sec. \ref{sec:em-symmetric CSS}, for the class of CSS codes with an $em$ symmetry, the critical error rates $p_{x,z}^\star(R)$ for $R=2,3,\infty$ are exactly fixed (see Eq. \eqref{eq:self-dual error rate}) by the dualities (assuming the uniqueness of DIPT for each $R$ in the range $p_{x,z} \in (0,1/2)$). These exact critical error rates also obey our conjecture.  

More evidence of our conjecture will be provided when we discuss specific examples, including the 3D toric code, the X-cube model, and the 2D toric code. In particular, we will see that our conjecture is consistent with the error thresholds, i.e., $p_{x,z}^\star(R\rightarrow 1)$, of these models obtained in previous literature.

If true, our conjecture establishes a general relation between DIPTs with different R\'enyi indices $R$. It would provide an interesting method to upper bound the error threshold $p_{x,z}^\star(R\rightarrow 1)$ using $p_{x,z}^\star(R=2)$. Note that studying the limit $R\rightarrow 1$ requires the averaging over quenched disorders/randomness in the corresponding statistical models, while the $R=2$ case is essentially captured by the same statistical model (with a renormalized coupling) but without randomness. We expect the latter to be generally simpler to analyze than the former.

\subsection{Example: 3D toric code} \label{sec:3dtoric}

As a pedagogical example, let us apply our formalism to the decohered toric code model. The 2D toric code belongs to a special class of CSS codes with an $em$ symmetry, which we will focus on in Sec. \ref{sec:em-symmetric CSS}. Here, we consider the random statistical models that describe the 3D toric code with bit-flip and phase-flip errors.

The 3D toric code is a CSS code defined on a cubic lattice where the qubits, labeled by $\mu$, are located on the edges of the lattice. As shown in Fig. \ref{fig:3dtoric}, the 3D toric code has an $X$-type stabilizer $A_i[X]$ for every site $i$ and a $Z$-type stabilizer $B_p[Z]$ for every plaquette $p$. Each $A_i[X]$ is a product of the Pauli-$X$ operators on the six edges connected to the site $i$, while each $B_p[X]$ is a product of the Pauli-$Z$ operators on the four edges that belong to the plaquette $p$:
\begin{align}
   & A_i = \prod_{\mu\in \vcenter{\hbox{\includegraphics[width=1em]{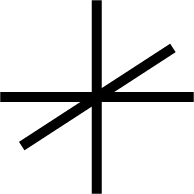}}}_i    } X_\mu,~~~~  B_p = \prod_{\mu\in \square_p}Z_\mu,
\end{align}

\begin{figure}[h!]
    \centering
    \includegraphics[width=0.9\linewidth]{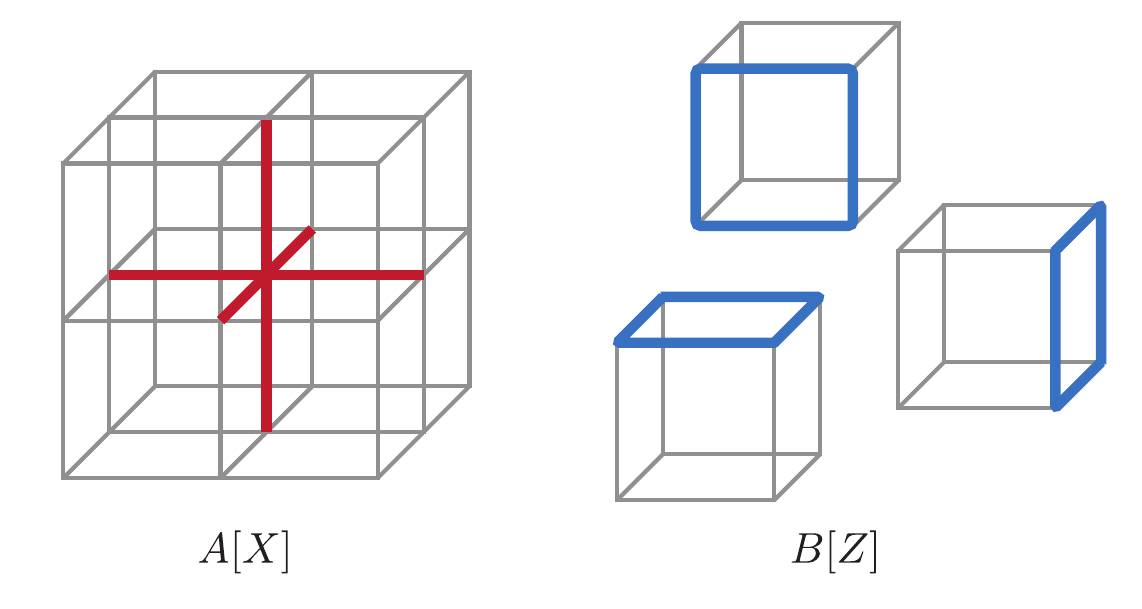}
    \caption{Stabilizers of the 3D toric code. Qubits live on the edges of the lattice. 
    An $X$-type stabilizer (red) is a product of 6  Pauli-$X$ operators on the edges connected to the same site. A $Z$-type stabilizer (blue) is a product of 4  Pauli-$Z$ operators on the edges of a plaquette.}
    \label{fig:3dtoric}
\end{figure}

Now we construct $\SM{1}$ following the recipe in Sec. \ref{sec:sm_construction}. We introduce a classical spin $\tau_i$ on each site $i$ (where the $A_i[X]$ stabilizer is located). Eq. \eqref{eq:def_Oz} tells us that the spin interaction ${\cal O}^z_\mu$ is a product of $\tau_i$'s whose associated stabilizers $A_i$ overlap with the edge $\mu$. Hence, we have 
\begin{align}
    {\cal O}^z_{\mu=\langle i,i'\rangle} = \tau_i\tau_{i'},
\end{align}
where $\mu=\langle i,i'\rangle$ means $\mu$ is the edge between the nearest-neighbor pair of sites $i$ and $i'$ (see Fig. \ref{fig:3dtoricsm} (a)). The classical Hamiltonian of the non-random $\SM1$ is then given by 
\begin{align}
    H_{\SM1} = -K\sum_{\braket{i,i'}}\tau_i\tau_{i'},
\end{align}
which is exactly the Hamitlonian of the 3D classical Ising model (with the nearest-neighbor interaction). 

For $\SM2$, we introduce a classical spin $\tilde{\tau}_p$ for each plaquette $p$. Each spin interaction $\cO^x_\mu$ is the product of the four spins $\tilde{\tau}_p$ on the four plaquettes bordering the edge $\mu$. Hence, the Hamiltonian of $\SM2$ is given by
\begin{align}
    H_{\rm SM_2} = -\tilde{K}\sum_\mu ~ \prod_{p~{\rm s. t.}~\mu\in \partial p} \tilde{\tau}_p,
\end{align}
where ``$\prod_{p~{\rm s. t.}~\mu\in \partial p}$" represents the product over the plaquette $p$ such that the edge $\mu$ belongs to the boundary $\partial p$ of the plaquette. As shown in Fig. \ref{fig:3dtoricsm} (b), each spin interaction term in $H_{\rm SM_2}$, when depicted on the dual cubic lattice, involves four classical spins on the four dual edges that form a dual plaquette. Hence, $H_{\rm SM_2}$ describes a 3D classical $\mathbb Z_2$ gauge theory on the dual lattice. It is well-known that the non-random 3D Ising model and the 3D classical $\mathbb Z_2$ gauge theory are dual to each other under an HLT duality (or Kramers-Wannier duality)\cite{wegner1971duality}.

\begin{figure}[tb]
    \centering
    \includegraphics[width=0.9\linewidth]{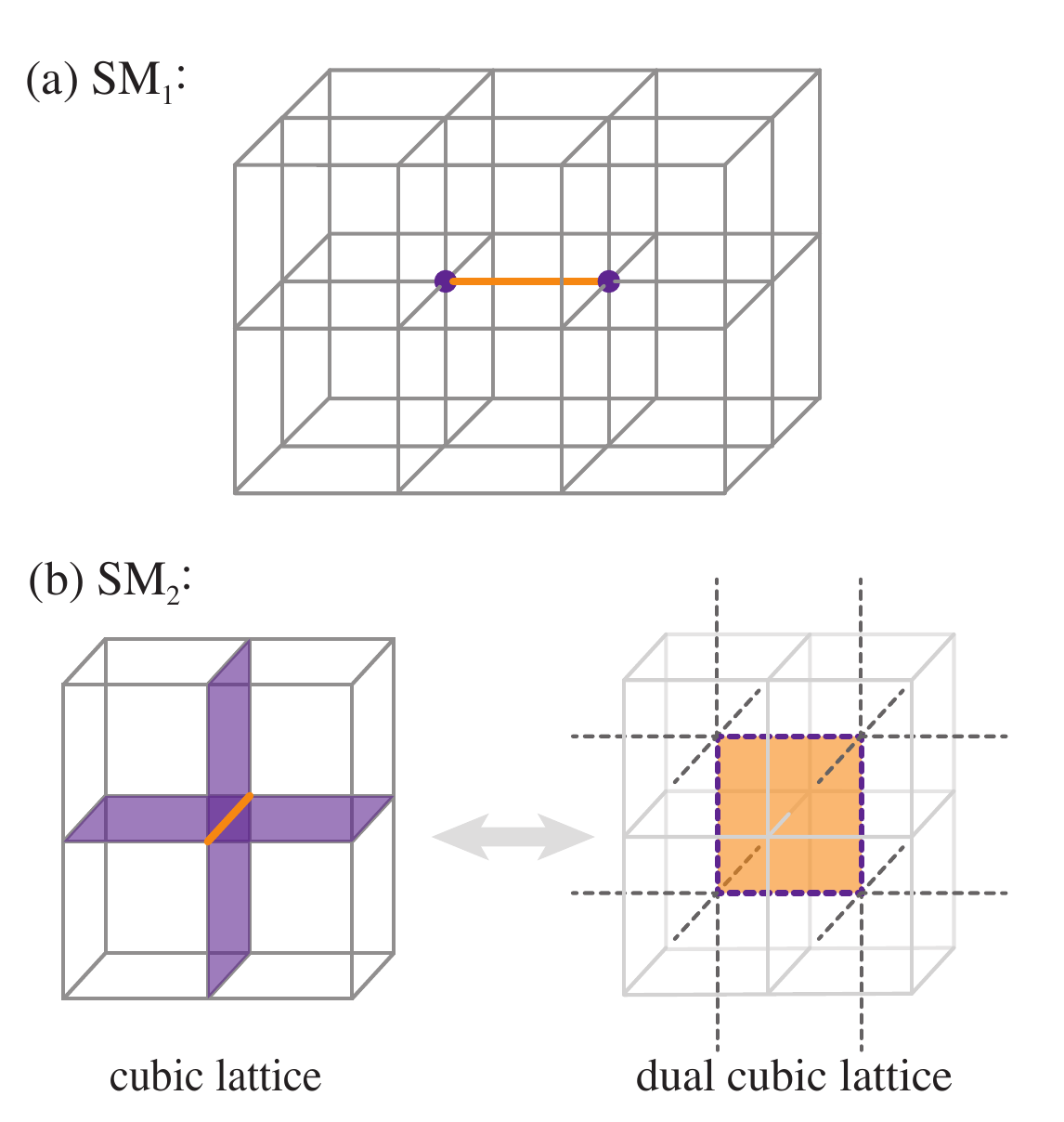}
    \caption{Statistical models associated with the 3D toric code: (a) The 3D Ising model has a spin on each site of the cubic lattice and nearest-neighbor interactions (depicted as the two spins (purple dots) ``coupled" by the orange edge). (b)  Classical 3D $\mathbb{Z}_2$ gauge theory on the dual cubic lattice has a spin per dual edge (dashed lines). Each spin interaction involves the four classical spins on the boundary of a dual plaquette (orange).}
    \label{fig:3dtoricsm}
\end{figure}

In the following, we discuss the DIPTs in the decohered 3D toric code using the multi-replica $\SM1$ and $\SM2$ with random couplings. We will focus on the cases with $R\rightarrow 1$ and $R= 2,3,\infty$. Note that the rRC in the Ising model is also known as the random-bond disorder, while the rRC in the 3D $\mathbb Z_2$ gauge theory is often referred to as the random-plaquette disorder. 

For $R\rightarrow 1$, it follows from Corollary \ref{corollary:S1_b} and \ref{corollary:S1_p} that, the von Neumann entropy $S_1(\rb)$ of the error-corrupted mixed state $\rb$ is the quenched-disorder-averaged free energy of the 3D random-bond Ising model, while $S_1(\rp)$ is quenched-disorder-averaged free energy of the 3D random-plaquette $\mathbb Z_2$ gauge theory. The respective disorder distribution satisfies the Nishimori condition. As mentioned earlier, these quenched-disordered statistical models are exactly the ones previously used to study the bit-flip and phase-flip error thresholds $p^\star_{x/z}(R\rightarrow 1)$ for the 3D toric code \cite{dennis2002topological, Preskill_3dToricCodeThreshold,ohno2004phase}. The phase diagram of the 3D random-bond Ising model (in the $R\rightarrow 1$) was studied numerically in Refs. \onlinecite{ozeki1987phase, ito1999non}. The 3D random-plaquette $\mathbb Z_2$ gauge theory was numerically investigated in Refs. \onlinecite{Preskill_3dToricCodeThreshold,ohno2004phase}. These numerical studies obtained the following error thresholds for the 3D toric code:
\begin{align}
    p_x^\star(R\rightarrow1) \approx 0.233,~~
    p_z^\star(R\rightarrow1) \approx 0.033. 
\end{align}

For $R=2$, as described around Eqs. \eqref{eq:2ndRenyi_rb_clean_SM1} and \eqref{eq:2ndRenyi_rp_clean_SM2}, the $2$-replica random statistical models that describes the 2nd R\'enyi entropy $S_2(\rho_{\rm b/p})$ can be reduced to the non-random $\SM{1,2}$ with renormalized couplings $K'={\rm arctanh}((1-2p_x)^2)$ and $\tK'={\rm arctanh}((1-2p_z)^2)$. The critical point of $\SM1$, i.e. the 3D classical Ising model in this case, has been numerically studied (see Ref. \onlinecite{ferrenberg1991critical} for example), which allows us to extract the critical bit-flip error rate $p^\star_x(2)$ for the DIPT with $R=2$.  The critical phase-flip error rate $p^\star_z(2)$ is related to the $p^\star_x(2)$ via the BPD duality Eq. \eqref{eq:BPD_duality_condition_1}. Therefore, we conclude
\begin{align}
    p_x^\star(2) \approx 0.266,~~
    p_z^\star(2) \approx 0.099,
\end{align}
for the DIPTs with $R=2$.

For $R=3$, the $3$-replica disordered statistical models that describe the behaviors of $S_3(\rho_{\rm b/p})$ can be reduced to Askin-Teller-type statistical models, namely two coupled copies of $\SM{1,2}$, after the rRC or iRC are integrated out (see general discussions in App. \ref{app:intdisorder}). For bit-flip errors, $S_3(\rho_{\rm b})$ is captured by the 3D Askin-Teller model on the cubic lattice: 
\begin{equation}
    H_{\rm AT, b} = -K''\sum_{\langle ii'\rangle}\left(\tau_i^{(1)}\tau_{i'}^{(1)} + \tau_i^{(2)}\tau_{i'}^{(2)} + \tau_i^{(1)}\tau_{i'}^{(1)}\tau_i^{(2)}\tau_{i'}^{(2)}\right).
    \label{eq:H_ATb_3dtc}
\end{equation}
Here, the coupling constant $K''$ is given by the relation $\tanh K'' + (\tanh K'')^{-1} = (1-2p_x)^{-2}+1$. $\tau_i^{(1)}$ and $\tau_i^{(2)}$ are the classical $\mathbb Z_2$ variables on the two copies of cubic lattices respectively. Previous numerical studies showed that this model exhibits a first-order transition between the paramagnetic and the ferromagnetic phase at the critical coupling $K_c^{''}  \approx 0.157$ \cite{arnold1997monte}.  This result implies that, with $R=3$, the DIPT driven by bit-flip errors is a first-order transition occurring at the critical bit-flip error rate 
\begin{align}
    p_x^\star(3) \approx 0.288.
\end{align}

For phase-flip errors, $S_3(\rho_{\rm p})$ is described by a 3D Askin-Teller-type $\mathbb Z_2$ gauge theory with the Hamiltonian 
\begin{equation}
    H_{\rm AT, p} = -\tK''\sum_\mu \left( \prod_{ \mu \in \partial p } \tilde{\tau}^{(1)}_p + \prod_{\mu \in \partial p} \tilde{\tau}^{(2)}_p + \prod_{\mu \in \partial p} \tilde{\tau}^{(1)}_p\tilde{\tau}^{(2)}_p\right). 
\end{equation}
To the best of our knowledge, this model has not been studied before. The BPD duality implies that this model is dual to the model in Eq. \eqref{eq:H_ATb_3dtc}. Therefore, the DIPT with $R=3$ driven by phase-flip errors is also first-order. It occurs at the critical phase-flip error rate 
\begin{align}
    p_z^\star(3) \approx 0.135.
\end{align}

As shown in Sec. \ref{sec:tapestry}, in the limit $R\rightarrow \infty$, the behavior of the R\'enyi entropies $S_\infty({\rho_{\rm b/p}})$ are effectively described by the non-random $\SM{1,2}$ with the couplings $K = \frac{-1}{2}\log \frac{p_x}{1-p_x}$ and $\tK = \frac{-1}{2}\log \frac{p_z}{1-p_z}$. Hence, we can obtain the critical error rates:
\begin{align}
    p_x^\star(\infty) \approx 0.391, ~~ p_z^\star(\infty) \approx 0.179. 
\end{align}

We've discussed the critical error rates of the DIPTs with several different $R$'s in the decohered 3D toric code. The results are summarized in Table \ref{tab:3dtoric}. As we can see, these critical error rates agree with our conjecture that $p_{x,z}^\star(R)$ monotonically increase as functions of $R$.

\begin{table}[]
    \centering
    \begin{tabular}{c|c|c}
        $R$ & $p_x^\star(R)$ & $p_z^\star(R)$ \\
        \hline
         1 & 0.233& 0.033\\
         2 & 0.266& 0.099 \\
         3 & 0.288& 0.135 \\
         $\infty$ & 0.391 & 0.179 
    \end{tabular}
    \caption{Critical error rates of the DIPTs in the decohered 3D toric code}
    \label{tab:3dtoric}
\end{table}

\subsection{Example: X-cube model} \label{sec:xcube}

Our second example is the X-cube model, which is a CSS code that exhibits fracton topological order \cite{Vijay_2016}. The X-cube model is defined on a 3D cubic lattice with a qubit located on each edge of the lattice.  Each $X$-stabilizer involves four nearby edges forming a cross, while each $Z$-stabilizer involves twelve edges of a unit cube (see Fig. \ref{fig:xcube}):
\begin{align}
    A_{i,a} =  \prod_{\mu\in +_{i,a}} X_\mu, ~~B_c = \prod_{\mu\in \cube_c}Z_\mu,
\end{align}
where $+_{i,a}$ denotes a cross formed by the four edges centered at site $i$, and $a=x,y,z$ labels the perpendicular direction to the cross.  $\cube_c$ denotes the cube labeled by $c$.

In $\SM1$, there should be one classical spin per $A_{i,a}$ stabilizer. Hence, we introduce three colors of classical spins $\tau_{i,a=x,y,z}$ for each site $i$ on the cubic lattice and represent them as the purple arrows in Fig. \ref{fig:xcubesm} (a). 
The ungauging procedure in Sec. \ref{sec:sm_construction} produces the following $\SM1$ Hamiltonian, which describes a tricolor Ising model with the four-spin interactions
\begin{align}
    H_{\rm SM_1} & = -K\left(\sum_{\braket{ii'}\parallel\hat{z}}\tau_{i,x}\tau_{i,y} \tau_{i',x} \tau_{i',y}  \right. 
     \nonumber \\ 
    & \left.
+\sum_{\braket{ii'}\parallel\hat{x}}\tau_{i,z}\tau_{i,y}\tau_{i',z}\tau_{i',y} + \sum_{\braket{ii'}\parallel\hat{y}}\tau_{i,x}\tau_{i,z}\tau_{i',x}\tau_{i',z} \right).
\label{eq:xcubesm1a}
\end{align}
For each nearest-neighbor pair of sites $\langle i,i'\rangle$, only the spins with their colors different from the direction of the edge $\langle i,i'\rangle$ are involved in the nearest-neighbor four-spin interaction (see Fig. \ref{fig:xcubesm} (a)).

\begin{figure}[tb]
    \centering
    \includegraphics[width=\linewidth]{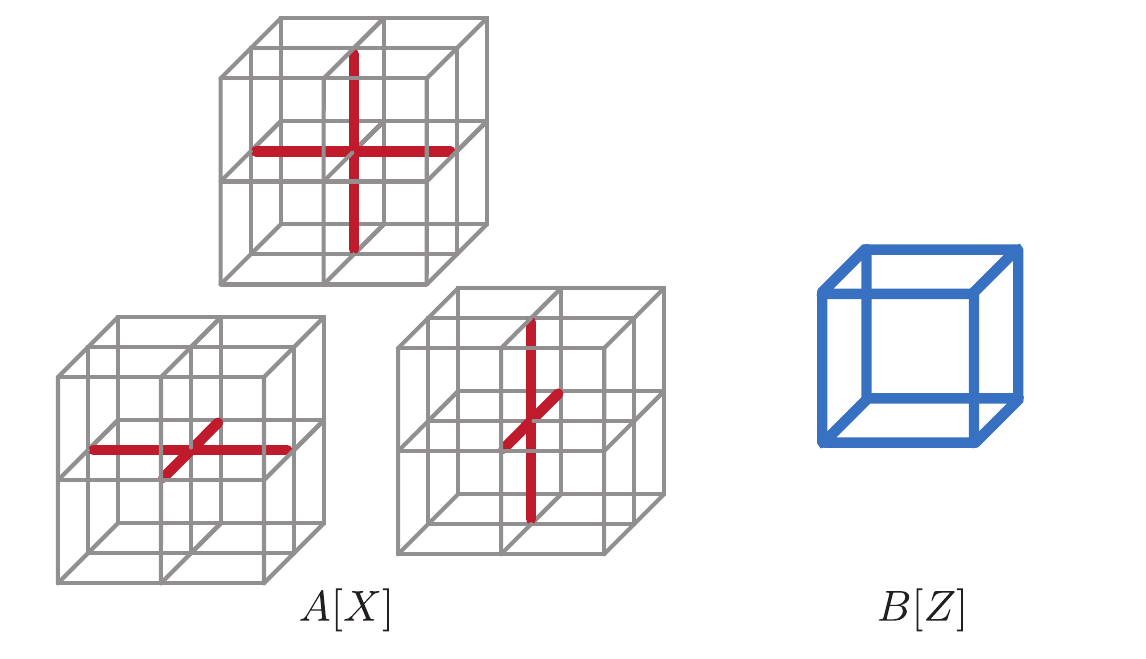}
    \caption{$X$-type stabilizers (red) and $Z$-type stabilizers (blue) of the X-cube model }
    \label{fig:xcube}
\end{figure}

This tricolor Ising model has a classical gauge symmetry that changes $\tau_{i,a}\to -\tau_{i,a}$ independently for each site $i$. A change of variables effectively imposes a gauge fixing:
\begin{align}
    \eta_{i,x} \equiv \tau_{i,y}\tau_{i,z}, ~~~~ \eta_{i,y} \equiv \tau_{i,x}\tau_{i,z},
\end{align}
where $\eta_{i,x/y}=\pm1$ are gauge invariant classical $\Z_2$ variables. Using these new variables, we can rewrite the Hamiltonian as an 
anisotropic 3D Ashkin-Teller model \cite{johnston2020four,song2022optimal},
\begin{align}
    H_{\SM1} = -K & \left(\sum_{\braket{ii'}\parallel\hat{z}}\eta_{i,x}\eta_{i,y}\eta_{i',x}\eta_{i',y} \right. \nonumber \\
    & \left.
    + \sum_{\braket{ii'}\parallel\hat{x}}\eta_{i,x}\eta_{i',x} + \sum_{\braket{ii'}\parallel\hat{y}}\eta_{i,y}\eta_{i',y} \right). \label{eq:xcubesm1b}
\end{align}
Note we can also directly obtain this statistical model if we remove the $A_{i,z}$ stabilizers from the X-cube model. Removing the $A_{i,z}$ stabilizers does not change the nature of this CSS code because $A_{i,z}$ can be expressed using the remaining stabilizers, i.e., $ A_{i,z}= A_{i,x} A_{i,y}$. 

\begin{figure}[tb]
    \centering
    \includegraphics[width=\linewidth]{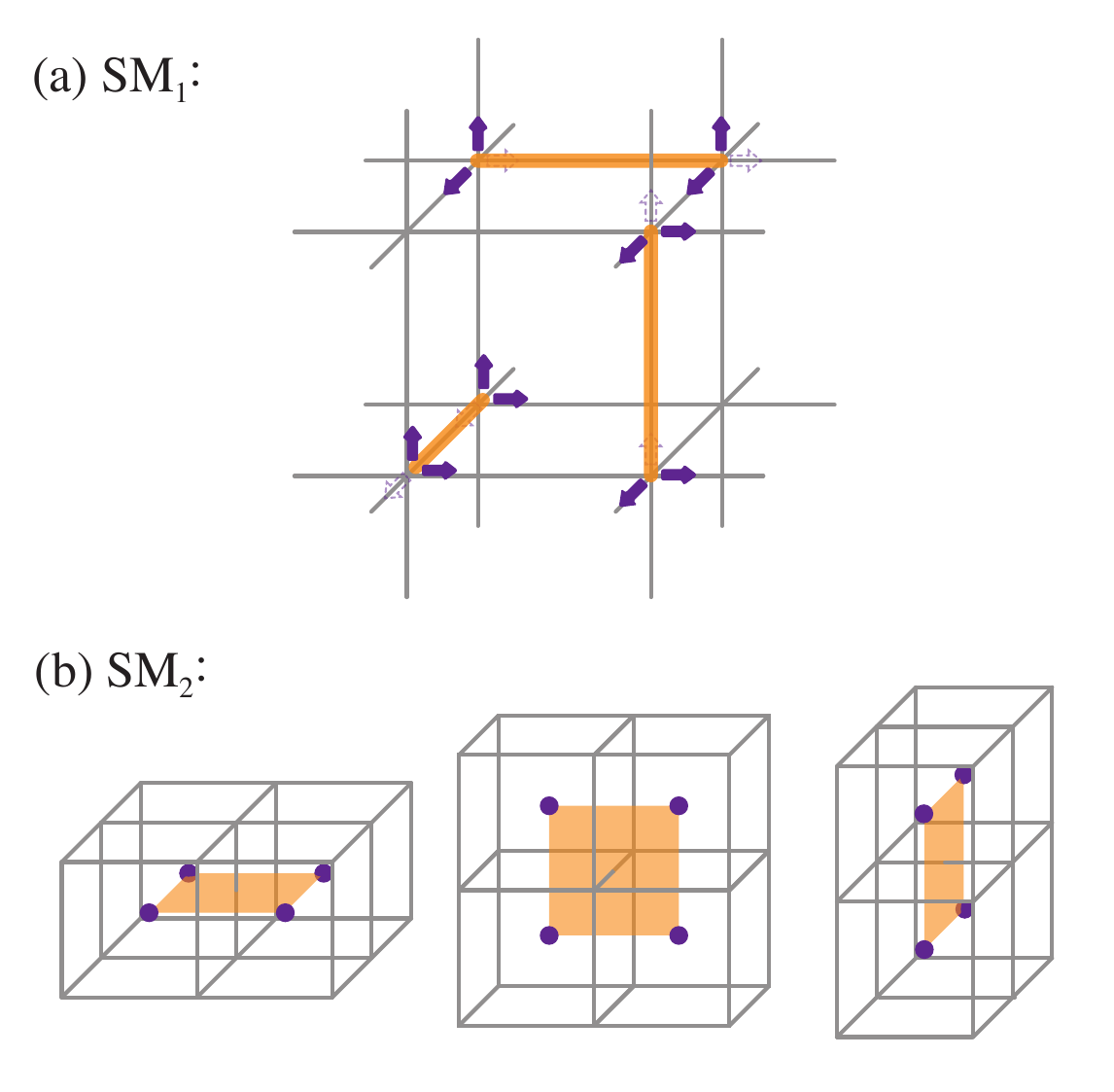}
    \caption{Two statistical models originate from the X-cube model. (a) $\SM1$ is a model with three colors of classical spins (depicted by the purple arrows along the three directions) per site on a cubic lattice. Each nearest-neighbor interaction (orange) involves the four spins on the two neighboring sites whose colors (arrow directions) are perpendicular to the edge connecting the neighboring sites. (b) $\SM2$ is a plaquette Ising model on the dual cubic lattice. There is one classical spin per dual site. Each spin interaction term involves the four spins on the four corners of a dual plaquette. }
    \label{fig:xcubesm}
\end{figure}

For $\SM2$, we introduce a classical spin $\ttau_c$ to each cube $c$. The cube label $c$ can also be viewed as the site index for the dual cubic lattice. The Hamiltonian of $\SM2$ is given by
\begin{equation}
    H_{\rm SM_2} = -\tilde{K} \sum_{\tilde{\square}} \prod_{c\in \tilde{\square}} \tilde{\tau}_c \label{eq:xcubesm2},
\end{equation}
where $\tilde{\square}$ sums over the dual plaquettes $\tilde{\square}$ on the dual lattice and ``$c\in \tilde{\square}$" indicates that $c$ is a dual site on one of the corners of $\tilde{\square}$. $H_{\rm SM_2}$ represents a 3D plaquette Ising model with a four-spin interaction for each dual plaquette (see Fig. \ref{fig:xcubesm} (b)).

We are now ready to examine the DIPTs of the X-cube model. The non-random anisotropic Ashkin-Teller model (Eq. \eqref{eq:xcubesm1b}) and the 3D plaquette Ising model (Eq.\eqref{eq:xcubesm2}) were studied numerically in Ref. \onlinecite{mueller2014nonstandard}. First-order phase transitions were found at the critical couplings $K_c \approx 0.657$ and $\tilde{K}_c \approx 0.276$. We can convert these results into the critical error rates of the DIPTs with $R=2$ and $R=\infty$. The $R \rightarrow 1$ limit was numerically studied by Ref. \onlinecite{song2022optimal} as the error thresholds of the X-cube model. We summarize these results in Table \ref{tab:xcube}. These numerically obtained critical error rates agree with our conjectured monotonicity of $p^\star_{x,z}(R)$.

\begin{table}[t]
    \centering
    \begin{tabular}{c|c|c}
        $R$ & $p_x^\star$ & $p_z^\star$\\
        \hline
         1& 0.075 & 0.152\\
         2& 0.120 & 0.241\\
         $\infty$ & 0.212 & 0.365\\
    \end{tabular}
    \caption{Critical error rates for the X-cube model. The results for $R=2,\infty$ are derived from the study of the (non-random) anisotropic Ashkin-Teller model and the plaquette Ising model in Ref. \onlinecite{mueller2014nonstandard}. The error thresholds in the limit $R\rightarrow 1$ were obtained in Ref. \onlinecite{song2022optimal}}.
    \label{tab:xcube}
\end{table}

The $R=3$ case requires the consideration of two coupled copies of plaquette Ising models on the dual cubic lattice:
\begin{equation}
H_{\rm AT,p} = -\tilde{K}''\sum_{\tilde\square} \left(\prod_{c\in \tilde\square} \tilde{\tau}^{(1)}_c + \prod_{c\in \tilde\square} \tilde{\tau}^{(2)}_c + \prod_{c\in \tilde\square} \tilde{\tau}^{(1)}_c \tilde{\tau}^{(2)}_c    \right)
\end{equation}
and its BPD dual. To our knowledge, this model has not been studied before.

\section{Decohered CSS code with electric-magnetic symmetry}
\label{sec:em-symmetric CSS}
\begin{figure*}[t!]
    \centering
    \includegraphics[width=.9\linewidth]{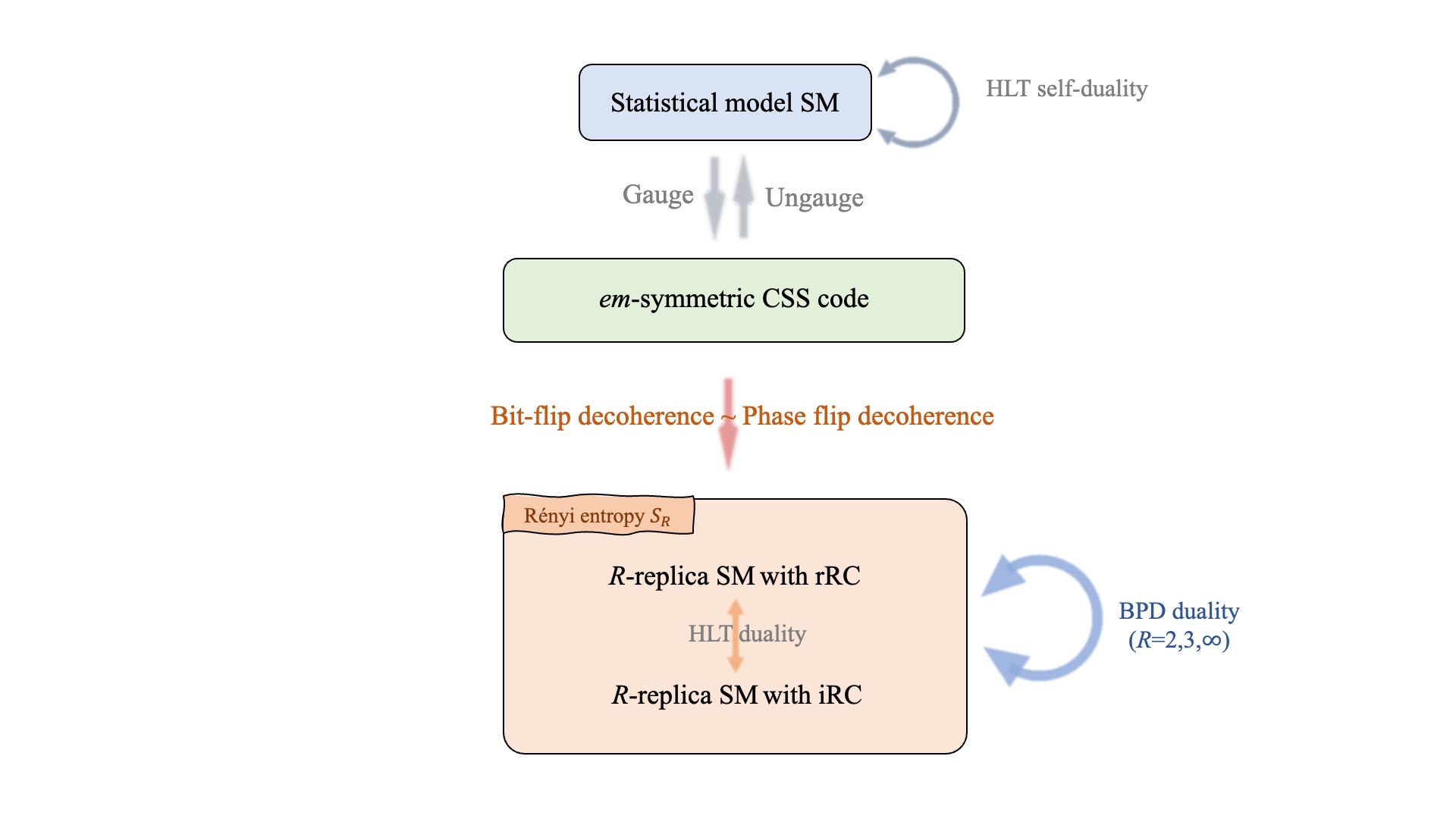}
    \caption{The tapestry of dualities of an $em$-symmetric CSS code. The ungauging procedure leads to SM with a self-duality. Under bit-flip or phase-flip decoherence, $R$-replica SM with rRC and $R$-replica SM with iRC emerge as the description of the system's R\'enyi entropy. These random statistical models are related by HLT dualities for a general $R$. For $R=2,3,\infty$, the BPD dualities (combined with the $em$ symmetry) relate different values of error rates.
    }
    \label{fig:self-duality}
\end{figure*}

In this section, we study the consequences of the dualities for the CSS codes with an electric-magnetic ($em$) symmetry, a symmetry relating the $X$-type and $Z$-type stabilizers. We will see that the $em$ symmetry effectively folds the tapestry of dualities in Fig. \ref{fig:duality1} in half, resulting in a new tapestry shown in Fig. \ref{fig:self-duality}.  The new tapestry allows us to pin down the {\it super-universal} self-dual error rates for the R\'enyi entropies $S_R(\rho_{\rm b/p})$ with $R=2,3,\infty$. These super-universal self-dual error rates must coincide with the critical error rates $p^\star_{x/z}$ of the DIPTs if there is a unique DIPT for each $R$. We will also discuss concrete examples of $em$-symmetric CSS codes, including the 2D toric code and Haah's code in 3D.  

\subsection{Tapestry of dualities and self-dual error rates in $em$-symmetric CSS codes}

Let's first define the $em$ symmetry. A CSS code is $em$-symmetric if there is a unitary transformation ${\cal U}_{em}$ such that 
\begin{align}
    & {\cal U}_{em} A_i[X] {\cal U}_{em}^\dag= B_{j=f(i)}[Z], \nonumber \\
    &{\cal U}_{em} B_j[X] {\cal U}_{em}^\dag= A_{i=g(j)}[Z].
    \label{eq:em_symmetry_U}
\end{align}
Here, ${\cal U}_{em}$ is a product of the Hadamard gate $\frac{1}{\sqrt{2}} \left(\begin{smallmatrix}
    1 & 1\\1& -1
\end{smallmatrix}\right)$ on each qubit, which exchanges the Pauli-$X$ and the Pauli-$Z$ operators, and a spatial action that permutes the locations of the qubits. The functions $f$ and $g$ are bijections between the locations of $X$- and $Z$-type stabilizers induced by the spatial action in ${\cal U}_{em}$.

Since the $em$ symmetry maps between the $X$-type and $Z$-type stabilizers, the two statistical models $\SM{1,2}$ produced by the ungauging procedure in Sec. \ref{sec:sm_construction} become identical (up to some spatial rearrangement of the classical spins) in an $em$-symmetric CSS code. Hence, we will not distinguish $\SM{1,2}$ in this section and call them both SM. The original HLT duality between the non-random $\SM1$ and $\SM2$ becomes an HLT self-duality of the non-random SM.

The $em$ symmetry maps the bit-flip errors and the phase-flip errors into each other. Therefore, the decoherence effects caused by the two error types are effectively identical. We remark that the $em$ symmetry is an intrinsic property of the CSS code under the current discussion. We do {\it not} require the error model to respect this symmetry. As followed from Theorem \ref{thm:sm_b} and \ref{thm:sm_p}, the $R$th R\'enyi entropies $S_R({\rho}_{\rm b/p})$ are described by both the $R$-replica SM with rRC and $R$-replica SM with iRC. The two types of the $R$-replica random statistical models are dual to each other under the HLT dualities. 

The BPD duality leads to profound consequences for $em$-symmetric CSS codes. Combining the $em$ symmetry action and the BPD duality, both of which map between bit-flip errors and phase-flip errors, we conclude that the random statistical models for $S_R(\rb)$ with $R=2,3,\infty$ are self-dual, and so are those for $S_R(\rp)$. 
\begin{tcolorbox}
\begin{corollary} \label{thm:selfduality}   
{\rm For an $em$-symmetric CSS code with decoherence, the R\'enyi entropies $S_R (\rho_{\rm b/p})$ with $R = 2,3,\infty$ are described by $R$-replica random statistical models with self-dualities. The dualities map between the two error rates $0<p,\tilde{p} <\frac{1}{2}$ related by:
}
\begin{align}
    \left[(1-p)^R + p^R\right]\left[(1-\tilde{p})^R + \tilde{p}^R\right] = \frac{1}{2^{R-1}}. \label{eq:BPD_self_duality}
\end{align}
{\rm Here, both of the error rates $p$ and $\tilde p$ can be either the bit-flip error rates or the phase-flip error rates, depending on the error type considered. }
\end{corollary}
\end{tcolorbox}
\noindent 
Similar to Theorem. \ref{thm:BPD}, the duality relations Eq. \eqref{eq:BPD_self_duality} are dualities between weak and strong decoherence strength. By solving the self-dual condition $p = \tilde{p}$ together with Eq. \eqref{eq:BPD_self_duality}, we obtain the super-universal self-dual error rates (listed below) shared by {\it all} $em$-symmetric CSS codes. If we further assume that there is a unique DIPT for each $R$ (within the range $p\in (0,1/2)$), the self-dual error rates must coincide with the critical error rates of the corresponding DIPTs:
\begin{tcolorbox}
\begin{align}
    p^\star(R)=
    \begin{dcases}
     \frac{1}{2}\left(1-\sqrt{\sqrt{2} -1}\right)\approx 0.178 &{\rm for~} R=2,\\
     \frac{1}{6}(3-\sqrt{3})\approx 0.211 &{\rm for~} R=3, \\
     \frac{1}{2}(2-\sqrt{2})\approx 0.293 &{\rm for~} R \rightarrow \infty.
    \end{dcases}
    \label{eq:self-dual error rate}
\end{align}
\end{tcolorbox}

\noindent It is worth noting that these super-universal critical error rates are consistent with the conjectured monotonicity of $p^\star_{x/z} (R)$. Also, we called these values super-universal because they encompass DIPTs of different universality classes and in different dimensions, as demonstrated by the examples discussed below.  

If there is more than one DIPT for $p\in [0,1/2]$, the values given Eq. \eqref{eq:self-dual error rate} should be viewed as the self-dual error rates. The critical error rates that differ from the self-dual error rates must form pairs. For every pair, there must be a critical error rate smaller than the self-dual value.   Therefore, in the case of an $em$-symmetric CSS code with multiple DIPTs for $p\in [0,1/2]$ and at a given $R$, the self-dual error rate in Eq. \eqref{eq:self-dual error rate} upper bounds the critical error rate of the first DIPT encountered as the error rate increases from 0. 

The tapestry summarizing all the dualities of an $em$-symmetric CSS code is shown in Fig. \ref{fig:self-duality}. Conceptually, this tapestry is Fig. \ref{fig:duality1} folded in half in the middle by the $em$ symmetry.

In subsequent subsections, as demonstrations, we discuss two concrete models where our general results of $em$-symmetric CSS codes apply.

\subsection{Example: 2D toric code}
\label{sec:2dtc}

The 2D toric code is a familiar example of an $em$-symmetric CSS code. It is defined on a 2D square lattice with a qubit on each edge. As shown in Fig. \ref{fig:2dtoric}, each $X$-type stabilizer is a product of four Pauli-$X$ operators on the four edges forming a cross while each $Z$-type stabilizer is a product of four Pauli-$Z$ operators on the four edges forming a plaquette:
\begin{align}
    A_i = \prod_{\mu\in +_i} X_\mu ~~ B_p = \prod_{\mu \in \square_p} Z_\mu,
\end{align}
where $+_i$ labels the cross centered at the site $i$ and $\square_p$ is the plaquette labeled by $p$.

It is straightforward to visualize the $em$ symmetry by comparing the pictorial representations of the stabilizers on the square lattice and their representations on the dual square lattice (see Fig. \ref{fig:2dtoric}). This $em$ symmetry is generated by the unitary operator ${\cal U}_{em}$ that interchanges the $X$ and $Z$ operators on each qubit and translates the system both in the horizontal and vertical directions by half of the lattice spacing.

\begin{figure}[tb]
    \centering
    \includegraphics[width=\linewidth]{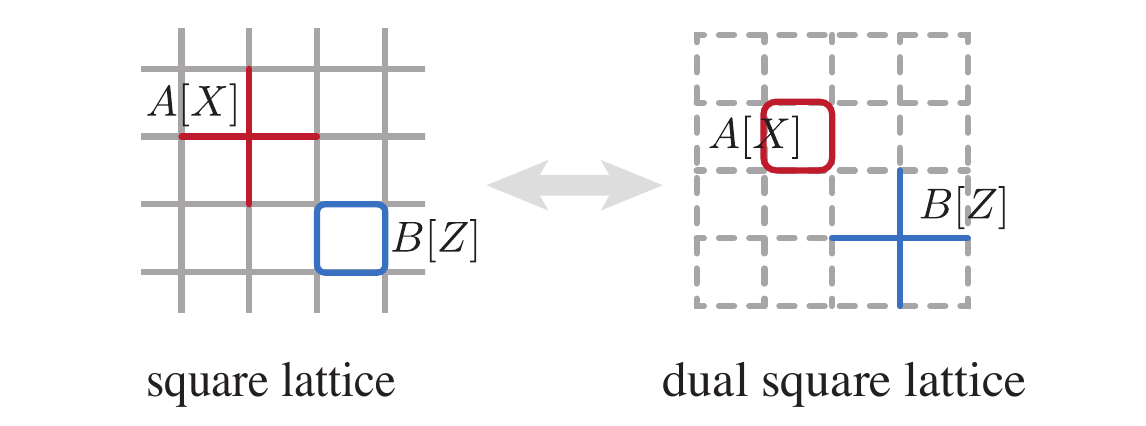}
    \caption{The $em$ symmetry of the 2D toric code is evident from the pictorial representations of the $X$- and $Z$-type stabilizers on the square lattice and its dual. The unitary operator ${\cal U}_{em}$ is a product of Hadamard gates on all sites and a translation action by half of a lattice spacing in both the horizontal and vertical directions. 
    }
    \label{fig:2dtoric}
\end{figure}

For the 2D toric code, SM is the 2D classical Ising model with the nearest-neighbor coupling:
\begin{align}
    H_{\SM{}} = -K\sum_{\braket{i,i'}}\tau_i\tau_{i'},
\end{align}
where the classical spins $\tau_i$ are located on the site of the square lattice. The HLT self-duality of this model (without randomness) is the well-known Kramers-Wannier duality. 

Now, we discuss the R\'enyi entropies $S_R(\rho_{\rm b/p})$ in the decohered 2D toric code. The $R\rightarrow 1$ limit is described by SM with quenched random-coupling disorder, which exactly recovers the random-bond Ising model introduced in Ref. \onlinecite{dennis2002topological} for the study of the error thresholds $p^\star(R\rightarrow 1)$. The cases with $R=2$ and $R = \infty$ can be effectively captured by the non-random SM (see App. \ref{app:intdisorder}). Hence, there is a unique DIPT for $R=2,\infty$, whose universality class is equivalent to the Ising critical point in 2D. The corresponding critical error rates, fixed by the self-dualities, are given by Eq. \eqref{eq:self-dual error rate}. 

For $R=3$, integrating out the randomness in the statistical model yields the 2D square-lattice Ashkin-Teller model, whose Hamiltonian is given by
\begin{equation}
    H_{\rm AT} = -K''\sum_{\langle ii'\rangle}\left(\tau_i^{(1)}\tau_{i'}^{(1)} + \tau_i^{(2)}\tau_{i'}^{(2)} + \tau_i^{(1)}\tau_{i'}^{(1)}\tau_i^{(2)}\tau_{i'}^{(2)}\right).
    \label{eq:H_ATb_2dtc}
\end{equation}
This Ashkin-Teller model tuned by the coupling constant $K''$ has a unique second-order transition \cite{kohmoto1981hamiltonian}, corresponding to $p^\star(3)$. 

For $R=4,5,6$, Ref. \onlinecite{FanBaoTopoMemory} numerically simulated the corresponding multi-replica statistical models and found a unique phase transition in each case. The values of $p^\star(R)$ with $R\rightarrow 1$ and $R=2,3,4,5,6,\infty$ are consistent with the conjectured monotonicity of $p^\star(R)$ as a function of $R$.

\subsection{Example: Haah's code}
Haah's code \cite{haah2011local} provides an example of 3D $em$-symmetric CSS code. It is defined on a 3D cubic lattice with two qubits per site. There is one $X$-type and one $Z$-type stabilizer per cube. They are pictorially represented in Fig. \ref{fig:haah} (a). Each $A[X]$ stabilizer is a product of eight Pauli-$X$ operators on the qubits in red, while each $B[Z]$ stabilizer is a product of eight Pauli-$Z$ operators on the qubits in blue. The $em$ symmetry action ${\cal U}_{em}$ of Haah's code is the product of (1) a Hadamard gate on each qubit, (2) a spatial inversion $\Vec{r} \rightarrow -\Vec{r}$, and (3) a swap of the two qubits on each site.

\begin{figure}[tb]
    \centering
    \includegraphics[width=0.95\linewidth]{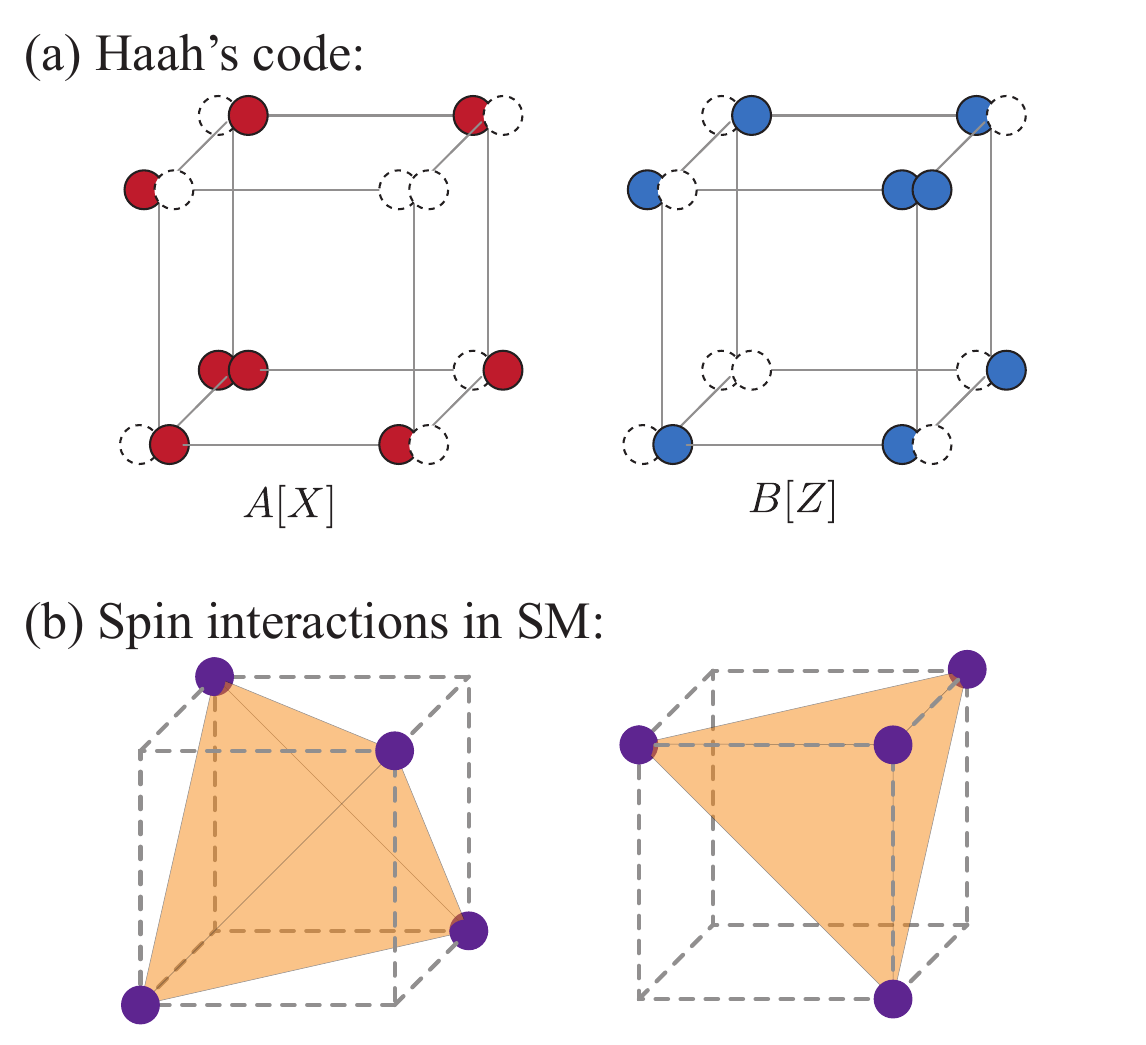}
    \caption{ (a) Haah's code is defined on a 3D cubic lattice with two qubits (depicted as the red, blue, and dashed circles) per site. For each cube, the $X$-type stabilizer is the product of the Pauli-$X$ operators on the 8 red qubits, and the $Z$-type stabilizer is the product of the Pauli-$Z$ operators on the 8 blue qubits.
    (b) In SM, which originates from Haah's code, there are two types of four-spin interactions in each unit cell. Each type is a product of the 4 classical spins (purple dots) on the corners of the corresponding tetrahedron (orange).}
    \label{fig:haah}
\end{figure}

Apply the ungauging procedure in Sec. \ref{sec:sm_construction}, we obtain a 3D SM with a classical $\mathbb Z_2$ spin per site on the dual cubic lattice. The Hamiltonian of this SM contains two types of four-spin interaction per unit cell, each associated with a tetrahedron inside a cube. They are pictorially represented in Fig. \ref{fig:haah}. This SM is also known as the ``fractal Ising model". It was introduced in Ref. \onlinecite{Vijay_2016} as the classical model that produces Haah's code after gauging. The (non-random) fractal Ising model was recently shown through numerical studies to exhibit a unique first-order transition \cite{canossa2023exotic}.

Following the recipe in Sec. \ref{sec:sm_construction}, the R\'enyi entropies $S_R(\rho_{\rm b/p})$ are mapped to the $R$-replica random fractal Ising models. As discussed earlier, the cases with $R=2$ and $R\rightarrow \infty$ can be effectively reduced to the non-random fractal Ising model, which has a unique first-order transition. Therefore, the DIPTs with $R=2$ and $R= \infty$ are first-order transitions occurring at the self-dual error rates listed in Eq. \eqref{eq:self-dual error rate}. For $R=3$, the $3$-replica random fractal Ising model can be reduced to an Ashkin-Teller-type model after integrating out the randomness (see App. \ref{app:intdisorder}). The phase diagram of this model and the nature of the self-dual error rate $p^\star(3)$ are both interesting problems that we leave for future investigation.

\section{General stabilizer codes under decoherence} \label{sec:GPN}

The mapping from the R\'enyi entropies to statistical models can be generalized to general stabilizer codes subject to a general Pauli noise channel (to be explained below). As we will see, we can associate a statistical model SM with a general stabilizer code without the CSS structure. The non-random version of SM has an HLT self-duality. The $R$-th R\'enyi entropy of the decohered stabilizer code is mapped to both the $R$-replica SM with rRC and the $R$-replica SM with iRC, which are related by an HLT duality. For $R = 2$ and $R\rightarrow \infty$, there is a general-Pauli-noise (GPN) duality that maps between two different Pauli noise channels with different sets of error rates. We obtain the self-dual conditions of the Pauli noise channels and discuss their relations to DIPTs. We also include the Chamon model \cite{chamon2005quantum} as a concrete example to demonstrate our general construction.

\subsection{Stabilizer codes and general Pauli noise}

Let's start our discussion with the basics of stabilizer codes and the model for
general Pauli noise. We still focus on the stabilizer code defined on a system of qubits (with each carrying a 2-fold local Hilbert space). 

For a general stabilizer code, each stabilizer $A_J[X,Z]$ can be written as a product of the Pauli-$X$ and $Z$ operators. In this convention, any Pauli-$Y$ operator in $A_J[X,Z]$ is decomposed into $\i X Z$. The subscript $J$ of $A_J[X,Z]$ labels the stabilizer's center location (and the species of stabilizers if applicable). Each stabilizer $A_J[X,Z]$ can be represented by a pair of $\mathbb Z_2$ vectors $(\sfa_J,\sfb_J)$ with $\sfa_J,\sfb_J \in {\cal V}$:
\begin{align}
    A_J[X,Z] = (\i)^{\sfa_J \cdot \sfb_J} \prod_\mu (X_\mu)^{(\sfa_J)_\mu} \prod_\mu (Z_\mu)^{(\sfb_J)_\mu}. \label{eq:z2pauli}
\end{align}
Recall ${\cal V} = {\mathbb Z}_2^N$, the ${\mathbb Z}_2$ vector space  associated with a system of $N$ qubits. The prefactor $(\i)^{\sfa_J \cdot \sfb_J}$ ensures that $A_J[X,Z]^2 = \mathds{1}$. The fact that all the stabilizers commute with each other is equivalent to the condition
\begin{align}
    \sfa_J\cdot \sfb_{J'} + \sfa_{J'}\cdot \sfb_{J} = 0, ~~~~\forall J, J'. \label{eq:bsp}
\end{align}
Note that the ``$\cdot$" above is the dot product in the ${\mathbb Z}_2$ vector space $\cal V$.  This ${\mathbb Z}_2$-vector-based representation of the stabilizers is called ``binary symplectic representation" in the literature \cite{eczoo_qubit_stabilizer}. Eq. \eqref{eq:bsp} defines the binary symplectic inner product between $(\sfa_J,\sfb_J)$ and $(\sfa_{J'},\sfb_{J'})$.  The binary symplectic representation is applicable to any Pauli strings. In other words, any product of Pauli operators can be represented by a pair of $\mathbb Z_2$ vectors $({\sfe}, {\sff}) \in {\cal V}\oplus {\cal V}$.

In this section, we consider the general Pauli noise that induces the Pauli noise channel $\cal E = \otimes_\mu \cal E_\mu $, a product of the local Pauli noise channel $\cal E_\mu$ on each qubit $\mu$: 
\begin{align}
    \mathcal{E}_\mu[\rho_0] =  & (1-p_x-p_y-p_z)\rho_0 + p_x X_\mu \rho_0 X_\mu 
    \nonumber \\
    &+ p_y Y_\mu\rho_0 Y_\mu + p_z Z_\mu\rho_0 Z_\mu.
\end{align}
where $p_{x,y,z}>0$ are the probability/error rates of the $X$-, $Y$-, $Z$-type errors on each qubit.  The general Pauli noise channel $\cal E$ includes the two quantum channels ${\cal N}_{x,z}$ induced by the bit-flip and the phase-flip errors and studied in Sec. \ref{sec:CSS} and \ref{sec:em-symmetric CSS}. $\cal E$ reduces to the depolarizing channel when we set $p_x = p_y = p_z$. For the following analysis, we make a technical assumption that $(1-p_x-p_y-p_z)>p_{x,y,z}$ without the loss of generality. For a general Pauli noise channel, one can always combine it with an appropriate choice of unitary global action $\prod_\mu X_\mu$, $\prod_\mu Y_\mu$, or $\prod_\mu Z_\mu$ to construct a new quantum channel that satisfies the condition $(1-p_x-p_y-p_z)>p_{x,y,z}$. Since this global action does not change the system's entropy, studying the decoherence effect of the new channel is equivalent to studying the original channel. 

In what follows, we investigate the R\'enyi entropies $S_R(\rho_{\rm m})$ of the error-corrupted mixed state $\rho_{\rm m} = {\cal E}(\rho_0)$ where $\rho_0$ is a pure logical state of the stabilizer code. Like the case of CSS codes, we will construct the statistical models associated with $S_R(\rho_{\rm m})$ and study their dualities.

\subsection{Statistical models for general stabilizer codes under general Pauli decoherence} \label{sec:recipe2}

To motivate the statistical model associated with a general decohered stabilizer code $\cal C$, it is useful first to consider the setting of CSS code with the general Pauli noise channel $\cal E$ defined above. In the case where the decoherence channel is generated by independent bit-flip and phase-flip errors, namely when ${\cal E} = {\cal N}_x \circ {\cal N}_z$, we have 
\begin{align}
    p_x = (1-q_z) q_x,~~p_y = q_x q_z,~~p_z = (1-q_x) q_z, 
\end{align}
where $q_{x/z}$ are the respective rates of the bit-flip and phase-flip errors in ${\cal N}_{x/z}$. The $Y$-type error, now occurring with probability $p_y = q_x q_z$, is the consequence of a bit-flip and a phase-flip error occurring on the same qubit. As discussed in Sec. \ref{sec:CSS}, the effects of independent bit-flip and phase-flip errors decouple on a CSS code. Therefore, the entropy of the entire system can be captured by two decoupled (random and multi-replica) statistical models $\SM1$ and $\SM2$. The former lives on the lattice formed by the centers of the $X$-type stabilizers, while the latter lives on the lattice formed by the $Z$-type stabilizers. 

It turns out that, when $p_y$ deviates from $q_x q_z$, namely when the $Y$-type errors start to become independent from the bit-flip and phase-flip errors, $\SM1$ couples to $\SM2$ forming a single statistical model, which will be called $\SM{}$ in the following, living on the lattice formed by the centers of {\it all} stabilizers. The duality between $\SM1$ and $\SM2$ becomes the self-duality of SM. In fact, as we show below, even for a general stabilizer code $\cal C$ with Pauli decoherence, there is a corresponding self-dual SM whose $R$-replica random versions describe the R\'enyi entropies of the error-corrupted mixed state.

Like Sec. \ref{sec:CSS}, we first introduce the non-random version of the statistical model SM associated with the general stabilizer code $\cal C$. 
We assign one classical $\mathbb Z_2$ spin $\tau_J$ to each stabilizer $A_J[X,Z]$. The Hamiltonian of SM is given by
\begin{align}
    H_{{\rm SM}} = -\sum_{\mu}K_x {\cal O}^x_{\mu}[\tau]+K_y {\cal O}^y_{\mu}[\tau] + K_z {\cal O}^z_{\mu}[\tau],
\end{align}
which contains three types of spin interactions
\begin{align}
   {\cal O}^z_\mu[\tau] = \prod_J (\tau_J)^{(\mathsf{a}_J)_{\mu}}, & ~~~
    {\cal O}^x_\mu[\tau] = \prod_J (\tau_J)^{(\mathsf{b}_J)_\mu}, \nonumber \\ 
    {\cal O}^y_\mu[\tau] = \prod_J &(\tau_J)^{(\mathsf{a}_J+\mathsf{b}_J)_\mu},
    \label{eq:def_Oxyz}
\end{align}
with coupling constants $K_{x,y,z}$. These spin interactions are the generalizations of Eqs. \eqref{eq:def_Oz}
and \eqref{eq:def_Ox}. They satisfy the relation 
\begin{align}    \mathcal{O}^x_\mu[\tau]\mathcal{O}^y_\mu[\tau]\mathcal{O}^z_\mu[\tau] = 1.
\end{align}

For the non-random SM, the comparison of the high-temperature and low-temperature expansion of the partition function $Z_{\rm SM} = \sum_{\{\tau_J = \pm 1 \}} e^{-H_{\rm SM}}$ yields a self-duality of SM that maps the coupling constants $\mathbf{K}=(K_x, K_y, K_z)$ to $\mathbf{\tilde{K}}=(\tK_x, \tK_y, \tK_z)$,
\begin{align}
    e^{-2(\tilde{K}_x+\tilde{K}_y)} = \frac{\tanh K_z + \tanh K_x\tanh K_y }{1+ \tanh K_x\tanh K_y\tanh K_z}, \nonumber\\
    e^{-2(\tilde{K}_z+\tilde{K}_x)} = \frac{\tanh K_y + \tanh K_z\tanh K_x }{1+ \tanh K_x\tanh K_y\tanh K_z}, \label{eq:noncssKW} \\
    e^{-2(\tilde{K}_y+\tilde{K}_z)} = \frac{\tanh K_x + \tanh K_y\tanh K_z }{1+ \tanh K_x\tanh K_y\tanh K_z}. \nonumber 
\end{align}
The detailed derivation of this duality in presented in App. \ref{app:ncss}.

Under the Pauli noise channel, the error-corrupted mixed state of the stabilizer code is denoted as $\rho_{\rm m}$. We prove the following theorem pertaining to the $R$-th R\'enyi entropy $S_R(\rho_{\rm m}) = \frac{1}{1-R} \log \Tr(\rho_{\rm m}^R)$ of the decohered stabilizer code.

\begin{widetext}
\begin{tcolorbox}[width = \textwidth]
\begin{theorem}\label{thm:noncss}
{\rm For the general stabilizer code in the infinite system limit, $\Tr\rho_{\rm{m}}^R$ of the error-corrupted mixed state $\rho_{\rm{m}}$ generated by the Pauli noise channel is proportional to the partition function of the $R$-replica SM with the rRC. $\Tr\rho_{\rm{m}}^R$ is also proportional to the partition function of the $R$-replica SM with iRC: }
    \begin{align}
        {\rm Tr}(\rho_{\rm m}^R) &\propto \sum_{E\in{\cal V}\oplus {\cal V}} \Big(Z_{\rm SM}(\mathbf{K}, E)\Big)^R\\
        &\propto \sum_{E\in{\cal V}\oplus {\cal V}} \Big(W_{\rm SM}(\mathbf{\tK}, E)\Big)^R,
    \end{align}
{\rm with the coupling constants $\bf{K}$ and $\bf{\tilde{K}}$ given by the error rates $p_{x,y,z}$ via}
    \begin{align}
        e^{-2(K_x+K_y)} = \frac{\tanh \tilde{K}_z+\tanh \tilde{K}_x\tanh \tilde{K}_y}{1+\tanh \tilde{K}_x \tanh \tilde{K}_y \tanh \tilde{K}_z} = \frac{p_z}{1-p_x-p_y-p_z}, \nonumber\\
        e^{-2(K_y+K_z)} = \frac{\tanh \tilde{K}_x+\tanh \tilde{K}_y\tanh \tilde{K}_z}{1+\tanh \tilde{K}_x \tanh \tilde{K}_y \tanh \tilde{K}_z} = \frac{p_x}{1-p_x-p_y-p_z},   \label{eq:ncssrelation} \\
        e^{-2(K_x+K_z)} = \frac{\tanh \tilde{K}_y+\tanh \tilde{K}_x\tanh \tilde{K}_z}{1+\tanh \tilde{K}_x \tanh \tilde{K}_y \tanh \tilde{K}_z} = \frac{p_y}{1-p_x-p_y-p_z}. \nonumber
    \end{align}
\end{theorem}
\end{tcolorbox}
\noindent Here, $E \equiv (\mathsf{e}, \mathsf{f}) \in \cV\oplus \cV $ is the pair of $\mathbb Z_2$ vectors that represents the random couplings in the statistical models. We have defined the partition function of SM with the rRC pattern $E$ (and the coupling constants $\bf K$) as
    \begin{align}
        Z_{\rm SM}(\mathbf{K}, E) = \sum_{\{\tau_J=\pm1\}}\exp \left(\sum_\mu \left(K_x (-1)^{\sff_\mu} \mathcal{O}^x_\mu[\tau]+K_y (-1)^{(\sfe+\sff)_\mu} \mathcal{O}^y_\mu[\tau]+K_z (-1)^{\sfe_\mu} \mathcal{O}^z_\mu[\tau]\right)\right),
    \end{align}
and the SM partition function with the iRC pattern $E$ (and the coupling constants $\bf \tilde K$) as
    \begin{align}
        W_{\rm SM}(\mathbf{\tK}, E) & = \sum_{\{\tau_J=\pm1\}}\exp(-H_{\rm SM}(\mathbf{\tK})) \prod_{\mu} \Big(\cO^x_\mu[\tau]\Big)^{\sfe_\mu} \Big(\cO^z_\mu[\tau]\Big)^{\sff_\mu}. \nonumber \\        
        & = \sum_{\{\tau_J=\pm1\}}\exp\left(-H_{\rm SM} (\mathbf{\tK}) -  \frac{\i \pi}{2} \sum_\mu \mathsf{e}_\mu \Big(\cO^x_\mu[\tau] - 1\Big) -  \frac{\i \pi}{2} \sum_\mu \mathsf{f}_\mu \Big(\cO^z_\mu[\tau] - 1\Big)    \right).
    \end{align}
\end{widetext}

The two types of $R$-replica random SM are dual to each other under the HLT duality. The proof of Theorem \ref{thm:noncss} (and the HLT duality) is provided in App. \ref{app:ncss}.

The phase diagram of the decohered stabilizer code $\cal C$ has three tuning parameters $p_{x,y,z}$, which can potentially lead to even richer physics than the case of CSS codes with only bit-flip and phase-flip errors.  The DIPTs of the decohered stabilizer code $\cal C$, indicated by the singularities in the R\'enyi entropy $S_R({\rho_{\rm m}})$, corresponds to the phase transitions in these $R$-replica random statistical models.

Taking the limit of $R\rightarrow 1$ of Theorem \ref{thm:noncss}, we obtain the relation between the von Neumann entropy $S_1(\rho_{\rm m})$ and the quenched-disorder-average partition function of SM with random couplings (derivation summarized in App. \ref{app:ncss}):
\begin{tcolorbox}
\begin{corollary} \label{corollary:S1_m}
 {\rm The von Neumann entropy of the error-corrupted state $\rho_{\rm m}$ is given by the quenched-disorder-averaged free energy of SM with rRC (up to an unimportant additive constant):}
\begin{align}
    S_1(\rho_{\rm m}) = -\sum_{E\in {\cal V}\oplus {\cal V}} P(E) \log(Z_{\rm SM}(\mathbf{K}, E)),
    \label{eq:S1_m}
\end{align}
{\rm with the generalized Nishimori condition Eq. \eqref{eq:ncssrelation} relating the probability distribution of the randomness and the coupling constant $\bf K$. }
\end{corollary}
\end{tcolorbox}
\noindent Here, the error chain $E\in {\cal V}\oplus {\cal V}$ represents a Pauli string in a similar fashion as Eq. \eqref{eq:z2pauli}. The probability $P(E)$ of the error chain is given by
\begin{align}
    P(E) = 
    p_0^{N-N_x(E)-N_y(E)-N_z(E)} p_x^{N_x(E)} p_y^{N_y(E)} p_z^{N_z(E)},
\end{align}
where $N_{x/y/z}(E)$ are the number of Pauli $X/Y/Z$-operators contained in the Pauli string $E$, $p_0\equiv(1-p_x-p_y-p_z)$, and $N$ is the total number of qubits. Similar to the cases of CSS code, the $R\rightarrow 1$ limit recovers the statistical models introduced in the previous literature to specifically study the decodability and the error thresholds in stabilizer codes (see Ref. \onlinecite{PhysRevX.2.021004} for example).

At this point, we have finished the general construction of the random statistical models for $S_R(\rho_{\rm m})$. These statistical models offer tools to study emergent quantum matters in decohered stabilizer codes and the DIPTs between them. The physical implication of the DIPTs for a general $R$ will be left for future studies.

For $R=2$, based on the same reasoning as in Sec. \ref{sec:differentR}, we can use the Choi-Jamio{\l}kowski isomorphism to map the error-corrupted matrix state $\rho_{\rm m}$ to its Choi representation $\kket{\rho_{\rm m}}$ in the doubled Hilbert space. We can show that
\begin{align}
    \kket{\rho_{\rm m}} \propto \left(\prod_\mu e^{\tK_x X_\mu\otimes X_\mu - \tK_y Y_\mu\otimes Y_\mu + \tK_z Z_\mu\otimes Z_\mu} \right)\kket{\rho_0},
\end{align}
where $\tK_{x,y,z}$ are given by $p_{x,y,z}$ through Eq. \eqref{eq:ncssrelation}. The negative sign in front of the $Y_\mu\otimes Y_\mu $ terms is the consequence of $Y_\mu^\mathsf{T}=-Y_\mu$ under the  Choi-Jamio{\l}kowski isomorphism. The spatial correlation function on $\kket{\rho_{\rm m}}$ is identical to the correlation functions in the $2$-replica random SM. One can construct a frustration-free parent Hamiltonian similar to Eq. \eqref{eq:Ham_doubled} for which $\kket{\rho_{\rm m}}$ is an exact ground state (see App. \ref{app:choi} for details). Therefore, the $R=2$ DIPT for a decohered stabilizer code is related to a quantum phase transition in the doubled Hilbert space.

As a simple illustration of the general construction of the statistical models, we briefly discuss its application to the Chamon model \cite{chamon2005quantum} decohered by the Pauli noise channel. The Chamon model is defined on a cubic lattice with a single qubit per site. There is one stabilizer for each cube on the lattice.  Each stabilizer is a product of 6 Pauli operators (see Fig. \ref{fig:chamon} (a)). The SM associated with the decohered Chamon model is defined on the dual cubic lattice with one classical $\mathbb Z_2$ spin per site. The interactions ${\cal O}^{x,y,z}$ in the classical Hamiltonian of SM are four-spin interactions illustrated in Fig. \ref{fig:chamon} (b). To the best of our knowledge, the phase diagram of SM (with or without randomness) has not been studied before. It is interesting to study the possible phases of the decohered Chamon model using SM in the future. 

\begin{figure}[tb]
    \centering
    \includegraphics[width = 0.95\linewidth]{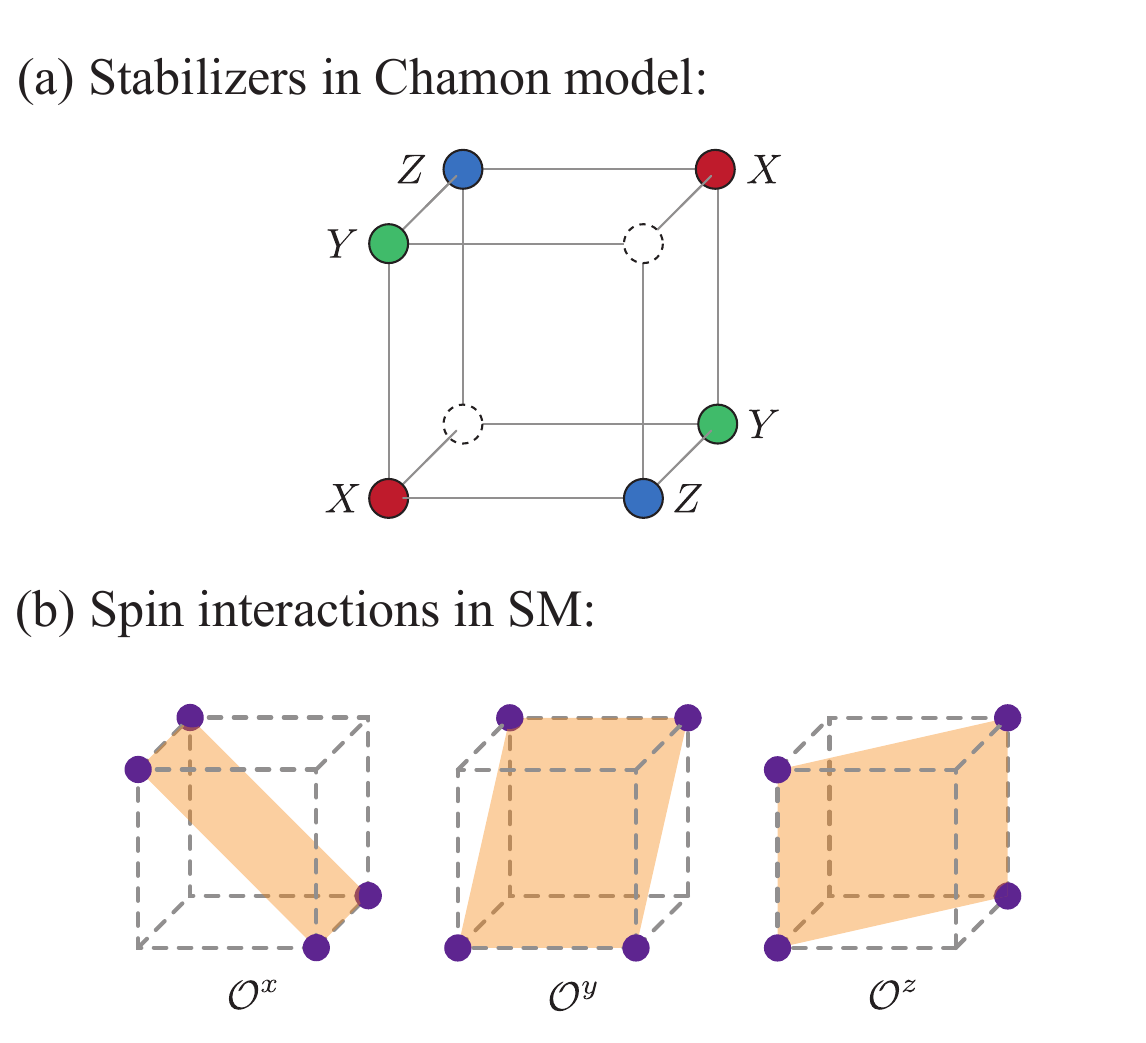}
    \caption{(a) The stabilizer of Chamon model (b) Three types of interaction ${\cal O}^{x,y,z}$ in SM: Each term (orange) is a product of four classical spins (purple dots). The dashed lines represents the edges of the dual cubic lattice.}
    \label{fig:chamon}
\end{figure}

\subsection{Dualities for general decohered stabilizer codes}

Now, we discuss the GPN (general Pauli noise) duality, a generalization of the BPD dualities (associated with the decohered CSS codes) to more general stabilizer codes. The statement of the GPN duality is the following.

\begin{tcolorbox}[width = 0.50\textwidth]
\begin{theorem}\label{thm:BPDncss}   
{\rm For $R = 2$, and $\infty$, the $R$-replica random statistical models that map to the R\'enyi entropy $S_R(\rho_{\rm m})$ at the error rates $(p_x,p_y,p_z)$ are dual to those statistical models at the error rates $(\tilde p_x, \tilde p_y,\tilde p_z)$. The two sets of dual error rates obey}
\begin{align}
\begin{dcases}
    \left(1-2p_y-2p_z\right)^2 = \frac{2\tilde{p}_0\tilde{p}_x+2\tilde{p}_y\tilde{p}_z}{\tilde{p}_0^2+\tilde{p}_x^2+\tilde{p}_y^2+\tilde{p}_z^2},\\
    \left(1-2p_z-2p_x\right)^2 =  \frac{2\tilde{p}_0\tilde{p}_y+2\tilde{p}_z\tilde{p}_x}{\tilde{p}_0^2+\tilde{p}_x^2+\tilde{p}_y^2+\tilde{p}_z^2} , &{\rm for~} R=2, \\
    \left(1-2p_x-2p_y\right)^2 =  \frac{2\tilde{p}_0\tilde{p}_z+2\tilde{p}_x\tilde{p}_y}{\tilde{p}_0^2+\tilde{p}_x^2+\tilde{p}_y^2+\tilde{p}_z^2}.
\end{dcases}
\label{eq:GPN_duality_condition_1}
\end{align}
{\rm and} 
\begin{align}
\begin{dcases}
   1-2p_y-2p_z =  \tilde p_x/ \tilde p_0, \\
   1-2p_z-2p_x = \tilde p_y/ \tilde p_0, &{\rm for~} R=\infty,\\
   1-2p_x-2p_y =  \tilde p_z/ \tilde p_0. \\
\end{dcases}
\label{eq:GPN_duality_condition_2}
\end{align}
{\rm Here, $p_0 \equiv 1-p_x-p_y-p_z$ and $\tilde p_0 \equiv 1-\tilde p_x-\tilde p_y- \tilde p_z$. }
\end{theorem}
\end{tcolorbox}

\noindent A proof of this theorem is given in App. \ref{app:decoheredual}. Conceptually, the GPN dualities with $R=2$ and $\infty$ can be effectively viewed as the descendants of the self-duality of the non-random SM. For $R=2$, after we integrate out the random couplings $E$, the two-replica random SM reduces to a single copy of non-random SM.  For the limit $R=\infty$, the partition function of the $R$-replica random SM is dominated by the trivial random coupling configuration, i.e. $E=0$. Therefore, the self-duality of the non-random SM implies the GPN dualities for $R =2, \infty$.  We caution that Theorem \ref{thm:noncss} is applicable only when the error rates $p_{x,y,z}$ and $\tilde p_{x,y,z}$ are both physical, namely they are between 0 and 1. In App. \ref{app:decoheredual}, we also discuss how the GPN duality recovers the BPD duality for $R=2,\infty$ for CSS codes.

\begin{figure}[t]
    \centering
    \includegraphics[width = 1\linewidth]{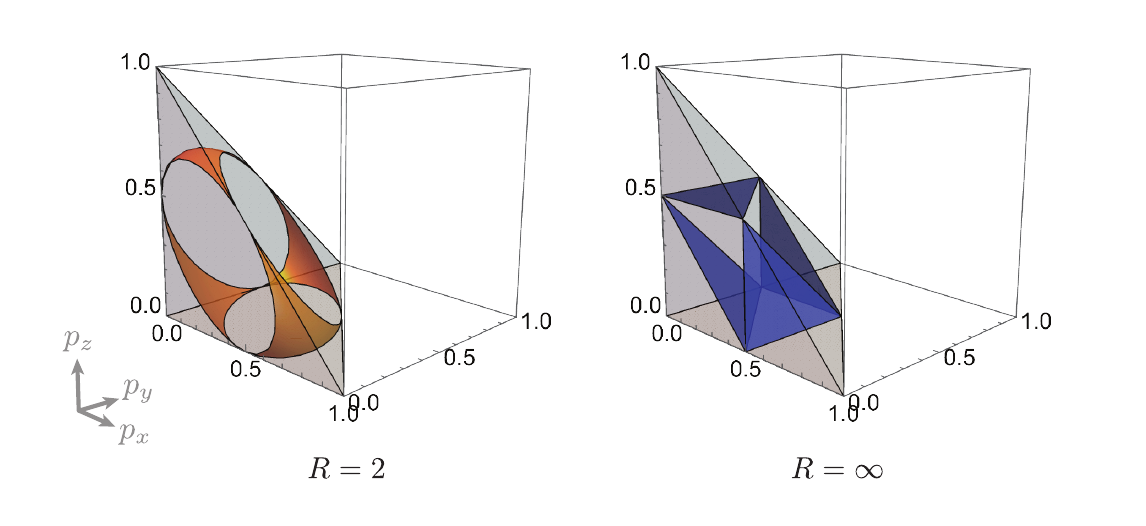}
    \caption{The surfaces of self-dual error rates for $R=2$ and $R=\infty$. The shaded area is the physical parameter regime with $p_{0,x,y,z}>0$.}
    \label{fig:self-dual_surface}
\end{figure}

For both $R=2$ and $R=\infty$, there is a surface of self-dual error rates. The error rates $(p_x, p_y, p_z)$ that satisfy the following equations map back to themselves under the GPN duality 
\begin{align}
\begin{cases}
      (1-p_x-p_y-p_z)^2 + p_x^2 +  p_y^2 + p_z^2 = \frac{1}{2},  ~~&{\rm for}~~ R = 2,  \\
     p_x + p_y +p_z = \frac{1}{2}, ~~&{\rm for}~~ R = \infty.
\end{cases}
\end{align}
Note that these self-dual conditions are derived within the regime with $p_{x,y,z}< p_0=1-p_x-p_y-p_z$. As commented earlier, the roles of $p_0$, $p_x$, $p_y$, and $p_z$ can be permuted once the Pauli noise channel is followed by a global unitary action $\prod_\mu X_\mu$, $\prod_\mu Y_\mu$, or $\prod_\mu Z_\mu$. With these permutations taken into account, the surfaces of self-dual error rates are depicted in Fig. \ref{fig:self-dual_surface}.

The surfaces of self-dual error rates are expected to constrain the phase diagram of the decohered stabilizer codes. One can argue that the renormalization group flow cannot cross this self-dual surface. A possible scenario is that this self-dual surface matches the critical surface of the DIPTs. The specific role of this self-dual surface in the phase diagram of the decohered stabilizer code should depend on the details of the model.

\section{Conclusions and Outlook}
\label{sec:conclusions}
In this paper, we systematically develop the formalism that studies the stabilizer codes decohered by Pauli noise/errors using statistical models. We focus on the R\'enyi entropies $S_R$ of the decohered codes as a probe of the systems entanglement structure. We find a general mapping between the $S_R$'s and classical statistical models that can be systematically constructed from the code's defining data. The phase transitions in these statistical models indicate non-trivial DIPTs in the decohered code. We discover intricate tapestries of dualities emerging among these statistical models. These dualities cast strong constraints on the phase diagram of the decohered quantum matter hosted in the stabilizer codes with noise. More specifically, this paper focuses on three general cases:
(1) CSS codes decohered by bit-flip and phase-flip errors (2) $em$-symmetric CSS codes under bit-flip and phase-flip decoherences. (3) General stabilizer codes under generic Pauli-noise decoherence.

First, for generic CSS codes, we construct a pair of statistical models ${\rm SM_1}$ and ${\rm SM_2}$, related to each other by an HLT duality (Eq. \eqref{eq:kw}), through an ungauging procedure. We show that the $R$-th R\'enyi entropy $S_R(\rho_{\rm b})$ for the error-corrupted mixed state $\rho_{\rm b}$ under the bit-flip decoherence is described by (1) $R$-replica SM$_1$ with rRC and (2) $R$-replica SM$_2$ with iRC. The two $R$-replica statistical models with different types of randomness are dual to each other through an HLT duality. 
Similarly, $S_R(\rho_{\rm p})$ for the error-corrupted mixed state $\rho_{\rm p}$ under phase-flip decoherence is described by (1) $R$-replica SM$_2$ with rRC and (2) $R$-replica SM$_1$ with iRC. These two models are also the HLT dual of each other. Moreover, for $R=2,3,\infty$, we find BPD dualities that relate $S_R(\rho_{\rm b})$ and $S_R(\rho_{\rm p})$, connecting the decoherence effects from the two error types.  It is a ``strong-weak" duality in that it maps strong bit-flip decoherence to weak phase-flip decoherence and vice versa.  The tapestry of dualities associated with a general CSS code is illustrated in Fig. \ref{fig:duality1}.

The classical statistical models that describe $S_R$ provide us with powerful tools to investigate the phase diagrams of decohered CSS codes. In particular, the DIPTs in a decohered CSS code are identified as the phase transitions in the corresponding statistical models. The DIPTs with $R=2,\infty$ are particularly simple as we have shown their equivalence to the phase transitions of ${\rm SM}_{1,2}$ without randomness. In general, DIPTs for different $R$'s happen at different critical error rates $p_{x/z}^\star(R)$. The universality classes with different $R$'s are also different. Thus, there is a family of DIPTs indexed by $R$ for a single CSS code. $p_{x/z}^\star(R\rightarrow 1)$ corresponds to the decoding error thresholds of the decohered CSS code, which has been extensively studied in many earlier works (for example Ref. \onlinecite{dennis2002topological, Preskill_3dToricCodeThreshold,ohno2004phase}). The DIPTs for $R\geq 2$ reveal other interesting singularities in the entanglement structure of the error-corrupted mixed states. For $R=2$, the DIPT can also be interpreted as a quantum phase transition in a doubled Hilbert space. The nature and physical implication of the DIPTs at $p^\star(R)$ for $R\geq 3$ are interesting questions for the future. We also propose a conjecture on the monotonicity of the critical error rates $p_{x/z}^\star(R)$ as a function of $R$ and present several pieces of evidence for this conjecture. 

Second, for CSS codes with an $em$ symmetry between the $X$-type and $Z$-type stabilizers, the corresponding $\SM1$ and $\SM2$ become the same model SM. The tapestry of dualities of $em$-symmetric CSS code becomes Fig. \ref{fig:self-duality}, which is effectively the tapestry of a general CSS code (Fig. \ref{fig:duality1}) folded in half. The R\'enyi entropy $S_R$ under either bit-flip or phase-flip decoherence is described by (1) $R$-replica SM with rRC and (2) $R$-replica SM with iRC. The two $R$-replica statistical models with iRC and rRC are related to each other through an HLT duality. The BPD duality for $R=2,3,\infty$ becomes a self-duality that relates strong bit-flip (phase-flip) decoherence to weak bit-flip (phase-flip) decoherence. Remarkably, these dualities yield super-universal self-dual error rates for $R=2,3,\infty$ (Eq. \eqref{eq:self-dual error rate}). If there is a unique DIPT (at a given $R$), the self-dual error rate must coincide with the critical error rate of the transition. In the two examples we examined, 2D toric code and Haah's code in 3D, the self-dual error rates indeed match the critical error rates of the DIPTs.
The properties of the DIPTs in these two codes are very different, demonstrating that the self-dual error rates encompass different universality classes of phase transitions and are super-universal. For cases with multiple DIPTs for single $R$, the BPD self-duality implies that the DIPT away from the self-dual error rates must appear in pairs, and the self-dual error rates are upper bounds to the transition with the lowest error rate.

Finally, we extend our analysis to general stabilizer codes under decoherence of general Pauli noises, i.e., independent $X$-, $Y$-, $Z$-errors with error rates $p_x, p_y, p_z$. The description of the decohered code is centered around a single statistical model SM. For a CSS code, SM is the previously constructed ${\rm SM}_{1,2}$ coupled together. The entropy $S_R(\rho_{\rm m})$ caused by the general Pauli noise is mapped to (1) $R$-replica SM with rRC and (2) $R$-replica SM with iRC. The two $R$-replica statistical models with randomness are related to each other by an HLT duality.
In addition, we find a GPN self-duality that relates different error rates for $R=2,\infty$, analogous to the BPD duality. Super-universal self-dual surfaces of error rates are identified for $R=2,\infty$ (Fig. \ref{fig:self-dual_surface}). These self-dual surfaces exist for a general stabilizer code and constrain the possible universal properties of the decohered code. One potential scenario is that the self-dual surfaces of error rates are exactly the phase boundary between different phases of the decohered code. A possible next step of investigation is to understand the physical meaning of these self-dual surfaces in specific stabilizer codes.

Our systematic construction of the statistical models provides powerful tools to investigate the DIPTs of QEC codes. For future research, an important question to address concerns the quantum-information-theoretic interpretation of these DIPTs indexed by $R$. As explained, in the $R\rightarrow 1$ limit, the DIPT is associated with the decodability of the logical information from the decohered code. Does the DIPT with a given $R>1$ mark the limit of decoherence for some concrete $R$-dependent quantum information processing protocol? Finding such protocols will deepen our understanding of the physical meaning of this family of DIPTs and will help refine the notion of topological order in error-corrupted mixed states. Moreover, such protocols will provide guidance on how to experimentally observe the DIPTs in a noisy intermediate-scale quantum platform. 

Another topic for future research pertains to our conjectured monotonic dependence of the critical error rates for DIPTs on the R\'enyi index $R$. The current conjecture concerns the bit-flip and phase-flip errors. If the conjecture is true, it will offer a tool to bound the critical error rates for a specific $R$, particularly the error threshold (in the limit of $R\rightarrow 1$), using results for larger $R$ values. We have presented a few pieces of evidence of the conjecture. Developing a deeper understanding of the relationship between statistical models at different $R$ will be important for resolving this conjecture. Moreover, it is also interesting to investigate if there is a similar relation between the critical error rates at different $R$ for other types of errors, such as the general Pauli noise. 

In this work, we have focused on the decoherence effects caused by Pauli noise/errors on the physical qubits of the QEC code. Investigating the effect of coherent errors, such as random-angle $X$- or $Z$-rotations, is a natural direction to generalize our framework. It was found that 2D surface code with coherent errors can be mapped to a 2D Ising model with complex couplings \cite{venn2023coherent}. It would be interesting to generalize the systematic construction of the statistical models and the tapestry of dualities to general QEC codes with coherent errors. For a QEC code, another type of error that affects its performance is the measurement read-out error. For the decodability problem, taking the read-out error into account results in a statistical model living in one higher dimension than the spatial dimension of the QEC code \cite{dennis2002topological}. An interesting open question is whether the measurement read-out errors can be investigated from the perspective of decohered quantum matter and if there is also a family of phase transitions similar to the family of DIPTs indexed by $R$.

Considering the decoherence-induced phases and transitions beyond stabilizer codes of qubit systems is another natural direction for future expeditions. For instance, it should be feasible to generalize our study to decohered stabilizer codes with qudits (each having a local Hilbert space of dimension $d>2$). Exploring the decoherence effect in subsystems codes \cite{Kribs2005, Kribs2006} and the recently discovered Floquet codes \cite{Hastings2021dynamically} can potentially uncover new forms of decoherence-induced quantum matter.

\acknowledgements
We thank Cenke Xu, Yuri Lensky and Yimu Bao for helpful discussions. C.-M.J. is supported by a faculty startup grant at Cornell University. 

{\it Note added} While completing this manual script, we noticed the independent related works \cite{rakovszky2023physics, zhao2023extracting}.

\onecolumngrid

\appendix

\section{Proofs of Theorem \ref{thm:sm_b} and Corollary \ref{corollary:S1_b}}\label{app:thmsm}

Let $\rb$ be the error-corrupted mixed state after bit-flip decoherence with error rate $p$. Combine Eqs. \eqref{eq:decohprobx}, \eqref{eq:renyisymm_b},
\begin{align}
    \Tr(\rho_{\rm b}^R) &= \frac{(1-p)^{NR}}{2^{{\rm dim}\vx}}\sum_{E\in {\cal V}}\sum_{C_{1,2,..,R}\in \vx}\lambda^{|E + C_1|}\cdots \lambda^{|E+C_R|} ,\label{eq:renyilambda}
\end{align}
where $\lambda = \frac{p}{1-p}$.

We will use two key observations. (1) the low-temperature expansion (LTE) of ${\rm SM_1}$ generates $\vx$; (2) the high-temperature expansion (HTE) of ${\rm SM_2}$ generates $\vz^\perp$ (see Eqs. \eqref{eq:SM1_HTE}, \eqref{eq:SM2_LTE}).
Take 3D toric code as an example (see \ref{sec:3dtoric}). The corresponding  $\rm SM_1$ is the 3D Ising model, and $\rm SM_2$ is the 3D lattice $\mathbb{Z}_2$ gauge theory. $\vx$ consists of boundary (in the sense that it bounds a volume) surfaces on the dual lattice. LTE of Ising model gives domain walls which are precisely those surfaces, while HTE of lattice $\mathbb{Z}_2$ gauge theory gives closed surfaces (not necessarily a boundary) which corresponds to $\vz^\perp$.

We will also substitute $\sum_{\vz^\perp} \to \sum_{\vx}$. This substitution essentially ignores the vectors in $\vz^\perp/\vx$, which corresponds to the logical operators. Under our assumption, there are no local logical operators. Hence, we argue that the contributions to the partition function from the vectors $\vz^\perp/\vx$ are exponentially suppressed as the system size goes to infinity. Additionally, we know that $N_{\rm c} = \text{dim}\left(\vz^\perp/\vx\right)$ is the number of logical qubits. And we've assumed the code rate $N_{\rm c} /N \rightarrow 0$ in the large system limit. Hence, when we substitute $\sum_{\vz^\perp}$ with $\sum_{\vx}$, the change in the free energy density of the corresponding statistical model is negligible in the infinite-system limit. 

Since free energy density is the standard diagnosis for phase transitions, we can do this substitution safely for the construction of the statistical models that target the DIPTs. In the following, we will use $\ninfeq$ whenever this substitution happens.

We now have all the ingredients for the proof. 
\begin{proof}[Proof of Theorem \ref{thm:sm_b}]
  \mbox{}\\*
    First consider LTE of $\rm SM_1$,
    \begin{align}
         Z_{\rm SM_1}(K, E) &= \sum_{\{\tau_i=\pm 1\}}\exp \left(K \sum_\mu  (-1)^{E_\mu}\mathcal{O}^z_\mu[\tau]\right)\nonumber\\
         &= N_s \sum_{C\in\vx}\exp \left(K \sum_\mu  (-1)^{E_\mu+C_\mu}\right)\nonumber\\
         &= N_s \sum_{C\in\vx}\exp \Big(K (N-2|E+C|)\Big)\nonumber\\
         &=\lambda^{-N/2}N_s\sum_{C\in \vx}\lambda^{|E+C|}. \label{eq:mapsm1}
    \end{align}
    In the second line, the summation is over all possible ``domain walls" $C$ generated by spin flips. The interaction term flips the sign on the domain walls and takes the value $\cO^z_\mu[\tau] = (-1)^{C_\mu}$. The symmetry factor $N_s$ results from different spin flips giving the same domain walls when symmetry is present in the system.  For instance, $N_s=2$ for the Ising model because of the global $\mathbb{Z}_2$ symmetry. And $N_s=3L$ for the plaquette Ising model, where $L$ is the length of the lattice, due to a planar $\mathbb{Z}_2$ symmetry (flipping all the spins on any planes) \cite{Vijay_2016}. In the third equality, we used the identity, $\sum_\mu (-1)^{E_\mu} = N-2|E|$. Compare with Eq. \eqref{eq:renyilambda}, we obtain,
\begin{equation}
    \Tr(\rho_{\rm b}^R) = \frac{[p(1-p)]^{RN/2}}{N_s^R 2^{{\rm dim}\vx} }\sum_{E\in {\cal V}} \left(Z_{\rm SM_1}(K, E)\right)^R.
\end{equation}
Second, consider HTE of $\rm SM_2$, 
\begin{align}
W_{\rm SM_2}(\tilde{K}, E) 
    &= \sum_{\{\tilde{\tau}_j=\pm 1\}}e^{\tilde{K} \sum_\nu \cO^x_\nu[\tilde{\tau}]} \prod_{\mu}\left(\cO^x_{\mu}[\tilde{\tau}]\right)^{E_\mu}\nonumber\\
    &= (1-\lambda^2)^{-N/2}\sum_{\{\tilde{\tau}_j= \pm1\}}\prod_\nu(1+\lambda \cO^x_\nu[\tilde{\tau}]) \prod_{\mu}\left(\cO^x_\mu[\tilde{\tau}]\right)^{E_\mu}\nonumber\\
     &=(1-\lambda^2)^{-N/2} 2^{N_B}\sum_{C\in \vz^\perp} \lambda^{|E+C|}\nonumber\\
    &\ninfeq (1-\lambda^2)^{-N/2}2^{N_B}\sum_{C\in \vx} \lambda^{|E+C|}. \label{eq:wsm2}
\end{align}
In the second equality, we used $\lambda = \tanh(\tilde{K})$. In the third equality, we kept the products of $\cO^x_\nu$ which cancel all the spins in the random insertion $\prod_{\mu}\left(\cO^x_{\mu}[\tilde{\tau}]\right)^{E_\mu}$. Such products are necessarily of the form $\prod_\nu \left(\cO^x_\nu[\tilde{\tau}]\right)^{(E+C)_\nu}, C\in \vz^\perp$. $N_B$ is the number of $B_j[Z]$ stabilizers or number of spins $\tilde{\tau}_j$.  Compare with Eq. \eqref{eq:renyilambda}, we obtain,
\begin{equation}
    \Tr(\rho_{\rm b}^R) \ninfeq \frac{(1-2p)^{RN/2}}{2^{\left({\rm dim}\vx + N_BR\right)}}\sum_{E\in {\cal V}}  \left(W_{\rm SM_2}(\tilde{K}, E)\right)^R.
\end{equation}
As a side note, Eqs. \eqref{eq:mapsm1} and \eqref{eq:wsm2} demonstrate the HLT duality between the random statistical models $Z_{\rm SM_1}(K,E)\propto W_{\rm SM_2}(\tK, E)$.

Finally, to obtain von Neumann entropy, combine Eq \eqref{eq:vn} and Eq. \eqref{eq:mapsm1},
\begin{equation}
    S_1(\rho_{\rm b}) = -\sum_{E\in {\cal V}} P_x(E)\log(Z_{\rm SM_1}(K,E)) - \frac{N}{2}\log(p(1-p)) + \log(N_s).
\end{equation}
\end{proof}

\section{Proof of Theorem \ref{thm:BPD}} \label{app:bpd}
Let us first consider bit-flip decoherence with error rate $p$ and $\lambda = \frac{p}{1-p}$ (ignoring subscripts $x$ for clarity). Start with Eq. \eqref{eq:renyisymm_b}, 
\begin{align}
    \text{Tr}(\rho_{\rm b}^R)  &=(1-p)^{RN} \sum_{E\in{\cal V}}\lambda^{|E|} \sum_{C_{2,..,R}\in \vx} \lambda^{|E+C_2|}\cdots\lambda^{|E+C_R|}. \label{eq:form1R1}
\end{align}
Now let's sum over $E$. Observe,
\begin{align}
    |E + C_\alpha| &= \sum_{\mu=1}^N (E_\mu + C_{\alpha,\mu}-2E_\mu C_{\alpha,\mu}),
\end{align}
where all arithmetic is done in $\mathbb{Z}$.
    \begin{align}
    \text{Tr}(\rho_{\rm b}^R) &= (1-p)^{RN}\sum_{C_{2,..,R}\in \vx} \sum_{\{E_\mu = 0,1\}}\lambda^{\sum_\mu\left[(R-2\sum_{\alpha=2}^R C_{\alpha,\mu})E_\mu +\sum_{\alpha=2}^R C_{\alpha,\mu} \right]}\nonumber\\
    &=(1-p)^{RN}\sum_{C_{2,..,R}\in \vx} \prod_{\mu=1}^N \sum_{E_\mu =0,1}\lambda^{\left[(R-2n_\mu)E_\mu +n_\mu \right]}\nonumber\\
    &=(1-p)^{RN}\sum_{C_{2,..,R}\in \vx} \prod_{\mu=1}^N \left(\lambda^{n_\mu} + \lambda^{R-n_\mu}\right) \label{eq:renyi1}\\
    &=(1-p)^{RN}(1+\lambda^R)^N\sum_{C_{2,..,R}\in \vx}\prod_{\mu=1}^N \left(\frac{\lambda^{n_\mu} + \lambda^{R-n_\mu}}{1+\lambda^R}\right), \label{eq:form1}
\end{align}
where we have defined the ``occupation number" $n_\mu \equiv \sum_{\alpha=2}^R C_{\alpha,\mu}$. In the last step, we normalized the product so that unoccupied links contribute a weight of $1$.

To prove the duality, we rearrange the sum in $\vx$ into a sum in $\vz$. Introduce an indicator function, 
\begin{lemma}
$$\delta(C\in \vx)  = \frac{1}{2^{{\rm dim}\vx^\perp}} \sum_{\tilde{C}\in \vx^\perp} (-1)^{C\cdot \tilde{C}}. $$
\end{lemma}
\begin{proof}\mbox{}\\*
If $C\in \vx$ , then $C\cdot \tilde{C}=0$ $\forall \tilde{C}\in \vx^\perp$. Thus $$\text{RHS}=\frac{1}{2^{{\rm dim}\vx^\perp}} \sum_{\tilde{C}\in \vx^\perp} 1 =1.$$
If $C\notin \vx = (\vx^\perp)^\perp$, then $\exists \tilde{C}_0\in \vx^\perp$, s.t. $C\cdot \tilde{C}_0 =1$. We can rewrite, $$\text{RHS} = \frac{1}{2\cdot 2^{{\rm dim}\vx^\perp}} \sum_{\tilde{C}\in \vx^\perp} \left((-1)^{C\cdot \tilde{C}} + (-1)^{C\cdot (\tilde{C}+\tilde{C}_0)}\right)=0.$$  
\end{proof}
Insert this indicator function into Eq. \eqref{eq:renyi1},

\begin{align}
    \Tr(\rho_{\rm b}^R)
    =&\frac{(1-p)^{RN}}{2^{(R-1){\rm dim}\vx^\perp}} \sum_{\tilde{C}_{2,..,R} \in \vx^\perp}\sum_{\{C_{\alpha,\mu} = 0,1\}} \prod_{\mu=1}^N\left[\left(\lambda^{\sum_\alpha C_{\alpha,\mu}} + \lambda^{R-\sum_{\alpha} C_{\alpha,\mu}}\right)(-1)^{\sum_{\alpha} C_{\alpha,\mu} \tilde{C}_{\alpha,\mu}}\right]\nonumber\\
    =&\frac{(1-p)^{RN}}{2^{(R-1){\rm dim}\vx^\perp}} \sum_{\tilde{C}_{2,..,R} \in \vx^\perp}\prod_{\mu=1}^N \left[ \prod_{\alpha=2}^R\left(\sum_{C_{\alpha,\mu}=0,1} \lambda^{C_{\alpha,\mu}}(-1)^{C_{\alpha,\mu} \tilde{C}_{\alpha,\mu}}\right) + \lambda^R \prod_{\alpha=2}^R\left(\sum_{C_{\alpha,\mu}=0,1}\lambda^{-C_{\alpha,\mu}}(-1)^{C_{\alpha,\mu} \tilde{C}_{\alpha,\mu}}\right)\right]\nonumber\\
     =&\frac{(1-p)^{RN}}{2^{(R-1){\rm dim}\vx^\perp}} \sum_{\tilde{C}_{2,..,R} \in \vx^\perp}\prod_{\mu=1}^N \left\{ (1+\lambda)^{R-1}[1+(-1)^{\tilde{n}_\mu}\lambda]\left(\frac{1-\lambda}{1+\lambda}\right)^{\tilde{n}_\mu}\right\}\nonumber\\
     =&\frac{1}{2^{(R-1){\rm dim}\vx^\perp}} \sum_{\tilde{C}_{2,..,R} \in \vx^\perp}\prod_{\mu=1}^N \left[\frac{1+(-1)^{\tilde{n}_\mu}\lambda}{1+\lambda} \left(\frac{1-\lambda}{1+\lambda}\right)^{\tilde{n}_\mu}\right]\nonumber\\
     \ninfeq&\frac{1}{2^{(R-1){\rm dim}\vx^\perp}} \sum_{\tilde{C}_{2,..,R} \in \vz}\prod_{\mu=1}^N \left[\frac{1+(-1)^{\tilde{n}_\mu}\lambda}{1+\lambda} \left(\frac{1-\lambda}{1+\lambda}\right)^{\tilde{n}_\mu}\right], \label{eq:form2}
\end{align}
with $\tilde{n}_\mu \equiv \sum_{\alpha=2}^R \tilde{C}_{\alpha,\mu}$. In the last equality, we replaced $\vx^\perp$ by $\vz$. Now let's consider phase-flip error with error rate $p_z$, $\lambda_z \equiv \frac{p_z}{1-p_z}$. Applying Eq. \eqref{eq:form1} but swapping $X$ and $Z$, 
\begin{align}
\text{Tr}(\rho_{\rm p}^R)
=(1-p_z)^{RN}(1+\lambda_z^R)^N\sum_{\tilde{C}_{2,..R}\in \vz}\prod_{\mu=1}^N \left(\frac{\lambda_z^{\tilde{n}_\mu} + \lambda_z^{R-\tilde{n}_\mu}}{1+\lambda_z^R}\right) ,\label{eq:form1zdecoh}
\end{align}
Compare Eqs. \eqref{eq:form2} and \eqref{eq:form1zdecoh}, the sum over $\vz$ will have the same weight if we can match (restoring subscripts $x$),
\begin{equation}
    \frac{1+(-1)^{\tilde{n}_\mu}\lambda_x}{1+\lambda_x} \left(\frac{1-\lambda_x}{1+\lambda_x}\right)^{\tilde{n}_\mu} = \frac{\lambda_z^{\tilde{n}_\mu} + \lambda_z^{R-\tilde{n}_\mu}}{1+\lambda_z^R}, \label{eq:matching}
\end{equation}
for every $\tilde{n}_\mu = 1,\dots,R-1$, with some function $\lambda_z(\lambda_x)$. Note $\lambda_z$ represents the phase-flip error rate and shall not depend on $\tilde{n}$.  We now show that this matching is possible for $R=2,3$.

For $R=2$, Eq. \eqref{eq:matching} gives only one equation for $\tilde{n}=1$,
\begin{align}
    &\left(\frac{1-\lambda_x}{1+\lambda_x}\right)^2 = \frac{2\lambda_z}{1+\lambda_z^2} \nonumber\\
    \Leftrightarrow & \left[(1-p_x)^2 + p_x^2\right]\left[(1-p_z)^2+p_z^2\right]=\frac{1}{2}.\label{eq:dualr2}
\end{align}
For $R=3$, Eq. \eqref{eq:matching} looks the same for $\tilde{n}=1,2$. There is still only one equation,
\begin{align}
    &\left(\frac{1-\lambda_x}{1+\lambda_x}\right)^2 =  \frac{\lambda_z+ \lambda_z^2}{1+\lambda_z^3}\nonumber\\
    \Leftrightarrow & \left[(1-p_x)^3 + p_x^3\right]\left[(1-p_z)^3+p_z^3\right]=\frac{1}{4}.
\end{align}
For $R\geq4$, there are, in general, no solutions because the system of equations is overdetermined.

The $R\to \infty$ duality can be readily understood with the statistical models. We observe that the partition functions with randomness are always smaller than the non-random partition function,
\begin{lemma}
\begin{equation}
    Z_{\SM{1}}(K,E)\leq Z_{\SM{1}}(K),~~~~W_{\SM{1}}(K, E)\leq Z_{\SM{1}}(K).
\end{equation}
\end{lemma}
\begin{proof}\mbox{}\\*
For the first part of the lemma, write $Z_{\SM{1}}(K,E)$ with HTE,
\begin{align}
    Z_{\rm SM_1}(K, E) &= \sum_{\{\tau_i=\pm 1\}}\exp \left(K \sum_\mu  (-1)^{E_\mu}\mathcal{O}^z_\mu[\tau]\right) \nonumber\\
    &=\cosh(K)^N 2^{N_A}\sum_{\tilde{C}\in \vx^{\perp}}\tanh(K)^{|\tilde{C}|}(-1)^{E\cdot \tilde{C}} \nonumber\\
    &\leq \cosh(K)^N 2^{N_A}\sum_{\tilde{C}\in \vx^{\perp}}\tanh(K)^{|\tilde{C}|} = Z_{\rm SM_1}(K),
\end{align}
where $N_A$ is the number of $A_i[X]$ stabilizers or number of $\tau_i$ spins.

For the second part of the lemma, write $W_{\SM{1}}(K, E)$ with LTE,
\begin{align}
    W_{\SM{1}}(K, E) &=\sum_{\{\tau_i=\pm 1\}} e^{K\sum_\nu\cO^z_\nu[\tau]} \prod_\mu \Big(\cO^z_\mu[\tau]\Big)^{E_\mu} \nonumber\\
    &=N_s\sum_{C\in \vx} \exp\left(K(-1)^{|C|}\right) (-1)^{C\cdot E} \nonumber\\
    &\leq N_s\sum_{C\in \vx} \exp\left(K(-1)^{|C|}\right) =Z_{\rm SM_1}(K).
\end{align}
\end{proof}
Now consider a bit-flip error with $\frac{p_x}{1-p_x} = e^{-2K}$ and a phase-flip error with $\frac{p_z}{1-p_z} = \tanh(K)$. In both cases, the quantity $\Tr(\rho^\infty) \approx (Z_{\rm SM_1}(K))^R$ (upto proportionality constant), because the partition functions with $E\neq 0$ are exponentially suppressed by large $R$. Thus the two decoherences are dual at $R\to \infty$ with the relation,
\begin{align}
    (1-p_x)(1-p_z) = \frac{1}{2}.
\end{align}
Note we assumed $p_x,p_z \in (0,\frac{1}{2})$. We observe in all three cases $R=2,3,\infty$, the BPD duality relates strong bit-flip decoherence ($p_x\to 1/2$) to weak phase-flip decoherence ($p_z\to0$) and vice versa.

\section{Integrate out the randomness of replica statistical models} \label{app:intdisorder}
In this section, we integrate out randomness from the replica statistical models. We shall see that effective inter-replica interactions emerge which are ``mediated" by randomness. We also provide an understanding of the $R=2,3$ BPD dualities from these statistical models. Consider bit-flip error-corrupted mixed state $\rho_{\rm b}$ and the quantity related to R\'enyi entropy $\text{Tr}(\rho_{\rm b}^R)$. Let us first look at the corresponding statistical model with rRC. Combine Eqs. \eqref{eq:renyisymm_b}, \eqref{eq:mapsm1} and the identity $\sum_\mu(-1)^{E_\mu} = N-2|E|$,
\begin{align}
    {\rm Tr}(\rho_{\rm b}^R)&\propto \sum_{E\in {\cal V}}\exp\left(K\sum_\mu (-1)^{E_\mu}\right) \left(Z_{{\rm SM}_1}(K,E)\right)^{R-1}\nonumber\\
    &=\sum_{\{E_\mu =0,1\}} \sum_{\{\tau_j^\alpha=\pm 1\}} \exp\left(K\sum_\mu (-1)^{E_\mu} \left(1+\sum_{\alpha=1}^{R-1} {\cal O}^z_\mu[\tau^\alpha] \right)\right)\nonumber\\
    &=\sum_{\{\tau_j^\alpha=\pm 1\}} \prod_{\mu}\sum_{E_\mu=0,1}\exp\left(K (-1)^{E_\mu} \left(1+\sum_{\alpha=1}^{R-1} {\cal O}^z_\mu[\tau^\alpha] \right)\right)\nonumber\\
    &=\sum_{\{\tau_j^\alpha=\pm 1\}} \prod_{\mu}\left(e^{K  \left(1+\sum_{\alpha=1}^{R-1} {\cal O}^z_\mu[\tau^\alpha] \right)} + e^{-K  \left(1+\sum_{\alpha=1}^{R-1} {\cal O}^z_\mu[\tau^\alpha] \right)}\right)\nonumber\\
     &= \sum_{\{\tau_j^\alpha=\pm 1\}} \prod_{\mu} 2(\cosh(K))^{R} \sum_{r=0}^{R-1}\sum_{\alpha_1<\alpha_2<...<\alpha_r} (\tanh(K))^r \left( \frac{1+ (-1)^r e^{-2K}}{1+ e^{-2K}}\right){\cal O}^{\alpha_1}\dots {\cal O}^{\alpha_r},  \label{eq:intdisZ}
\end{align}
where in the last line we abbreviated ${\cal O}^\alpha \equiv {\cal O}^z_\mu[\tau^\alpha]$, $\alpha = 1,..,R-1$ is the replica index since we now have $R-1$ copies of spin models. The last line manifests arbitrary inter-replica coupling ${\cal O}^{\alpha_1}\dots {\cal O}^{\alpha_r}$ at the same site $\mu$. We also define ${\cal O}^{\alpha_1}\dots {\cal O}^{\alpha_r}= 1$ when $r=0$.

Let us now consider the special cases $R=2,3$. First, when $R=2$, the replica theory Eq. \eqref{eq:intdisZ} has only a single copy,
\begin{align}
    \text{Tr}(\rho_{\rm b}^2)    &\propto \sum_{\{\tau_j=\pm 1\}} \prod_\mu \Big(1+ (\tanh(K))^2 {\cal O}^z_\mu[\tau] \Big) \nonumber\\
    &\propto \sum_{\{\tau_j=\pm 1\}} \prod_\mu e^{K'{\cal O}^z_\mu[\tau]}  = Z_{{\rm SM}_1}(K'), \label{eq:R2ren}
\end{align}
where $\tanh(K') = (\tanh(K))^2$. So upon integrating out errors $E$, the statistical model with randomness reduces to a clean model with renormalized coupling $K'$.

Second, when $R=3$, there are now two replicas $\alpha =1,2$, 
\begin{align}
    \text{Tr}(\rho_{\rm b}^3) &\propto  \sum_{\{\tau_j^\alpha=\pm 1\}} \prod_\mu \Big(1+ (\tanh(K))^2(\cO^1+\cO^2 + \cO^1\cO^2)\Big) \nonumber\\
    &\propto \sum_{\{\tau_j^\alpha=\pm 1\}} \prod_\mu e^{K''(\cO^1+\cO^2+\cO^1\cO^2)}, \label{eq:genak1}
\end{align}
where $\frac{\tanh(K'')+(\tanh(K''))^2}{1+(\tanh(K''))^3} = (\tanh(K))^2$. For Ising model, $\cO_{\braket{ij}} = \tau_i\tau_j$, Eq. \eqref{eq:genak1} becomes the partition function of Ashkin-Teller model.

Let us now examine the case with iRC. Combine Eqs. \eqref{eq:renyisymm_b} and \eqref{eq:wsm2},
\begin{align}
    \text{Tr}(\rho_{\rm b}^R)&\propto  \sum_{E\in {\cal V}}\exp\left(K\sum_\mu (-1)^{E_\mu}\right) (W_{{\rm SM}_2}(\tK, E))^{R-1} \nonumber\\
    &=\sum_{\{E_\mu =0,1\}}\sum_{\{\ttau_j^\alpha = \pm 1\}}e^{K\sum_\mu(-1)^{E_\mu}} e^{\tK \sum_{\nu, \alpha} \cO^x_\nu[\ttau^\alpha]} \prod_{\sigma, \alpha}\left(\cO^x_{\sigma}[\ttau^\alpha]\right)^{E_\sigma}\nonumber\\
    &=\sum_{\{\ttau_j^\alpha = \pm 1\}} \prod_\mu e^{\tK \sum_{ \alpha} \cO^x_\mu[\ttau^\alpha]} \sum_{E_\mu = 0,1} e^{K(-1)^{E_\mu}} \left(\prod_\alpha \cO^x_\mu[\ttau^\alpha]\right)^{E_\mu}\nonumber\\
    &=\sum_{\{\ttau_j^\alpha = \pm 1\}} \prod_\mu \frac{e^K}{\cosh(\tK)} \exp\left(\tK \left(\sum_{ \alpha=1}^{R-1} \tcO^\alpha + \prod_{\beta=1}^{R-1}\tcO^\beta \right)\right), \label{eq:intdisW}
\end{align}
where $\alpha, \beta = 1,..,R-1$ are the replica indices. In the last line, we again adopts the shorthand $\tcO^{\alpha,\beta} \equiv \cO^x_\mu[\ttau^{\alpha,\beta}]$. We see that the effective interaction in Eq. \eqref{eq:intdisW} is more constrained than Eq. \eqref{eq:intdisZ}. In the former, different replica copies only couple via the product $\prod_\alpha \tcO^\alpha$ of all copies, while in the latter, arbitrary products are present. However, this difference disappears when $R=2,3$, as we can observe by comparing Eqs. \eqref{eq:R2ren}, \eqref{eq:genak1} and \eqref{eq:intdisW}. This coincidence facilitates the BPD duality. We note that Eq. \eqref{eq:intdisW} was derived for the special case of 2D toric code in Ref. \onlinecite{FanBaoTopoMemory} while our result applies for any CSS code.

\section{Proofs of Theorem \ref{thm:noncss} and Corollary \ref{corollary:S1_m}} \label{app:ncss}
Let us first introduce some notations. As mentioned in the main text, on an $N$-qubit system, a Pauli string $E$ representing errors can be specified by a $\mathbb{Z}_2$ vector of length $2N$, 
\begin{equation}
E = (\sfe,\sff) \in \cV\oplus\cV.
\end{equation}
The Pauli operator acting on the $\mu$-th qubit is $E_\mu=I, X,Y,Z$ if $(\sfe_\mu, \sff_\mu) = (0,0), (1,0), (1,1), (0,1)$ respectively. In each error chain $E$,  the number of $X, Y, Z$ operators, $N_x(E), N_y(E), N_z(E)$ are given by,
\begin{align}
    N_x(E) = \frac{|\sfe|-|\sff|+|\sfe+\sff|}{2}, \hspace{.5cm} N_y(E) = \frac{|\sfe|+|\sff|-|\sfe+\sff|}{2}, \hspace{.5cm} N_z(E) = \frac{-|\sfe|+|\sff|+|\sfe+\sff|}{2}.
\end{align}

Similarly, a stabilizer $A_J[X,Z]$ is represented by a vector 
\begin{equation} 
A_J = (\sfa_J, \sfb_J) \in \cV \oplus\cV.
\end{equation}
That all stabilizers commute amounts to $\sfa_J\cdot \sfb_{J'} + \sfa_{J'}\cdot \sfb_J=0, \forall J, J'$. The set of stabilizers $\{A_J\}$ span a subspace $\cV_s \subset \cVV$. We also define the dual subspace
\begin{align}
    &\cV^*_s \equiv \{\Omega W|W\in \cV_s\},
\end{align}
where $\Omega =
    \begin{pmatrix}
        0 & \mathds{1}_{N\times N} \\
        \mathds{1}_{N\times N} & 0
    \end{pmatrix}$. Note the matrix $\Omega$ effectively interchanges $X$ and $Z$ operators in a Pauli string. We can rewrite the commutation relation of stabilizers as,
\begin{equation}
    (A_J)\cdot (\Omega A_{J'}) =0, ~~~~\forall J, J'.
\end{equation}
This implies $\cV_s \subset \left(\cV_s^*\right)^\perp$. Similar to the case of CSS code, the quotient space $\left(\cV_s^*\right)^\perp/\cV_s$ contains logical operators for the QEC code. We assume that the logical operators are non-local and the code rate $N_c/N\to 0$. As in the CSS code case, the contribution of such non-local operators to $\Tr(\rho^R)$ is exponentially small, and the overall correction to free energy density is negligible.  We can, therefore, make the substitution $\left(\cV_s^*\right)^\perp\to\cV_s$ in what follows. As before, we will use $\ninfeq$ to signal this substitution.

Let us now see how the spaces $\cV_s, \cV_s^*$ emerge in the statistical model. Recall definitions of the interaction terms,
\begin{align}
   {\cal O}^z_\mu[\tau] = \prod_J (\tau_J)^{(\mathsf{a}_J)_{\mu}},~~~~
    {\cal O}^x_\mu[\tau] = \prod_J (\tau_J)^{(\mathsf{b}_J)_\mu}, ~~~~
    {\cal O}^y_\mu[\tau] = {\cal O}^x_\mu[\tau]{\cal O}^z_\mu[\tau]. \label{eq:ncssint}
\end{align}
LTE of this model gives the vectors in $\cV_s$. To see this, consider a single flipped spin, $\tau_J = -1$ and $\tau_{J'}=1, \forall J'\neq J$. The interaction terms become,
\begin{equation}
    \cO^z_\mu[\tau] = (-1)^{(\sfa_J)_\mu}, ~~~~\cO^x_\mu[\tau] = (-1)^{(\sfb_J)_\mu}.
\end{equation}
This domain wall pattern, therefore, encodes the vector $A_J =(\sfa_J,\sfb_J) \in \cV_s$.

On the other hand, HTE of the statistical model yields the vectors in $\left(\cV_s^*\right)^\perp$. The HTE sums over products of interaction terms that cancel all the spins. An arbitrary product of interaction terms has the form,
\begin{equation}
    \prod_\mu \Big( \cO^x_\mu[\tau]\Big)^{\sfa_\mu}\Big( \cO^z_\mu[\tau]\Big)^{\sfb_\mu}, \label{eq:oxoz}
\end{equation}
where $C\equiv(\sfa, \sfb) \in \cVV$. To cancel all the spins, this combination of interactions must satisfy,
\begin{align}
    1&=\prod_\mu \Big( \cO^x_\mu[\tau]\Big)^{\sfa_\mu}\Big( \cO^z_\mu[\tau]\Big)^{\sfb_\mu}  \nonumber\\
    &=\prod_\mu\prod_J (\tau_J)^{(\sfb_J)_\mu\sfa_\mu + (\sfa_J)_\mu\sfb_\mu} \nonumber\\
    &=\prod_J (\tau_J)^{C\cdot(\Omega A_J)}.
\end{align}
The exponent for each $\tau_J$ must be $0$. This implies 
\begin{equation}
C\cdot  (\Omega A_J) = 0, \forall J \Leftrightarrow C\in  \left(\cV_s^*\right)^\perp.
\end{equation}
We observe here that, for SM, the space of LTE configurations $\vs$ and the space of HTE configurations $\vsdp$ are identical up to non-local terms. This hints at the HLT self-duality of SM, which will be evident in subsequent discussions.

Now consider the general decoherence $\rho_0 = \ket{\Omega}\bra{\Omega}\to \rho_{\rm m}$,
\begin{align}
    \rho_{\rm m}&=\sum_{E\in\cVV} P(E) O^\dagger(E)\rho_0 O(E), \nonumber\\
    P(E) &= (1-p_x-p_y-p_z)^{N-N_x-N_y-N_z}p_x^{N_x}p_y^{N_y}p_z^{N_z}, \nonumber\\
    O(E) &=  (\i)^{\sfe \cdot \sff}\prod_\mu \left(X_\mu\right)^{\sfe_\mu}\left(Z_\mu\right)^{\sff_\mu},
\end{align} 
 where $O(E)$ is the Pauli string specified by the error chain $E = (\sfe, \sff)$. For later convenience we rewrite the error rates with $\lambda_i \equiv \frac{p_i}{1-p_x-p_y-p_z}$ for $i=x,y,z$,
\begin{align}
    P(E) &= (1-p_x-p_y-p_z)^N \lambda_x^{N_x} \lambda_y^{N_y} \lambda_z^{N_z} \label{eq:ncsspb1}\\
    &=(1-p_x-p_y-p_z)^N \left(\sqrt{\frac{\lambda_y \lambda_z}{\lambda_x}}\right)^{|\sff|} \left(\sqrt{\frac{\lambda_x \lambda_y}{\lambda_z}}\right)^{|\sfe|}  \left(\sqrt{\frac{\lambda_x \lambda_z}{\lambda_y}}\right)^{|\sfe+\sff|}. \label{eq:ncsspb2} 
\end{align}
The information theoretical quantity that detects DIPTs is,
\begin{align}
    &\hspace{.5cm}\Tr(\rho_{\rm m}^R) \nonumber\\
    &= \sum_{E_{1,..,R}\in \cVV} P(E_1)P(E_2)\dots P(E_R) \braket{O(E_1)O^\dagger(E_2)}_\Omega\dots\braket{ O(E_{R-1})O^\dagger(E_R)}_\Omega\braket{O(E_R)O^\dagger(E_1)}_\Omega\nonumber\\
    &=\sum_{E\in\cVV}\sum_{C_{2,..,R}\in \cV_s} P(E) P(E+C_2)\dots P(E+C_R) \times\nonumber\\ &\hspace{1cm}\braket{O(E)O^\dagger(E + C_2)}_\Omega\dots\braket{ O(E+C_{R-1})O^\dagger(E+C_R)}_\Omega\braket{O(E+ C_R)O^\dagger(E)}_\Omega\nonumber\\
    &=\sum_{E\in \cVV}\sum_{C_{2,..,R} \in \vs} P(E) P(E+C_2)\dots P(E+C_R) \times\nonumber\\ &\hspace{1cm} \braket{O(E)O^\dagger(E + C_2)O(E + C_2)O^\dagger(E+C_3)\dots O(E+ C_R)O^\dagger(E)}_\Omega\nonumber\\
    &=\sum_{E\in \cVV}\sum_{C_{2,..,R} \in \vs} P(E) P(E+C_2)\dots P(E+C_R) \label{eq:renyincssasym}\\
    &=\frac{1}{2^{{\rm dim}\vs}}\sum_{E\in\cVV} \left[ \sum_{C\in \vs} P(E+C)\right]^R. \label{eq:renyincsssym}
\end{align}
In the second equality, we used the condition that $\braket{OO^\dagger}_\Omega\neq 0$ iff $O O^\dagger$ is generated by stabilizers. However, there is an ambiguity of powers of $\i$ in $\braket{O O^\dagger}_\Omega$ because the $X,Y,Z$'s do not commute. To resolve this, in the third equality, we used the fact that $\ket{\Omega}$ is an eigenstate of all the operators $OO^\dagger$. In the fourth equality, we noted that the product of all $OO^\dagger$ is identity. In the last equality, we symmetrized the expression in a similar way as Eq. \eqref{eq:renyisymm_b}.  It remains to relate $\sum_{C\in \vs} P(E+C)$ to the partition functions of the statistical models with randomness.

First, consider LTE of SM with rRC $E=(\sfe,\sff) \in \cVV$, 
\begin{align}
    Z_{\rm SM}(\mathbf{K}, E) &= \sum_{\{\tau_J=\pm1\}}\exp \left(\sum_\mu \left(K_x (-1)^{\sff_\mu} \mathcal{O}^x_\mu[\tau]+K_y (-1)^{(\sfe+\sff)_\mu} \mathcal{O}^y_\mu[\tau]+K_z (-1)^{\sfe_\mu} \mathcal{O}^z_\mu[\tau]\right)\right) \nonumber\\
    &=N_s\sum_{C = (\sfa, \sfb) \in \vs} \exp\left(\sum_\mu \left(K_x(-1)^{(\sff+\sfb)_\mu} + K_y(-1)^{(\sfe+\sff+\sfa+\sfb)_\mu} + K_z(-1)^{(\sfe+\sfa)_\mu}\right)\right)\nonumber\\
    &=N_s e^{N(K_x+K_y+K_Z)}\sum_{C = (\sfa, \sfb) \in \vs}\left(e^{-2K_x}\right)^{|\sff+\sfb|}\left(e^{-2K_z}\right)^{|\sfe+\sfa|}\left(e^{-2K_y}\right)^{|(\sfe+\sfa)+(\sff+\sfb)|}\nonumber\\
    &=N_s e^{N(K_x+K_y+K_Z)}\sum_{C = (\sfa, \sfb) \in \vs} \left(\sqrt{\frac{\lambda_y \lambda_z}{\lambda_x}}\right)^{|\sff+\sfb|} \left(\sqrt{\frac{\lambda_x \lambda_y}{\lambda_z}}\right)^{|\sfe+\sfa|} \left(\sqrt{\frac{\lambda_x \lambda_z}{\lambda_y}}\right)^{|(\sfe+\sfa)+(\sff+\sfb)|} \nonumber\\
     &=N_s[p_xp_yp_z(1-p_x-p_y-p_z)]^{-N/4}\sum_{C\in \vs}P(E+C), \label{eq:ncssLTE}
\end{align}
where $N_s$ is a symmetry factor, in the third equality we used the identity $\sum_\mu (-1)^{\mathsf{w}_\mu} = N-2|\mathsf{w}|, \forall\mathsf{w}\in \cV$, in the fourth equality we applied generalized Nishimori condition Eq. \eqref{eq:ncssrelation}, in the last equality we compared to Eq. \eqref{eq:ncsspb2}. Combine Eq. \eqref{eq:renyincsssym} and \eqref{eq:ncssLTE},
\begin{equation}
    \Tr(\rho^R_{\rm m}) = \frac{[p_xp_yp_z(1-p_x-p_y-p_z)]^{NR/4}}{N_s^R 2^{{\rm dim}\vs}} \sum_{E\in{\cal V}\oplus {\cal V}} \Big(Z_{\rm SM}(\mathbf{K}, E)\Big)^R.
\end{equation}
This proves the first part of the theorem.

Now consider HTE of SM with iRC $E=(\sfe,\sff) \in \cVV$, 
\begin{align}
    W_{\rm SM}(\mathbf{\tK}, E) &= \sum_{\{\tau_J=\pm1\}}\exp\left(\sum_{\mu}\Big(\tilde{K}_x \cO^x_{\mu}[\tau]+\tilde{K}_y \cO^y_{\mu}[\tau] + \tilde{K}_z {\cal O}^z_{\mu}[\tau]\Big)\right) \prod_{\nu} \Big(\cO^x_\nu[\tau]\Big)^{\sfe_\nu} \Big(\cO^z_\nu[\tau]\Big)^{\sff_\nu}\nonumber\\
    &=f(\mathbf{\tK})^N\sum_{\{\tau_J\}}\prod_\mu(1+\lambda_x \cO_\mu^x[\tau] + \lambda_y \cO_\mu^y[\tau] + \lambda_z\cO_\mu^z [\tau]) \prod_{\nu} \Big(\cO^x_\nu[\tau]\Big)^{\sfe_\nu} \Big(\cO^z_\nu[\tau]\Big)^{\sff_\nu}\nonumber\\
    &=f(\mathbf{\tK})^N\sum_{\{\tau_J\}} \sum_{W=(\sfa, \sfb) \in \cVV} \lambda_x^{N_x(W)} \lambda_y^{N_y(W)} \lambda_z^{N_z(W)} \prod_{\mu} \Big(\cO^x_\mu[\tau]\Big)^{(\sfa+\sfe)_\mu} \Big(\cO^z_\mu[\tau]\Big)^{(\sfb+\sff)_\mu}\nonumber\\
    &=f(\mathbf{\tK})^N 2^{N_A}\sum_{C\in \vsdp} \lambda_x^{N_x(E+C)} \lambda_y^{N_y(E+C)} \lambda_z^{N_z(E+C)}\nonumber\\
    &\ninfeq f(\mathbf{\tK})^N 2^{N_A}\sum_{C\in \vs} \lambda_x^{N_x(E+C)} \lambda_y^{N_y(E+C)} \lambda_z^{N_z(E+C)}\nonumber\\
    &=[(1-2p_x-2p_y)(1-2p_x-2p_z)(1-2p_y-2p_z)]^{-N/4}2^{N_A}\sum_{C\in \vs} P(E+C), \label{eq:ncssHTE} 
\end{align}
where $f(\mathbf{\tK}) = \cosh(\tK_x)\cosh(\tK_y)\cosh(\tK_z)(1+\tanh(\tK_x)\tanh(\tK_y)\tanh(\tK_z))$ and $N_A$ is the number of stabilizers $A_J[X,Z]$. In the second equality we expanded the exponential for each $\mu$ into a polynomial and used Eq. \eqref{eq:ncssrelation}, in the third equality $W\in \cVV$ represents an arbitrary combination of interaction terms (cf. Eq. \eqref{eq:oxoz}), in the fourth equality only $W=E+C, C\in \vsdp$ survives the $\tau_J$ summation and the last equality applied Eq. \eqref{eq:ncsspb1}. Combine Eq. \eqref{eq:renyincsssym} and \eqref{eq:ncssHTE},
\begin{equation}
    \Tr(\rho^R_{\rm m}) \ninfeq \frac{[(1-2p_x-2p_y)(1-2p_x-2p_z)(1-2p_y-2p_z)]^{NR/4}}{2^{({\rm dim}\vs+N_AR)}} \sum_{E\in{\cal V}\oplus {\cal V}} \Big(W_{\rm SM}(\mathbf{\tK}, E)\Big)^R.
\end{equation}
This completes the second part of the proof. 

We now comment on the self-duality of the non-random SM (Eq. \eqref{eq:noncssKW}). Compare Eqs. \eqref{eq:ncssLTE} and \eqref{eq:ncssHTE}, we conclude that an HLT duality relates $Z_{\rm SM}(\mathbf{K}, E)$ and $W_{\rm SM}(\mathbf{\tK}, E)$ for arbitrary $E$. Setting $E=0$ gives the HLT self-duality of non-random SM, $Z_{\rm SM}(\mathbf{K})\propto Z_{\rm SM}(\mathbf{\tK}) $.

Finally, to obtain the von Neumann entropy $S_1(\rho_{\rm m})$, combine Eqs. \eqref{eq:renyincssasym} and \eqref{eq:ncssLTE} and take the $R\to1$ limit for $S_R$,
\begin{equation}
    S_1(\rho_{\rm m}) = -\sum_{E\in {\cal V}\oplus {\cal V}} P(E) \log(Z_{\rm SM}(\mathbf{K}, E)) - \frac{N}{4}\log\big(p_xp_yp_z(1-p_x-p_y-p_z)\big) + \log(N_s).
\end{equation}

\section{Proof of the GPN duality (Theorem \ref{thm:BPDncss})}\label{app:decoheredual}
Throughout this appendix, we use the shorthands,
\begin{align}
    p_0 \equiv 1-p_x-p_y-p_z, ~~~~\lambda_i \equiv \frac{p_i}{p_0} ~~\text{for~}i=x,y,z.
\end{align}
\subsection{Derivation for $R=2$}
  Purity $\Tr(\rho_{\rm m}^2)$ can be written as 
    \begin{equation}
    \sum_{E_{1,2}\in \cVV}P(E_1)P(E_2)\delta(E_1+E_2\in \cV_s).
    \end{equation}
    There are two ways to manipulate the $\delta(E_1+E_2 \in \cV_s)$. First one can write
    \begin{align}
        &\sum_{E_{1,2}\in \cVV}P(E_1)P(E_2)\delta(E_1+E_2\in \cV_s) \nonumber\\
        = &\sum_{E\in\cVV}\sum_{C\in \cV_s}P(E)P(E+C)\nonumber\\
        = &p_0^{2N} \sum_{C\in \cV_s} \left(\sum_{E\in\cVV}\lambda_x^{N_x(E)+N_x(E+C)}\lambda_y^{N_y(E)+N_y(E+C)}\lambda_z^{N_z(E)+N_z(E+C)}\right)\nonumber\\
        \equiv &p_0^{2N} \sum_{C\in \cV_s} \prod_\mu f(C_\mu, \bp).
    \end{align}
    The sum on $E$ is unconstrained and can be performed independently on each qubit with $E_\mu = I, X, Y, Z$, giving a Boltzmann weight,
    \begin{equation}
        f(C_\mu, \bp) = \begin{pmatrix}
            1+\lambda_x^2+ \lambda_y^2 + \lambda_z^2 \\
            2\lambda_x+2\lambda_y\lambda_z \\
            2\lambda_y+2\lambda_x\lambda_z\\
            2\lambda_z+2\lambda_x\lambda_y
        \end{pmatrix}
        = \frac{1}{p_0^2}\begin{pmatrix}
            p_0^2 + p_x^2+p_y^2+p_z^2\\
            2p_xp_0 + 2p_yp_z \\
            2p_yp_0 + 2p_xp_z\\
            2p_zp_0 + 2p_xp_y
        \end{pmatrix},
    \end{equation}
    where the four rows corresponds to $C_\mu = I, X, Y, Z$.

    The second way to write $\delta(E_1+E_2\in V)$ is through the following resolution
    \begin{equation}
        \delta(E_1+E_2\in \cV_s) = \frac{1}{2^{{\rm dim}\vsdp}}\sum_{C\in \vsdp} (-1)^{C\cdot\Omega (E_1+E_2)}.
    \end{equation}
    Purity now becomes,
    \begin{align}
    \Tr(\rho_{\rm m}^2) &= \frac{1}{2^{{\rm dim}\vsdp}}\sum_{C\in \vsdp}\sum_{E_{1,2}\in \cVV} P(E_1) P(E_2) (-1)^{C\cdot\Omega (E_1+E_2)} \nonumber\\
    &=\frac{1}{2^{{\rm dim}\vsdp}}\sum_{C\in \vsdp} \left(\sum_{E\in\cVV}P(E)(-1)^{C\cdot (\Omega E)}\right)^2\nonumber\\
    &\ninfeq \frac{(1-\sum_i p_i)^{2N}}{2^{{\rm dim}\vsdp}}\sum_{C\in \vs}\left(\sum_{E\in\cVV}\lambda_x^{N_x(E)}\lambda_y^{N_y(E)}\lambda_z^{N_z(E)}(-1)^{C\cdot (\Omega E)}\right)^2\nonumber\\
    &\equiv \frac{(1-\sum_i p_i)^{2N}}{2^{{\rm dim}\vsdp}}\sum_{C\in \vs} \prod_\mu \tilde{f}(C_\mu, \bp). \label{eq:ncssr2sm}
    \end{align}
    Again $E$ can be summed for each link individually $E_\mu = I, X, Y, Z$. Note $(-1)^{C\cdot (\Omega E)}$ tells us to insert a $(-1)$ whenever $[C_\mu, E_\mu]\neq0$. This gives Boltzmann weight,
      \begin{equation}
        \tilde{f}(C_\mu,\bp) = \begin{pmatrix}
            (1+\lambda_x+\lambda_y+\lambda_z)^2\\
            (1+\lambda_x-\lambda_y-\lambda_z)^2 \\
            (1+\lambda_y-\lambda_x-\lambda_z)^2\\
            (1+\lambda_z-\lambda_x-\lambda_y)^2
        \end{pmatrix}
        = \frac{1}{p_0^2} \begin{pmatrix}
            1\\
            (1-2p_y-2p_z)^2 \\
            (1-2p_x-2p_z)^2\\
            (1-2p_x-2p_y)^2
        \end{pmatrix}.
        \label{eq:ncssr2bw2}
    \end{equation}
    Equating the nomalized Boltzmann weights,
    \begin{equation}
        \frac{f(C_\mu, \bp)}{f(I, \bp)} = \frac{\tilde{f}(C_\mu, \btp)}{\tilde{f}(I, \btp)},     \label{eq:noncssbwmapping}
    \end{equation}
    for $C_\mu = X,Y,Z$ gives the GPN duality at $R=2$.

\subsection{Derivation for $R=\infty$}
At $R\to \infty$ the summation over randomness is dominated by the clean limit $E=0$,
\begin{equation}
0\leq Z(\mathbf{K}, E)\leq Z(\mathbf{K}), ~~~~ 0\leq W(\mathbf{K}, E)\leq Z(\mathbf{K}). \label{eq:cleandomnoncss}
\end{equation}
But $Z(\mathbf{K}, E)$ and $W(\mathbf{K}, E)$ describe decoherence channels with error rates $\bp$ and $\btp$ respectively,
\begin{align}
    \lx = e^{-2(K_y + K_z)},~~~~\ly = e^{-2(K_x + K_z)},~~~~\lz = e^{-2(K_x + K_y)}, \nonumber\\
    \tlx = e^{-2(\tK_y + \tK_z)},~~~~\tly = e^{-2(\tK_x + \tK_z)},~~~~\tlz = e^{-2(\tK_x + \tK_y)},
\end{align}
where $\mathbf{K},\mathbf{\tK}$ are related by Eq. \eqref{eq:noncssKW}. For both channels, $\Tr(\rho_{\rm m}^R)\approx (Z(\mathbf{K}))^R$ when $R\to \infty$. This gives the $R\to\infty$ GPN duality.

The proof of Eq.\eqref{eq:cleandomnoncss} parallels that of the CSS case. Consider HTE of $Z(\mathbf{K}, E)$ and LTE of $W(\mathbf{K}, E)$,
\begin{align}
    Z(\mathbf{K}, E) &\propto \sum_{C\in \vsdp} \tlx^{N_x(C)}\tly^{N_y(C)}\tlz^{N_z(C)}(-1)^{C\cdot \Omega E},\\
    W(\mathbf{K}, E) &\propto \sum_{C \in \vs} \lx^{N_x(C)}\ly^{N_y(C)}\lz^{N_z(C)}(-1)^{C\cdot \Omega E}.
\end{align}
In both cases, the proportionality constant is positive and independent of $E$. Thus, both lines are upper bounded by the $E=0$ case $Z(\mathbf{K})$.

\subsection{Connection to BPD duality}
To make a connection to the CSS case, consider the following setup: $\lambda_x\neq 0, \lambda_y = \lambda_z = 0$. This describes pure bit-flip decoherence. The matching of Boltzmann weights (Eq. \eqref{eq:noncssbwmapping}) for $R=2$ gives,
\begin{align}
    \frac{2\lx}{1+\lx^2} &= \left(\frac{1+\tlx-\tly-\tlz}{1+\tlx+\tly+\tlz}\right)^2 \nonumber\\
    0&= \left(\frac{1-\tlx+\tly-\tlz}{1+\tlx+\tly+\tlz}\right)^2 \\
    0&=\left(\frac{1-\tlx-\tly+\tlz}{1+\tlx+\tly+\tlz}\right)^2.\nonumber
\end{align}
Solving the above equation, we get the error rates of the dual theory at $R=2$,
\begin{align}
    &\tlx = 1, ~~~~\tly=\tlz, \nonumber\\
    &\frac{2\lx}{1+\lx^2} = \left(\frac{1-\tlz}{1+\tlz}\right)^2. \label{eq:bpdfromnoncss}
\end{align}
The dual decoherence channel $\mathcal{E} = \otimes_\mu \mathcal{E}_\mu$ factorizes into independent bit-flip and phase-flip errors,
\begin{align}
    \mathcal{E}_\mu(\rho_0) &= \mathcal{N}_{z,\mu} \circ \mathcal{N}_{x,\mu}(\rho_0), \nonumber \\
    \mathcal{N}_{x,\mu}(\rho_0) &=\frac{1}{2} \rho_0 + \frac{1}{2} X_\mu\rho_0 X_\mu, \quad \mathcal{N}_{z,\mu}(\rho_0)= (1-p'_z)\rho_0 + p'_z Z_\mu\rho_0 Z_\mu ,
\end{align}
where $p'_z$ is defined by $\tlz =  \frac{p'_z}{1-p'_z}$ (note $p'_z\neq \tilde{p}_z$). For CSS code, $\vs = \vx\oplus\vz$ and $\Tr(\rho^R)$ factorizes under independent bit-flip and phase-flip errors (Eq. \eqref{eq:bp_fractorized}), 
\begin{equation}
    \Tr (\rho^R) = \Tr (\rho_{\rm b}^R)  \times \Tr (\rho_{\rm p}^R).
\end{equation}
In this case, the bit-flip part is a constant factor while the phase-flip error part depends on the parameter $\tlz$ determined by Eq. \eqref{eq:bpdfromnoncss}. This relation is precisely the BPD duality for CSS code at $R=2$ (cf. Eq. \eqref{eq:dualr2}). The discussion for $R=\infty$ is analogous. 
\subsection{Self-dual surface}

The self-dual surface of $R=\infty$ is the set of points on $p_x+p_y+p_z=\frac{1}{2}$ since any points on this surface satisfies
\begin{equation}
    1-2p_y-2p_z = p_x/p_0, \label{eq:appselfdual}
\end{equation}
and similarly upon cyclic permutations of $p_x, p_y,p_z$. In terms of $K_x, K_y, K_z$, one can rediscover the duality in the statistical model: 
\begin{equation}
    e^{-2(K_x+K_y)} + e^{-2(K_y+K_z)} + e^{-2(K_x+K_z)} = 1. \label{eq:selfdualncssKs}
\end{equation}
To see that this defines the self-dual surface of the statistical model, one can apply Eq. \eqref{eq:selfdualncssKs} to get
\begin{align}
    e^{-2(K_y+K_z)} &= \frac{1+e^{-2(K_y+K_z)}-e^{-2(K_x+K_z)}-e^{-2(K_x+K_y)}}{1+e^{-2(K_y+K_z)} + e^{-2(K_x+K_z)} + e^{-2(K_x+K_y)}} \nonumber\\
    &= \frac{e^{K_x+K_y+K_z} + e^{K_x-K_y-K_z} - e^{K_y-K_x-K_z} - e^{K_z-K_x-K_y}}{e^{K_x+K_y+K_z}+e^{K_x-K_y-K_z}+e^{K_y-K_x-K_z} + e^{K_z-K_x-K_y}} \nonumber\\
    &= \frac{(e^{K_x}-e^{-K_x})(e^{K_y}+e^{-K_y})(e^{K_z}+e^{-K_z})+(e^{K_x}+e^{-K_x})(e^{K_y}-e^{-K_y})(e^{K_z}-e^{-K_z})}{(e^{K_x}+e^{-K_x})(e^{K_y}+e^{-K_y})(e^{K_z}+e^{-K_z})+(e^{K_x}-e^{-K_x})(e^{K_y}-e^{-K_y})(e^{K_z}-e^{-K_z})} \nonumber\\
    &= \frac{\tanh K_x + \tanh K_y \tanh K_z}{1+ \tanh K_x \tanh K_y \tanh K_z}.
\end{align}
The statistical model for $R=2$ is equivalent to a single copy of the statistical model with renormalized $K$'s, so we would expect a self-dual surface in that case as well. Compare Eqs. \eqref{eq:ncssr2sm}, \eqref{eq:ncssr2bw2} and LTE of $Z_{\rm SM}(\mathbf{K})$, one can identify, 
\begin{equation}
    (1-2p_y-2p_z)^2 = e^{-2(K_y+K_z)},
\end{equation}
and similarly for cyclic permutations of $x,y,z$. So the self-dual surface is determined by (cf. \eqref{eq:selfdualncssKs})
\begin{align}
    &(1-2p_y-2p_z)^2 + (1-2p_x-2p_z)^2 + (1-2p_x-2p_y)^2 = 1 \nonumber \\
    \Leftrightarrow & (1-p_x-p_y-p_z)^2 +p_x^2 + p_y^2 + p_z^2 = \frac{1}{2}.
\end{align}

\section{Parent Hamiltonian in the doubled Hilbert space}
\label{app:choi}
For a general stabilizer code $\cal C$, a stabilizer takes the general form
\begin{equation}
    A_J[X,Z] = (\i)^{\sfa_J \cdot \sfb_J} \prod_\mu (X_\mu)^{(\sfa_J)_\mu} \prod_\mu (Z_\mu)^{(\sfb_J)_\mu},
\end{equation}
where $a_J, b_J$ are $\mathbb{Z}_2$ vectors. Here, we choose to write the CSS code Hamiltonian as a sum of projectors so that the ground state has zero energy:
\begin{equation}
    H_s = \sum_J \frac{\mathds{1}-A_J}{2}. \label{eq:genericstabilizer}
\end{equation}
Notice that this Hamiltonian is related to the CSS code Hamiltonian in the main text by a factor of $\frac{1}{2}$ and some constant shift. 

The Choi-Jamio{\l}kowski isomorphism maps the system's density matrix into a state in the doubled Hilbert space. More specifically, we can first choose the basis $\{\ket{i}\}$ of the original Hilbert space given by the eigenstates of all the Pauli-$Z$ operators. The Choi-Jamio{\l}kowski isomorphism is specified by the mapping $\ket{i}\bra{j} \rightarrow \kket{ij}$, for all basis $i,j$. Therefore if we have a density matrix represented as $\sum_{ij}\rho_{ij}\ket{i}\bra{j}$, the Choi-Jamio{\l}kowski isomorphism map it to its Choi representation $\kket{\rho} = \sum_{ij}\rho_{ij}\kket{ij}$. 

One can multiply a density matrix with operators to modify it. It will, therefore, be crucial to spell out how to map these operations in the Choi representation. The most general operator action on the density matrix can be written as a sum of the following basic operation
\begin{equation}
    A\rho B = A_{ij}\rho_{jk}B_{kl}\ket{i}\bra{l}.
\end{equation}
Therefore, in the Choi representation, the resulting action is given by
\begin{equation}
    A_{ij}\rho_{jk}B_{kl}\kket{il} = A\otimes B^{\mathsf{T}} \kket{\rho}
\end{equation}
where the transpose is done in the eigenbasis of $Z$'s. For the pure-state density matrix $\rho_0$ of a logical state, it satisfies $H_s\rho_0 = 0$ and $\rho_0 H_s=0$. In the Choi representation, this statement implies that $\kket{\rho_0}$ is the ground state of the following Hamiltonian 
\begin{equation}
    H^D_0 = H_s\otimes \mathds{1} + \mathds{1}\otimes H_s^{\mathsf{T}}.
\end{equation}
For the error-corrupted density matrix, its Choi representation in the doubled Hilbert space takes the form
\begin{align*}
    \kket{\rho_{\rm m}} &\propto \left(e^{\sum_\mu \tilde{K}_x X_\mu\otimes X_\mu - \tilde{K}_y Y_\mu\otimes Y_\mu + \tilde{K}_z Z_\mu \otimes Z_\mu}\right) \kket{\rho_0} \\
    &\equiv \hat{\mathcal{E}} \kket{\rho_0},
\end{align*}
where we've defined the operator $\hat{\mathcal{E}}=e^{\sum_\mu \tilde{K}_x X_\mu\otimes X_\mu - \tilde{K}_y Y_\mu\otimes Y_\mu + \tilde{K}_z Z_\mu \otimes Z_\mu}$ acting on the doubled Hilbert space. One can construct a parent Hamiltonian for $\kket{\rho_{\rm m}}$ by noticing that $\kket{\rho_{\rm m}}$ is the ground state of the following frustration-free Hamiltonian as has been done in Ref. \onlinecite{LeeYouXu2022}
\begin{equation}
    {H^D}' = \sum_{J} \left[ \hat{\mathcal{E}}^{-1} \left(\frac{\mathds{1}-A_J}{2}\otimes\mathds{1}\right) \hat{\mathcal{E}}\right]\left[\hat{\mathcal{E}} \left(\frac{\mathds{1}-A_J}{2}\otimes \mathds{1}\right)\hat{\mathcal{E}}^{-1}\right] + \left[ \hat{\mathcal{E}}^{-1} \left(\mathds{1}\otimes\frac{\mathds{1}-A_J^{\mathsf{T}}}{2}\right) \hat{\mathcal{E}}\right]\left[\hat{\mathcal{E}} \left(\mathds{1}\otimes \frac{\mathds{1}-A_J^{\mathsf{T}}}{2}\right)\hat{\mathcal{E}}^{-1}\right].
\end{equation}
Note that every term in ${H^D}'$ is positive-semi-definite and annihilates the state $\kket{\rho_{\rm m}}$. Hence, $\kket{\rho_{\rm m}}$ must be a groundstate of ${H^D}'$.

Recall that with a general stabilizer $A_J[X,Z]$, we associate two binary vectors ${{\mathsf a}_{J,\mu}}$ and ${{\mathsf b}_{J,\mu}}$. We can then calculate the individual terms as
\begin{align*}
    &\hat{\mathcal{E}}\left(\frac{\mathds{1}-A_J}{2}\otimes \mathds{1}\right)\hat{\mathcal{E}}^{-1}\\
    =&\frac{1}{2} \mathds{1} \otimes \mathds{1} - \frac{1}{2} e^{2\sum_\mu \tK_x {{\mathsf b}_{J,\mu}}X_\mu\otimes X_\mu - \tK_y ({{\mathsf a}_{J,\mu}} + {{\mathsf b}_{J,\mu}})Y_\mu \otimes Y_\mu + \tK_z {{\mathsf a}_{J,\mu}} Z_\mu \otimes Z_\mu}A_J\otimes \mathds{1}.
\end{align*}
For later convenience, we introduce the following operator
\begin{equation}
    \hat{\mathcal{E}}_J = e^{\sum_\mu \tK_x {{\mathsf b}_{J,\mu}}X_\mu\otimes X_\mu - \tK_y ({{\mathsf a}_{J,\mu}} + {{\mathsf b}_{J,\mu}})Y_\mu \otimes Y_\mu + \tK_z {{\mathsf a}_{J,\mu}} Z_\mu \otimes Z_\mu}.
\end{equation}
We remark that $({{\mathsf a}_{J,\mu}} + {{\mathsf b}_{J,\mu}})$ should be understood as an addition in $\mathbb{Z}_2$. We thus have
\begin{align*}
    &\left[ \hat{\mathcal{E}}^{-1} \left(\frac{\mathds{1}-A_J}{2}\otimes \mathds{1}\right)\hat{\mathcal{E}}\right]\left[\hat{\mathcal{E}} \left(\frac{\mathds{1}-A_J}{2}\otimes \mathds{1}\right) \hat{\mathcal{E}}^{-1}\right]\\
    =& \left(\frac{1}{2} \mathds{1} \otimes \mathds{1} - \frac{1}{2}\hat{\mathcal{E}}_J^{-2} A_J\otimes \mathds{1}\right)\left(\frac{1}{2} \mathds{1} \otimes \mathds{1} - \frac{1}{2}A_J\otimes \mathds{1} \hat{\mathcal{E}}_J^{-2}\right) \\
    =& \frac{1}{4} \mathds{1} \otimes \mathds{1} + \frac{1}{4} \hat{\mathcal{E}}_J^{-4} - \frac{1}{4}\left(\hat{\mathcal{E}}_J^2 + \hat{\mathcal{E}}_J^{-2}\right) A_J \otimes \mathds{1}\\
\end{align*}
Thus, we can write
\begin{align}
    {H^D}' =  \sum_J  \frac{1}{2} \hat{\mathcal{E}}_J^{-4} - \frac{1}{4}\left(\hat{\mathcal{E}}_J^2 + \hat{\mathcal{E}}_J^{-2}\right) (A_J \otimes \mathds{1} + \mathds{1}\otimes A_J^{\mathsf T}).
\end{align}

A simpler parent Hamiltonian of $\kket{\rho_{\rm m}}$ can be derived by noting that the double-Hilbert space operator $\hat{\mathcal{E}}$ is Hermitian and factorizes site-wise. So, we can write down a frustration-free Hamiltonian of the following form
\begin{equation}
    H^{D} = \sum_J \hat{\mathcal{E}}_J^{-1} \left(\frac{\mathds{1}-A_J}{2}\otimes \mathds{1} + \mathds{1}\otimes \frac{\mathds{1}-A_J^{\mathsf T}}{2}\right)\hat{\mathcal{E}}_J^{-1}.
\end{equation}
Note that, for each $J$, $\hat{\mathcal{E}}_J$ is so constructed that all the terms in the exponent anti-commute with $A_J\otimes\mathds{1}$ and also $\mathds{1}\otimes A_J^{\mathsf T}$. One can then verify that
\begin{equation}
    H^{D} \hat{\mathcal{E}}\kket{\rho_0} = \sum_J \hat{\mathcal{E}}_J^{-1}\hat{\mathcal{E}}_{\Bar{J}}  \left(\frac{\mathds{1}-A_J}{2}\otimes \mathds{1} + \mathds{1}\otimes \frac{\mathds{1}-A_J^{\mathsf T}}{2}\right)\kket{\rho_0} = 0.
\end{equation}
where $\hat{\mathcal{E}}_{\Bar{J}}$ denotes all the terms in the original $\hat{\mathcal{E}}$ excluding the ones contained in $\hat{\mathcal{E}}_J$. $\hat{\mathcal{E}}_{\Bar{J}}$ commute with $\mathds{1}\otimes A_J^{\mathsf T}$ and $A_J \otimes \mathds{1}$. Since each individual term of $H^{D}$ is positive-semi-definite, $\hat{\mathcal{E}}\kket{\rho_0}=\kket{\rho_{\rm m}}$ must be its ground state. The explicit form of $H^{D}$ is obtained by moving $\hat{\mathcal{E}}_J$ across $\frac{\mathds{1}-A_J}{2}$,
\begin{align}
    H^{D} &= \sum_J \hat{\mathcal{E}}_J^{-2} - \frac{1}{2}A_J\otimes \mathds{1} - \frac{1}{2}\mathds{1}\otimes A_J^{\mathsf T} \nonumber\\
    &= H_s \otimes \mathds{1} + \mathds{1}\otimes H_s^{\mathsf T} + \sum_J \hat{\mathcal{E}}_J^{-2} + {\rm constant}.
\end{align}
When $\tK_x, \tK_y, \tK_z$ are small, ${H^D}$ to first order in $\tK$ is just $H^D_0$ plus the interaction between two copies of the Hilbert space $\sum_J \sum_\mu \tK_x {{\mathsf b}_{J,\mu}}X_\mu\otimes X_\mu - \tK_y ({{\mathsf a}_{J,\mu}} + {{\mathsf b}_{J,\mu}})Y_\mu \otimes Y_\mu + \tK_z{{\mathsf a}_{J,\mu}} Z_\mu \otimes Z_\mu$. This interaction competes with $H^D_0$. As we tune up the value of $\tK$'s, a possible scenario is that this competition leads to a quantum phase transition in $ H^{D}$. Whether there is indeed a quantum phase transition in this model needs to be examined independently for different stabilizer codes.

\twocolumngrid

\bibliography{ref.bib}

\end{document}